\documentclass[11pt]{article}
\usepackage{epsfig,amssymb,amsmath,amsthm,graphicx,mathtools,
xcolor,multirow,booktabs,hyperref}

\usepackage[sectionbib,longnamesfirst]{natbib}
\setcitestyle{authoryear}
\shortcites{Heaton2018,Fuglstad2019} 
\advance\topmargin by -60pt
\newcommand{\mcD}{\mathcal{D}}
\newcommand{\bfbeta}{\mbox{\boldmath $\beta$}}
\newcommand{\bftheta}{\mbox{\boldmath $\theta$}}
\newcommand{\bfvartheta}{\boldsymbol{\vartheta}}

\newcommand{\bfh}{\mbox{\boldmath $h$}}

\newcommand{\bfk}{\mbox{\boldmath $k$}}

\newcommand{\bfu}{\mbox{\boldmath $u$}}
\newcommand{\bfg}{\mbox{\boldmath $g$}}

\newcommand{\bfomega}{\mbox{\boldmath $\omega$}}

\newcommand{\mi}{\mathrm{i}} %% roman "i"
\newcommand{\compconj}[1]{\overline{#1}}

\numberwithin{equation}{section}
\newtheorem{lemma}{Lemma}
\newtheorem{thm}{Theorem}
\newtheorem{cor}{Corollary}%[section]

\newtheorem{proposition}{Proposition}

\begin{document}

\begin{center}
\large{\textbf{APPROXIMATE REFERENCE PRIORS FOR GAUSSIAN \\ RANDOM FIELDS}}

\vspace{1cm}

Victor De Oliveira \\
Department of Management Science and Statistics \\
The University of Texas at San Antonio, U.S.A. \\
{\tt victor.deoliveira@utsa.edu}

\vspace{0.5cm}

Zifei Han\footnote{Corresponding author.} \\
School of Statistics \\
University of International Business and Economics, China \\
{\tt zifeihan@uibe.edu.cn}

\vspace{0.3in}
 Jan 7, 2022 \\
(final revision)
\end{center}

\vspace{0.1in}

{\small
\begin{center}
\textbf{Abstract}
\end{center}

\noindent
Reference priors are 
theoretically attractive 
for the analysis of geostatistical data 
since they enable automatic Bayesian analysis and have desirable Bayesian and frequentist properties. 
But their use is hindered by computational hurdles that make their application in practice challenging. 
In this work, we derive 
a new class of default priors that approximate reference priors
for the parameters of some Gaussian random fields. 
It is based on an approximation to the integrated likelihood of the covariance parameters 
derived from the spectral approximation of stationary random fields. 
 This prior depends on the structure of the mean function and the 
spectral density 
of the model evaluated at a set of spectral points associated with an auxiliary regular grid.
In addition to preserving the desirable Bayesian and frequentist properties, 
these approximate reference priors are more stable, and their computations are  
much less onerous than those of exact reference priors. 
Unlike exact reference priors, the marginal approximate reference prior of correlation parameter 
is always proper, regardless of the mean function or the smoothness of the correlation function. 
This property has important consequences for covariance model selection.
An illustration comparing default Bayesian analyses
is provided with a data set of lead pollution in Galicia, Spain.
\vspace{0.2in}


\noindent
{\bf Running headline}: 
Default Prior for Gaussian Random Fields
}

\vspace{0.2in}

\newpage

\section{Introduction}  \label{sec:intro}

Random fields are ubiquitous for the modeling of spatial data in most natural and 
earth sciences. 
Among these, Gaussian random fields play a prominent role due to their versatility to 
model spatially varying phenomena, and because they serve as building blocks for the
construction of more elaborate models 
\citep{Zimmerman2010,Gelfand2016}.
When the main goal of the data analysis is spatial interpolation,
the Bayesian approach offers some advantages over the frequentist plug--in approach since
it accounts for parameter uncertainty.  
One of the challenges for implementing the Bayesian approach is the specification of 
sensible prior distributions for covariance parameters. 
The early works specified prior distributions in an ad--hoc manner
\citep{Kitanidis1986,Handcock1993,DeOliveira1997}, but these may yield unwanted results, 
including improper posteriors.
Sensible priors for covariance parameters must depend on the scale in the problem, 
for which little subjective information is usually available, and must also guarantee
posterior propriety.

A theoretically sound alternative to ad--hoc and subjective prior specifications 
consists of using information--based default priors, and among these 
reference priors have been the most studied. 
\cite{Berger2001} provided an extensive discussion on foundational issues involving
the formulation of default prior distributions, and initiated work on default (objective) 
Bayesian methods for the analysis of spatial data.
They advocated for the use of reference priors for Bayesian analysis of spatial data 
due to their theoretical guarantees and the empirically observed good frequentist properties 
of inferences based on these priors.
In particular, they showed that these priors overcome several drawbacks of previously 
proposed priors (e.g., they are guaranteed to be proper).
\cite{Berger2001} focused on Gaussian random fields with isotropic correlation functions 
depending on a single range parameter, and
extensions of this methodology have been developed for the analysis of more elaborate models.
\cite{Paulo2005} developed reference priors for separable correlation functions depending on
several range parameters, while \cite{DeOliveira2007} developed reference priors for
isotropic correlation functions with an unknown nugget parameter and a known range parameter. 
\cite{Kazianka2012} and \cite{Ren2012} both developed reference priors for
isotropic correlation functions with unknown range and nugget parameters, while
\cite{Kazianka2013} developed reference priors for geometrically anisotropic correlation functions.
\cite{Ren2013} considered more general mean functions and models with 
separable correlation functions, while \cite{Gu2018} 
established reference posterior propriety 
for separable correlation functions based on more general designs, 
and investigated robustness properties of inferences based on the resulting posteriors; 
\cite{DeOliveira2010} provided a review of Jeffreys and reference priors 
for geostatistical and lattice data models up to 2010.
The above works focus on the derivation of reference priors and the study of their properties,
either for the analysis of geostatistical data or computer emulation data. 
But their %practical 
implementation is hindered by computational challenges that render their use 
prohibitive in large data sets. 
As a result, in spite of their theoretically appealing properties, 
reference priors are seldom used in geostatistical  applications, even for the basic model studied in \cite{Berger2001}, although they have sometimes been used in computer emulation applications.
Computationally scalable approximations that retain the theoretical properties 
can be a better alternative.

In this work, we use the spectral approximation to stationary random fields to derive 
a new class of easy--to--compute default priors that approximate 
reference priors. 
Spectral approximations have been used for likelihood approximation,   
Bayesian inference, and model diagnostics by \cite{Royle2005}, 
\cite{Paciorek2007} and \cite{Bose2018}, among others. 
We use them here for default prior elicitation, but
unlike previous works, we do not assume the {\it sampling design} is regular.
Instead, we approximate the distribution of the random field at an  
{\it auxiliary regular design}, and use this to obtain a default prior for 
the model parameters using the reference prior algorithm.
By tuning the auxiliary design, we obtain a good approximation to the reference prior 
computed from the distribution of the random field at the sampling design. 
The computation and analysis of these approximate reference priors are considerably 
simpler, and their computations are more stable than those of exact reference priors. 
For models with a constant mean function, 
 the simplifications are even more substantial as the resulting approximate reference prior 
has a matrix--free expression.
In addition, for the model considered in this work, the approximate marginal 
reference prior of the correlation parameter is proper, 
regardless of the smoothness of the random field.
This is not the case for the exact marginal  reference prior, 
which has important consequences when using default priors for covariance function selection.
The resulting joint approximate reference posterior of all model parameters 
is proper as well.

The computation of the approximate reference prior relies on the spectral density function 
of the random field rather than on its covariance function.
The proposed methodology assumes the model has an explicit (or easy to compute) spectral density 
that is differentiable w.r.t. the correlation parameter, and has a general form that includes 
many families previously proposed in the literature.
Examples of such spectral densities include the isotropic Mat\'ern model \citep{Stein1999}, 
the model proposed in \cite{Laga2017}, and some of the isotropic models with 
rational spectral densities studied in \cite{Vecchia1985} and \cite{Jones1993}.
 The proposed methodology is illustrated using a data set of lead pollution in Galicia, Spain. 
Some details of theoretical and practical results are given in the Supplementary Materials.

\subsection{The Data and Random Field Model}  \label{subsec:data-models}

Geostatistical data consist of triplets 
$\{({\bf s}_i, \boldsymbol{f}({\bf s}_i), z_i) : i=1,\ldots,n\}$, 
where $\mathcal{S}_n = \{{\bf s}_1,\ldots,{\bf s}_n\}$ is a set of sampling locations in 
the region of interest $\mcD$, called the {\it sampling design}, 
$\boldsymbol{f}({\bf s}_i) = (f_1({\bf s}_i),\ldots, f_p({\bf s}_i))^{\top}$ is a
$p$--dimensional vector with covariates measured at ${\bf s}_i$ 
(usually $f_{1}({\bf s}) \equiv 1$), 
and $z_i \in {\mathbb R}$ is the measurement of the quantity of interest collected at ${\bf s}_i$.
The stochastic approach relies on viewing the set of measurements $\{z_i\}_{i=1}^{n}$ 
as a partial realization of a random field $Z(\cdot)$. 

Let $\{Z({\bf s}) : {\bf s} \in \mcD\}$ be a Gaussian random field with mean function 
$\mu({\bf s})$ and covariance function $C({\bf s},{\bfu})$, 
with $\mcD \subset {\mathbb R}^d$ and $d \geq 1$. 
 It is typically assumed that 
$\mu({\bf s}) =  \sum_{j=1}^{p} \beta_{j} f_{j}({\bf s})$, where
$\bfbeta = (\beta_{1},\ldots,\beta_{p})^{\top} \in {\mathbb R}^{p}$ are 
unknown regression parameters.
Additionally, $C({\bf s},{\bfu})$ is assumed isotropic and belonging to a parametric family, 
$\big\{C_{\boldsymbol{\theta}}({\bf s},{\bfu}) = 
\sigma^2 K_{\boldsymbol{\vartheta}}(||{\bf s} - {\bfu}||) 
: \bftheta = (\sigma^2, \mbox{\boldmath $\vartheta$}) \in (0,\infty) \times \Theta \big\}$, 
$\Theta \subset {\mathbb R}^q$, where $K_{\boldsymbol{\vartheta}}(\cdot)$ is an isotropic 
correlation function in ${\mathbb R}^d$ and $\| \cdot \|$ is the Euclidean norm. 
A widely used model is the Mat\'ern family with the parametrization proposed 
in \cite{Handcock1993} 
\begin{equation} 
\label{eq:matern-cov}
C_{\boldsymbol{\theta}}(r) = \frac{\sigma^2}{2^{\nu - 1} \Gamma(\nu)}
\left( \frac{2\sqrt{\nu}}{\vartheta} r \right)^{\nu} 
\mathcal{K}_{\nu}\left( \frac{2\sqrt{\nu}}{\vartheta} r \right) , \quad\quad r \geq 0 ,
\end{equation}
where $r$ is Euclidean distance,
$\sigma^2 > 0$, $\mbox{\boldmath $\vartheta$} = (\vartheta,\nu) \in (0,\infty)^2$, 
$\Gamma(\cdot)$ is the gamma function and  $\mathcal{K}_{\nu}(\cdot)$ is the modified Bessel function 
of second kind and order $\nu$ \citep{Abramowitz1964}.
It holds that $\sigma^2 = {\rm var}(Z({\bf s}))$, 
$\vartheta$ (mostly) controls how fast $C_{\boldsymbol{\theta}}(r)$ goes to zero when 
$r$ increases, and $\nu$ controls the degree of differentiability of
$C_{\boldsymbol{\theta}}(r)$ at $r = 0$. 
From these interpretations, $\sigma^2$ is called the {\it variance} parameter, 
$\vartheta$ the {\it range} parameter and $\nu$ the {\it smoothness} parameter. 

Sometimes in applications the measurements $z_i$ are corrupted by measurement error,
in which case they are modeled as 
$z_i = Z({\bf s}_i) + \epsilon_i$, for $i=1, \ldots, n$, where
$\epsilon_1,\ldots,\epsilon_n$ are i.i.d. with ${\rm N}(0,\tau^2)$ distribution and 
independent of $Z(\cdot)$; $\tau^2 \geq 0$ is called the {\it nugget} parameter. 
 The Mat\'ern family of covariance functions (\ref{eq:matern-cov}) is featured in this work 
as the primary example, but the proposed methodology also applies to other covariance families with
explicit spectral density functions.

\section{Reference Priors}  \label{sec:ref-priors}

In this work, we develop an approximate reference prior for the basic model studied in 
\cite{Berger2001} that assumes the data have no measurement error ($\tau^2 = 0$) and 
the correlation function depends on a single unknown range parameter; 
any other correlation parameter (e.g., $\nu$ in the Mat\'ern family) is assumed known,
a common assumption in geostatistical applications. 
Hence unless stated otherwise, for the remaining of the article we assume 
$\boldsymbol{\vartheta} = \vartheta$, a single range parameter, so $q=1$ and $\Theta = (0,\infty)$. 
This model provides the starting point to develop the proposed methodology.
Extensions to other models are currently being developed and will be considered elsewhere.

\subsection{Derivation}  \label{sec:derivations}
Below we briefly summarize the development of reference priors in models with 
regression parameters $\bfbeta$ and covariance parameters $\bftheta \;  = (\sigma^2,\vartheta)$.
It involves the following steps.
First, the parameters are classified as either of primary or secondary interest. 
The covariance parameters are typically considered of primary interest, and the regression parameters 
are of secondary interest. 
Second, the prior is factored accordingly as $\pi^{\rm R}(\bfbeta,\bftheta) 
= \pi^{\rm R}(\bfbeta ~|~ \bftheta) \pi^{\rm R}(\bftheta)$.
Third, the conditional Jeffreys prior of the secondary parameters given the primary parameters 
is computed, which for the current model is $\pi^{\rm R}(\bfbeta ~|~ \bftheta) \propto 1$.
Finally, $\pi^{\rm R}(\bftheta)$ is computed using the Jeffreys prior based 
%not on the original model, but 
on the `marginal model' defined via the integrated likelihood of $\bftheta$
\begin{eqnarray}
L^{\rm I}(\bftheta; \boldsymbol{z}) & = & \int_{\mathbb{R}^{p}} 
L(\bfbeta, \bftheta; \boldsymbol{z}) \pi^{\rm R}(\bfbeta ~|~ \bftheta) d \bfbeta \nonumber
\label{eq:intlik-org} \\
& \propto & (\sigma^{2})^{-\frac{n-p}{2}} 
| \boldsymbol{\Sigma}_{\vartheta}|^{-\frac{1}{2}} 
|\boldsymbol{X}^{\top} \boldsymbol{\Sigma}_{\vartheta}^{-1}
\boldsymbol{X}|^{- \frac{1}{2}}
\exp\Big\{ -\frac{S^{2}_{\vartheta}}{2 \sigma^2} \Big\} , 
\end{eqnarray}
where $L(\bfbeta, \bftheta; \boldsymbol{z})$ is the Gaussian likelihood 
of all model parameters based on the data $\boldsymbol{z} = (z_1,\ldots,z_n)^{\top}$, 
$S^{2}_{\vartheta} =  
(\boldsymbol{z} - \boldsymbol{X} \hat{\bfbeta}_{\vartheta})^{\top} 
\boldsymbol{\Sigma}_{\vartheta}^{-1} 
(\boldsymbol{z} - \boldsymbol{X} \hat{\bfbeta}_{\vartheta})$, 
$\hat{\bfbeta}_{\vartheta} =
(\boldsymbol{X}^{\top}\boldsymbol{\Sigma}^{-1}_{\vartheta} \boldsymbol{X})^{-1}
\boldsymbol{X}^{\top}\boldsymbol{\Sigma}^{-1}_{\vartheta} \boldsymbol{z}$, 
$\boldsymbol{X}$ is the known $n \times p$ design matrix with entries  
$\boldsymbol{X}_{ij} = f_{j}({\bf s}_i)$, and 
$\boldsymbol{\Sigma}_{\vartheta}$ is the $n \times n$ matrix with entries
$(\boldsymbol{\Sigma}_{\vartheta})_{ij} = 
K_{\vartheta}(\|{\bf s}_i - {\bf s}_j\|)$.

An alternative expression for the above integrated likelihood was used by 
\citet[Lemma 1]{DeOliveira2007} and \cite{Mure2021} to derive an alternative representation 
for reference priors. 
Let $\boldsymbol{W}$ be a full rank $n \times (n-p)$ matrix satisfying 
$\boldsymbol{W}^{\top}\boldsymbol{W} = \boldsymbol{I}_{n-p}$ and
$\boldsymbol{X}^{\top}\boldsymbol{W}$ equals to the $p \times (n-p)$ null matrix, so
the columns of $\boldsymbol{W}$ form an orthonormal basis of the orthogonal complement of the subspace of
$\mathbb R^n$ spanned by the columns of $\boldsymbol{X}$.
Then, it holds that
\begin{equation}
L^{\rm I}(\bftheta; \boldsymbol{z}) \ \propto \ (\sigma^{2})^{-\frac{n-p}{2}} 
|\boldsymbol{W}^{\top} \boldsymbol{\Sigma}_{\vartheta}\boldsymbol{W}|^{-\frac{1}{2}} \;
\exp\left\{ -\frac{\boldsymbol{z}^{\top} 
\boldsymbol{W} \big(\boldsymbol{W}^{\top} \boldsymbol{\Sigma}_{\vartheta}\boldsymbol{W}
\big)^{-1} \boldsymbol{W}^{\top} \boldsymbol{z}}{2 \sigma^2} \right\} .
\label{eq:intlik-var}
\end{equation}
The matrix $\boldsymbol{W}$ always exists, but is not unique. 
One such matrix can be computed from the singular value decomposition of the design matrix,
namely $\boldsymbol{X} = \boldsymbol{U}\boldsymbol{S}\boldsymbol{V}^{\top}$ with
$\boldsymbol{U}$ and $\boldsymbol{V}$ orthogonal matrices of sizes $n \times n$ and $p \times p$, 
respectively, and $\boldsymbol{S}$ an $n \times p$ matrix whose only non--null entries are
on the main diagonal.
Taking $\boldsymbol{W}$ as the last $n-p$ columns of $\boldsymbol{U}$ satisfies 
the requirements \citep{Mure2021}.

\begin{proposition}[Reference Prior]  \label{thm:ref-prior}
The reference prior of $(\bfbeta, \sigma^2, \vartheta)$ is given by 
\begin{equation}
\pi^{\rm R}(\bfbeta, \sigma^2, \vartheta) \propto \frac{\pi^{\rm R}(\vartheta)}{\sigma^2} ,
\label{ref-prior}
\end{equation}
where $\pi^{\rm R}(\vartheta)$ admits the following representations: 

(a)
\begin{equation}
\pi^{\rm R}(\vartheta) \propto 
\left\{ {\rm tr}\left[\left\{ \left(\frac{\partial}{\partial \vartheta} \boldsymbol{\Sigma}_{\vartheta}\right) \boldsymbol{Q}_{\vartheta}\right\}^2 \right] - \frac{1}{n-p} 
\left[{\rm tr}\left\{\left(\frac{\partial}{\partial \vartheta} \boldsymbol{\Sigma}_{\vartheta}
\right) \boldsymbol{Q}_{\vartheta}\right\} \right]^{2} \right\}^{\frac{1}{2}},
\label{eq:reference-prior-1}
\end{equation}
where $\boldsymbol{Q}_{\vartheta} := \boldsymbol{\Sigma}_{\vartheta}^{-1} - 
\boldsymbol{\Sigma}_{\vartheta}^{-1}\boldsymbol{X} (\boldsymbol{X}^{\top}\boldsymbol{\Sigma}_{\vartheta}^{-1}\boldsymbol{X})^{-1} \boldsymbol{X}^{\top}\boldsymbol{\Sigma}_{\vartheta}^{-1}$. 

(b)	
\begin{equation}
\pi^{\rm R}(\vartheta) \propto 
\left\{ {\rm tr}\left[\left\{ \left(\frac{\partial}{\partial \vartheta} \boldsymbol{\Sigma}^{W}_{\vartheta}\right) \left(\boldsymbol{\Sigma}^{W}_{\vartheta} \right)^{-1}
\right\}^2 \right] - \frac{1}{n-p} 
\left[{\rm tr}\left\{\left(\frac{\partial}{\partial \vartheta} \boldsymbol{\Sigma}^{W}_{\vartheta}
\right)  \left(\boldsymbol{\Sigma}^{W}_{\vartheta} \right)^{-1} \right\} \right]^{2} \right\}^{\frac{1}{2}} ,
\label{eq:reference-prior-2}
\end{equation}
where $\boldsymbol{\Sigma}^{W}_{\vartheta} \coloneqq 
\boldsymbol{W}^{\top} \boldsymbol{\Sigma}_{\vartheta}\boldsymbol{W}$.  
\end{proposition}
The proof of (\ref{eq:reference-prior-1}) was given in 
\citet{Berger2001} which is based on (\ref{eq:intlik-org}),
while the proof of (\ref{eq:reference-prior-2}) was given in 
\citet[Proposition 2.1]{Mure2021}, which is based on (\ref{eq:intlik-var}).

The propriety of the reference posterior distribution derived from the reference prior (\ref{ref-prior}) 
requires that the integral
\begin{equation}
\int_{{\mathbb R}^p \times (0,\infty)^2} L(\bfbeta, \sigma^2 , \vartheta ; \boldsymbol{z}) 
\frac{\pi^{\rm R}(\vartheta)}{\sigma^2} d \bfbeta d\sigma^2 d\vartheta
	\; = \; \int_{(0,\infty)} L^{\rm I}(\vartheta ; \boldsymbol{z}) \pi^{\rm R}(\vartheta) d\vartheta ,
\label{integrability}
\end{equation}
is finite, where $L^{\rm I}(\vartheta ; \boldsymbol{z})$ is the so--called 
{\it integrated likelihood} of $\vartheta$ obtained by integrating the product of  
the likelihood and $1/\sigma^2$  ($= \pi^{\rm R}(\sigma^2 ~|~\vartheta)$)
over $\bfbeta$ and $\sigma^2$. It is given by the following result.

\begin{proposition}[Integrated Likelihood]   \label{thm:prop-intlik}
The integrated likelihood of $\vartheta$ admits the following representations: 

(a)
\begin{equation}
\label{eq:intlik-1}
L^{\rm I}(\vartheta; \boldsymbol{z}) \propto 
|\boldsymbol{\Sigma}_{\vartheta}|^{-\frac{1}{2}} 
|\boldsymbol{X}^{\top}\boldsymbol{\Sigma}_{\vartheta}^{-1}\boldsymbol{X}|^{- \frac{1}{2}}
(S^{2}_{\vartheta})^{-\frac{n-p}{2}} . 
\end{equation}

(b)	
\begin{equation}
\label{eq:intlik-2}
L^{\rm I}(\vartheta; \boldsymbol{z}) \propto 
|\boldsymbol{\Sigma}^{W}_{\vartheta}|^{-\frac{1}{2}}
\left( \big(\boldsymbol{z}^{W}\big)^{\top} 
\big( \boldsymbol{\Sigma}^{W}_{\vartheta}\big)^{-1} \big(\boldsymbol{z}^{W}\big)
\right)^{-\frac{n-p}{2}}, 
\end{equation}
where $\boldsymbol{z}^{W} \coloneqq \boldsymbol{W}^{\top} \boldsymbol{z}$. 
\end{proposition}
The expression (\ref{eq:intlik-1}) follows by direct calculation \citep{Berger2001}, 
while the proof of (\ref{eq:intlik-2}) is given in 
\citet[Proposition 2.2]{Mure2021}. 
\cite{Berger2001} stated results on the asymptotic behaviour of $\pi^{\rm R}(\vartheta)$ 
and $L^{\rm I}(\vartheta; \boldsymbol{z})$ as 
$\vartheta \rightarrow 0^+$ and $\vartheta \rightarrow \infty$,
and established results on the propriety of reference posterior distributions for the 
parameters of many families of isotropic covariance functions, including the Mat\'ern family.

\medskip

{\bf Discussion}.
\cite{Mure2021} recently noticed that a key technical assumption in the proofs of
auxiliary results in \citet[Lemmas 1 and 2]{Berger2001} that were used to prove
reference posterior propriety does not hold for smooth families of isotropic correlations 
functions, namely those that are twice continuously differentiable at the origin.
Specifically, it was assumed that the correlation matrix of the data can be expressed as
\begin{equation}
\boldsymbol{\Sigma}_{\vartheta} = 
{\bf 1}_{n}{\bf 1}_{n}^{\top} + q(\vartheta)(\boldsymbol{D} +o(1)) , \quad\quad 
{\rm as} \ \ \vartheta \rightarrow \infty ,
\label{sigma-expansion}
\end{equation}
where ${\bf 1}_{n}$ is the vector of ones, $q(\vartheta)$ is a continuous function
satisfying $\lim_{\vartheta \rightarrow \infty}q(\vartheta) = 0$, and
$\boldsymbol{D}$ is a fixed $n \times n$ {\it non--singular} matrix.
This property of $\boldsymbol{D}$ is used extensively in 
\cite{Berger2001}.
The identity (\ref{sigma-expansion}) follows from the Maclaurin expansion of the 
correlation function $K_{\vartheta}(\cdot)$, where $\boldsymbol{D}$ has entries 
$\boldsymbol{D}_{ij} = \|{\bf s}_i - {\bf s}_j \|^{b}$, 
for some $b > 0$ that depends on $K_{\vartheta}(\cdot)$.
 For non--smooth families of correlation functions (e.g., Mat\'ern families with $\nu < 1$),
$b < 2$ and in this case it indeed holds that $\boldsymbol{D}$ is non--singular \citep{Schoenberg1937};
see \citet[Proposition 3.1 and Table 2]{Mure2021} and \citet[Appendix C]{Berger2001}.
On the other hand, based on a result by \cite{Gower1985}, \cite{Mure2021} showed that
 for smooth families of correlation functions (e.g., Mat\'ern families with $\nu \geq 1$),
$b = 2$ and in this case the rank of $\boldsymbol{D}$ is at most $d+2$.
So the key technical assumption (\ref{sigma-expansion}), with $\boldsymbol{D}$ non--singular,
does not hold when $n > d+2$, which is always the case in geostatistical applications where $d=2$; 
see \citet[Proposition 3.4]{Mure2021}. 
Therefore, the proof of propriety of the reference posterior given in \cite{Berger2001} 
is valid only for non--smooth families of isotropic correlation functions.
Nevertheless, the propriety result still holds more generally
since \citet[Theorem 4.4]{Mure2021} provided a proof for the case of smooth families of 
isotropic correlation functions that do not require $\boldsymbol{D}$ to be non--singular. 
It only requires that the data $\boldsymbol{z}$ do not belong to a certain  hyperplane of 
${\mathbb R}^n$, an assumption that holds with probability one under any of the 
considered models.

As noted in \cite{Mure2021}, the above findings are not limited to the model considered 
in this work, but also have strong bearings on many other spatial models.
With the exception of \cite{DeOliveira2007}, who assumed the range parameter is known, 
virtually all articles that have obtained reference posterior propriety results for other 
stationary covariance functions have used arguments that rely on (\ref{sigma-expansion}) 
or similar assumptions, with $\boldsymbol{D}$ non--singular.
As a result, their proofs are also incomplete and in need of completion for
smooth families of correlation functions.

\subsection{Bayesian and Frequentist Properties}  \label{bayes-frequentist-properties}

\cite{Berger2001} noticed that several ad--hoc automatic priors 
(that do not require subjective elicitation) proposed up to that time yielded improper 
posterior distributions.
Reference priors are also automatic, but it was shown for all models studied in the works 
listed in the Introduction that they yield proper posterior distributions.
Additionally, for several of these models, the marginal reference prior of the 
correlation parameters are proper, which allows the use of Bayes factors for selecting 
the smoothness of the covariance family  \citep{Berger2001}. 
This is a helpful property since the smoothness is often arbitrarily chosen, and few methods
are available for this purpose.

It has also been found that statistical inferences based on reference priors have good 
frequentist properties. For different stationary covariance models, 
\cite{Berger2001}, \cite{Paulo2005}, \cite{Kazianka2012}, \cite{Ren2012} and \cite{Ren2013} carried out
simulation studies showing that, when viewed as confidence intervals, credible intervals 
for range parameters based on reference priors have reasonably good frequentist coverage.
Additional evidence is provided in the Supplementary Materials.
It has also been empirically found that profile likelihoods of range parameters are very flat
for some geostatistical and computer emulation data, and in this case maximum likelihood estimates (MLE) 
tend to be either close to zero (negligible correlation) or very large (unrealistic high correlation).
These behaviours have deleterious effects on the performance of plug--in (kriging) predictors. 
\cite{Gu2018} and \cite{Gu2019} showed that in these cases, the marginal reference posteriors of 
range parameters are better behaved, and the mean square errors of plug--in predictors
based on maximum a posteriori estimates were smaller than those of MLE,
for several parametrizations and simulation scenarios.

\subsection{Practical Limitations}

In spite of their good theoretical properties, reference priors are seldom used 
in geostatistical applications due to several computational challenges. 
First, the evaluation of $\pi^{\rm R}(\vartheta)$ in (\ref{eq:reference-prior-1}) requires 
the computation of the $n \times n$ matrix $\boldsymbol{\Sigma}_{\vartheta}^{-1}$ 
and $p \times p$ matrix 
$(\boldsymbol{X}^{\top}\boldsymbol{\Sigma}_{\vartheta}^{-1}\boldsymbol{X})^{-1}$,
which require $O(n^3)$ operations. 
The evaluation of $\pi^{\rm R}(\vartheta)$ in (\ref{eq:reference-prior-2}) requires 
the computation of $\boldsymbol{W}$ (only once) and the $(n-p) \times (n-p)$ matrix 
$(\boldsymbol{\Sigma}^{W}_{\vartheta})^{-1}$, which require $O((n-p)^3)$ operations.
Second, except for certain degrees of smoothness in the Mat\'ern family, 
computation of $(\partial /{\partial \vartheta}) \boldsymbol{\Sigma}_{\vartheta}$
involves the evaluation of $\mathcal{K}_{\nu}(x)$ and its derivative w.r.t. $x$, 
which is given by $(\partial / \partial x) \mathcal{K}_{\nu}(x) = 
-(\mathcal{K}_{\nu-1}(x) + \mathcal{K}_{\nu+1}(x))/2$  \citep{Abramowitz1964}, so  
$O(n^2)$ evaluations of this Bessel function are needed. 
The same would hold for other families of correlations that involve special functions.
Third, for many families of correlation functions, including the Mat\'ern, 
the matrix $\boldsymbol{\Sigma}_{\vartheta}$ is often nearly singular when either 
$\vartheta$ or $\nu$ are large, so the computation of $\boldsymbol{\Sigma}_{\vartheta}^{-1}$ 
will be unstable or infeasible. 
The same holds for the computation of $(\boldsymbol{\Sigma}^{W}_{\vartheta})^{-1}$ since its
condition number is smaller than that of $\boldsymbol{\Sigma}_{\vartheta}^{-1}$
\citep{Dietrich1994}. 
All of these make the computation of $\pi^{\rm R}(\vartheta)$ either computationally expensive, 
unstable or infeasible, even for geostatistical data sets of moderate size.

We circumvent these challenges by deriving an approximate reference prior
that is more amenable for analysis and computation.
It relies on an approximation to the integrated likelihood of the covariance parameters
that is computed from the spectral representation of the stationary random fields. 
This approximation depends neither on $K_{\vartheta}(\cdot)$ nor on the inverse of the 
large and possibly numerically singular matrix $\boldsymbol{\Sigma}_{\vartheta}$, 
but instead on the spectral density function of the model.

\section{Spectral Approximation to the Integrated Likelihood}  \label{sec:spectral-appr}  

\subsection{Spectral Approximation}  \label{spectral-approximation}

The starting step to obtain a convenient approximation to the integrated likelihood in
(\ref{eq:intlik-org}) (or (\ref{eq:intlik-var})) is the spectral representation of 
stationary random fields.
Such approximation has been described and used for different purposes by 
\cite{Royle2005}, \cite{Paciorek2007} and \cite{Bose2018}. 
Unlike these works, this device is employed here to approximate the random field over 
a set of locations that may or may not be the sampling design. 
Although the basic tenets to construct the approximation are the same regardless of dimension, 
we describe in detail the approximation for random fields in the plane ($d=2$), 
the common scenario where geostatistical data arise.
Let $\mu({\bf s})$ and $\sigma^2 K_{\vartheta}(r)$ be, respectively, the mean and covariance 
functions of the random field $Z(\cdot)$, and $\sigma^2 f_{\vartheta}(\bfomega)$ its 
spectral density function.
For instance, for the Mat\'ern family in (\ref{eq:matern-cov}) we have 
\citep{Stein1999}
\begin{equation}
\label{eq:matern-spden}
\vspace{-0.1cm}	
f_{\vartheta}(\bfomega) = \frac{\Gamma(\nu + 1) (4\nu)^\nu}{\pi \Gamma(\nu) \vartheta^{2\nu}} 
\Big(\|\bfomega\|^2 + \frac{4\nu}{\vartheta^2}\Big)^{-(\nu + 1)},
\quad \bfomega = (\omega_1,\omega_2)^{\top} \in {\mathbb R}^2. 
\end{equation}
In the expression above and in what follows,
$\bfomega$ denotes `angular frequency', as commonly used in statistics, 
rather than `frequency', as used by \cite{Bose2018}.

Let $M_1$, $M_2$ be two positive even integers, $\Delta > 0$ and 
$\mathcal{U}_M = \{{\bfu}_{1,1},{\bfu}_{1,2},\ldots,{\bfu}_{M_1, M_2}\} = 
\{\Delta, \ldots,\Delta M_1\} \times \{\Delta,\ldots,\Delta M_2\}$, with $M \coloneqq M_1 M_2$,
be a set of spatial locations forming a regular rectangular grid in the plane.
The set $\mathcal{U}_M$ does not need to be the sampling design $\mathcal{S}_n$, but
is constructed in a way so that contains the convex hull of the region of interest $\mathcal{D}$. 
Associated with $\mathcal{U}_M$ we define a corresponding set of $M$ spatial frequencies 
(spectral points), also forming a regular rectangular grid in the plane, as
\begin{align*}
\mathcal{W}_{M} &= \left\{ \bfomega_{-\frac{M_1}{2} + 1, -\frac{M_2}{2} + 1}, \ldots, \bfomega_{0,0}, 
\ldots, \bfomega_{\frac{M_1}{2}, \frac{M_2}{2}} \right\}  \\ 
&=\frac{2\pi}{\Delta M_1} \left\{ -\frac{M_1}{2} + 1, \ldots, 0, 1, \ldots, \frac{M_1}{2} \right\} 
\; \times \; \frac{2\pi}{\Delta M_2} \left\{ -\frac{M_2}{2} + 1, \ldots, 0, 1, \ldots, \frac{M_2}{2} \right\} , 
\end{align*}
where $\mathcal{W}_{M} \subset [-\frac{\pi}{\Delta} , \frac{\pi}{\Delta}]^2$
is called the {\it spectral design}; Figure \ref{fi:spec-design-group} provides an example. 

Now, let $Z_{\Delta}({\bfk}) \coloneqq Z(\Delta {\bfk})$, ${\bfk} = (k_1,k_2)^{\top} \in {\mathbb Z}^2$, 
be the discrete index random field defined by sampling the random field $Z(\cdot)$ at the rate $\Delta$.
This random field has mean function $\mu(\Delta {\bfk})$ and covariance function
$\sigma^2 K_{\vartheta}(\Delta||{\bfk} - {\bfk}'||)$, for ${\bfk}, {\bfk}' \in {\mathbb Z}^2$,
while its spectral density function is given by \citep{Yaglom1987}
\begin{equation}
f^{\Delta}_{\vartheta}(\bfomega) = 	
\sum_{\boldsymbol{l} \in {\mathbb Z}^2} 
f_{\vartheta}\Big(\bfomega + \frac{2\pi}{\Delta}\boldsymbol{l}\Big) ,
\quad \quad \bfomega \in \Big[-\frac{\pi}{\Delta},\frac{\pi}{\Delta}\Big]^2 .
\label{eq:spec-den-alias}
\end{equation}
Note that \cite{Paciorek2007} and \cite{Bose2018} 
assumed %incorrectly that  
$f^{\Delta}_{\vartheta}(\bfomega) = f_{\vartheta}(\bfomega)$,
effectively ignoring the aliasing effect. 
The spectral representation of stationary random fields in ${\mathbb Z}^2$
states that for any ${\bfu}_{i,j} = \Delta (i,j)^{\top} \in \mathcal{U}_M$ it holds that
\begin{equation*}	
Z({\bfu}_{i,j}) \ = \ Z_{\Delta}\big((i,j)^{\top}\big) \ = \ \mu({\bfu}_{i,j}) + 
\int_{-\frac{\pi}{\Delta}}^{\frac{\pi}{\Delta}} \int_{-\frac{\pi}{\Delta}}^{\frac{\pi}{\Delta}} 
\exp({\rm i} \bfomega^{\top}{\bfu}_{i,j}) U_{\Delta}(d \bfomega)  ,	
\label{eq:spectral-approx}
\end{equation*}
where ${\rm i} = \sqrt{-1}$ and $U_{\Delta}(\cdot)$ is a complex zero--mean random 
orthogonal measure in the plane \citep{Yaglom1987}.

\begin{figure}[t]
\begin{center}
\psfig{figure=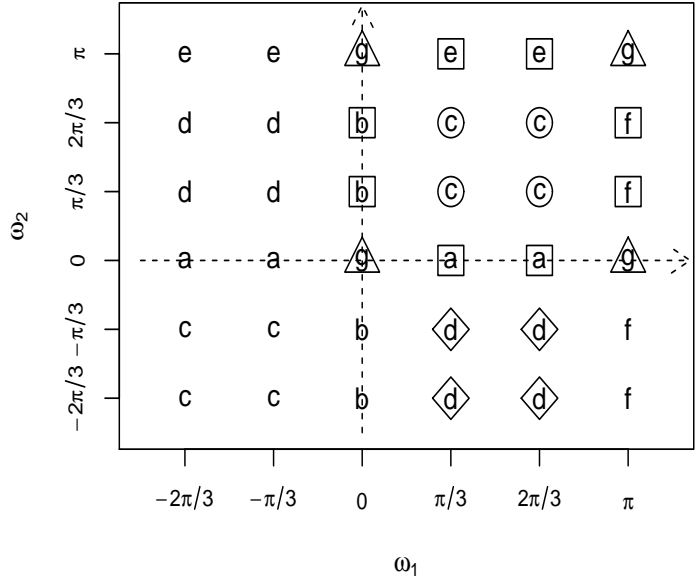, width= 10 cm,height= 10 cm}
\end{center}
\vspace{-0.8 cm}
\caption{Spectral design $\mathcal{W}_M$ when $M_1 = M_2 = 6$ and $\Delta = 1$.
The points $\bfomega = (\omega_1,\omega_2)^\top$ are identified by the letters a--g.
The enclosed letters are the points determined by the indices in $I_{\rm C}, I_{\rm B}, I_{\rm I}$ 
and $I_{\rm E}$ (see main text for the precise mapping).}
\label{fi:spec-design-group}
\end{figure}

This representation motivates the following lemma that provides an approximation to the 
distribution of $\big(Z({\bfu}_{i,j}) : {\bfu}_{i,j} \in \mathcal{U}_M\big)^{\top}$.
Before stating the result, a random object satisfying some assumptions needs to be defined. 
Consider the following sets of indices that determine subsets of the spectral design $\mathcal{W}_{M}$:
\begin{align*}
I_{\rm C} & \coloneqq \Big\{\big(0, 0\big), \Big(\frac{M_1}{2}, 0\Big), \Big(0, \frac{M_2}{2}\Big), \Big(\frac{M_1}{2}, 
\frac{M_2}{2}\Big)\Big\} \hspace{3.5cm} \mbox{(`corner' frequencies)} \\
I_{\rm B} & \coloneqq \Big\{\big(m_1, 0\big), \big(0, m_2\big), \Big(m_1, \frac{M_2}{2}\Big), \Big(\frac{M_1}{2}, m_2\Big) :
m_1 = 1,\ldots,\frac{M_1}{2}-1 ; m_2 = 1,\ldots, \frac{M_2}{2}-1\Big\} \\
& \hspace{10.8cm}  \mbox{(`boundary' frequencies)} \\
I_{\rm I} & \coloneqq \Big\{\big(m_1, m_2\big) : m_1 = 1,\ldots,\frac{M_1}{2}-1 ; m_2 = 1,\ldots,\frac{M_2}{2}-1 \Big\} 
\quad\quad \ \ \mbox{(`interior' frequencies)} \\
I_{\rm E} & \coloneqq \Big\{\big(m_1, m_2\big) : m_1 = 1,\ldots,\frac{M_1}{2}-1 ; m_2 = -\frac{M_2}{2}+1,\ldots, -1\Big\} 
\ \ \ \  \mbox{(`exterior' frequencies)}.
\end{align*}
Also, let $I \coloneqq I_{\rm B} \cup I_{\rm I} \cup I_{\rm E}$ which has $M/2 - 2$ elements.
The labels `corner', `boundary' and `interior' refer to spectral points in the first quadrant 
of the plane, while the label `exterior' refers to spectral points in the fourth quadrant.
For instance, the locations of the letters a--g in Figure \ref{fi:spec-design-group} represent the spectral points 
in $\mathcal{W}_M$ when $M_1 = M_2 = 6$ and $\Delta = 1$.
The indices in $I_{\rm C}$ determine the spectral points in the figure enclosed by triangles.
Likewise, the indices in $I_{\rm B}, I_{\rm I}$ and $I_{\rm E}$ determine the spectral points 
enclosed by, respectively, squares, circles, and diamonds.
(The matching of spectral points by letters, e.g., the points in the first and third quadrants 
labeled as `c', will be used to motivate the proof of Lemma \ref{thm:spec-approx-2d} below).

For any $M_1, M_2$ positive even integers, $\Delta > 0$ and 
$\bfu_{i,j} \in \mathcal{U}_M$ define the random object
\[
T_{M_1, M_2}(\bfu_{i,j}) \coloneqq  
\sum_{m_1 = -\frac{M_1}{2}+1}^{\frac{M_1}{2}} \sum_{m_2 = -\frac{M_2}{2}+1}^{\frac{M_2}{2}}
\exp({\rm i} \bfomega_{m_1,m_2}^{\top} \bfu_{i,j}) U_{m_1,m_2} ,
\]
where $U_{m_1, m_2} =  A_{m_1, m_2} + \, {\rm i} B_{m_1, m_2}$ are complex random variables. 
The real and imaginary parts satisfy the following assumptions, collectively denoted (A0):
\begin{itemize}
\item[(1)] $B_{m_1, m_2} = 0$ for $(m_1,m_2) \in I_{\rm C}$ 

\item[(2)] $U_{0, -m_2} = \bar{U}_{0, m_2}$, $U_{-m_1, 0} = \bar{U}_{m_1, 0}$, 
$U_{\frac{M_1}{2}, -m_2} = \; \bar{U}_{\frac{M_1}{2}, m_2}$ and
$U_{-m_1, \frac{M_2}{2}} = \; \bar{U}_{m_1, \frac{M_2}{2}}$ 
%\mbox{\hspace{1.1cm} 
for $(m_1,m_2) \in I_{\rm B}$

\item[(3)] $U_{-m_1, -m_2} = \bar{U}_{m_1, m_2}$ for $(m_1,m_2) \in I_{\rm I} \cup I_{\rm E}$

\item[(4)] For $(m_1,m_2) \in I_{\rm C} \cup I$, $A_{m_1, m_2}$ and $B_{m_1, m_2}$ are 
independent Gaussian variables with means $0$ and variances
\begin{align*}
{\rm var}(A_{m_1, m_2} ) &= 
\frac{c_{\Delta} \sigma^2}{M}  f^{\Delta}_{\vartheta}(\bfomega_{m_1, m_2})  
\quad\quad	{\rm if} \ (m_1, m_2) \in I_{\rm C} \\
{\rm var}(A_{m_1, m_2} ) = {\rm var}(B_{m_1, m_2}) &=
\frac{c_{\Delta} \sigma^2}{2M} f^{\Delta}_{\vartheta}(\bfomega_{m_1, m_2})
\quad\quad {\rm if} \ (m_1, m_2) \in I ,
\end{align*}
where $c_{\Delta} \coloneqq ({2\pi}/{\Delta})^2$ and $f^{\Delta}_{\boldsymbol{\vartheta}}(\bfomega)$
is given in (\ref{eq:spec-den-alias}).
\end{itemize}

\begin{lemma}[Spectral Approximation]  \label{thm:spec-approx-2d}
Consider the random object $T_{M_1, M_2}(\bfu_{i,j})$ defined above satisfying assumption (A0). 
Then, for any $\bfu_{i,j} \in \mathcal{U}_M$

\medskip

(a) $T_{M_1, M_2}(\bfu_{i,j})$
\begin{align}
\nonumber	
&= A_{0,0} + A_{\frac{M_1}{2},0} \cos\big(\bfomega_{\frac{M_1}{2},0}^{\top} {\bfu}_{i,j}\big)
+ A_{0,\frac{M_2}{2}} \cos\big(\bfomega_{0, \frac{M_2}{2}}^{\top} {\bfu}_{i,j}\big)
+ A_{\frac{M_1}{2},\frac{M_2}{2}} \cos\big(\bfomega_{\frac{M_1}{2}, \frac{M_2}{2}}^{\top} {\bfu}_{i,j}\big) \\
& \quad + \ 2 \sum_{(m_1,m_2) \in I} 
\big( A_{m_1,m_2} \cos(\bfomega_{m_1, m_2}^{\top} {\bfu}_{i,j}) 
- B_{m_1,m_2} \sin(\bfomega_{m_1, m_2}^{\top} {\bfu}_{i,j}) \big) ,
\label{eq:spectral-approx2}
\end{align}
and $(T_{M_1, M_2}(\bfu_{i,j}) : \bfu_{i,j} \in \mathcal{U}_M)^{\top}$ has a zero--mean 
real multivariate normal distribution.

\medskip

(b) For any $\bfu_{i,j}, \bfu_{i',j'} \in \mathcal{U}_M$ it holds that as
 $\min\{M_1,M_2\} \rightarrow \infty$
\[
{\rm cov}\{T_{M_1, M_2}(\bfu_{i,j}), T_{M_1, M_2}(\bfu_{i',j'})\} \rightarrow 
\sigma^2 K_{\vartheta}(||\bfu_{i,j} - \bfu_{i',j'}||) .
\]

\medskip

\noindent
Proof. See the Supplementary Materials.
\end{lemma}

\noindent
From part (a) it follows that 
$\boldsymbol{t} \coloneqq (T_{M_1, M_2}(\bfu_{1,1}),\ldots,T_{M_1, M_2}(\bfu_{M_1,M_2}))^{\top}
= \boldsymbol{H}_1 {\bfg}$ has a joint multivariate normal distribution, where 
$\boldsymbol{H}_1 \coloneqq ({\bf 1}_M,\boldsymbol{H})$, $\boldsymbol{H}$ is the $M \times (M-1)$ matrix 
whose columns are formed by the multiples $1$, $2$ or $-2$ of either cosines or sines evaluated 
at the inner products of appropriate frequencies and locations
(see the Supplementary Materials for details), and 
${\bfg}$ is the $M \times 1$ vector that stacks the variables $A_{m_1, m_2}$ and $B_{m_1, m_2}$ 
appearing in (\ref{eq:spectral-approx2}) in the order
\begin{equation}
{\bfg} = \Big( \big(A_{m_1,m_2} : (m_1,m_2) \in I_{\rm C}\big)^{\top}, 
\big(A_{m_1,m_2} : (m_1,m_2) \in I\big)^{\top}, \big(B_{m_1,m_2} : (m_1,m_2) \in I\big)^{\top} \Big)^{\top}.
\label{g-vector}
\end{equation}
In addition, if $\tilde{\boldsymbol{z}} \coloneqq (Z({\bfu}_{1,1}),\ldots,Z({\bfu}_{M_1, M_2}))^{\top}$,
then from part (b) we have that when $M_1$ and $M_2$ are both large,
$\tilde{\boldsymbol{z}} \; \stackrel{d}{\approx} \; \tilde{\boldsymbol{X}}\bfbeta + \boldsymbol{t}$,
where  $\tilde{\boldsymbol{X}}$ is the $M \times p$ matrix whose entries involve the covariates 
measured at the locations in $\mathcal{U}_M$
(it is assumed the covariates are available at any location, a common situation in geostatistical models, 
e.g., when $\mu({\bf s})$ is a function of the coordinates). Here the notation ${\bf Y}_1 \stackrel{d}{\approx} {\bf Y}_2$ means that random vectors
${\bf Y}_1$ and ${\bf Y}_2$ have approximately the same distribution.
As a result, it holds that
 \[
\tilde{\boldsymbol{z}} \; \stackrel{{\tiny \rm approx}}{\sim} 
{\rm N}\big(\tilde{\boldsymbol{X}}\bfbeta, 
\sigma^2 \boldsymbol{H}_1 \boldsymbol{G}_{\vartheta}\boldsymbol{H}_1^{\top}\big) ,
\]
with
\begin{align*}
\nonumber	
\boldsymbol{G}_{\vartheta} &= 
\frac{c_{\Delta}}{2M} {\rm diag}\Big(
\big(2 f^{\Delta}_{\vartheta}(\bfomega_{m_1,m_2}) : (m_1,m_2) \in I_{\rm C}\big)^{\top}, \
\big(f^{\Delta}_{\vartheta}(\bfomega_{m_1,m_2}) : (m_1,m_2) \in I\big)^{\top}, \\
& \hspace{6cm} \big(f^{\Delta}_{\vartheta}(\bfomega_{m_1,m_2}) : (m_1,m_2) \in I\big)^{\top} \Big) .
\label{eq:G}
\end{align*}
Hence, the covariance matrix of $\tilde{\boldsymbol{z}}$ is approximated by the 
orthogonal basis formed by the columns of $H_1$, which provides the sought approximation to the joint distribution of $Z(\cdot)$
in $\mathcal{U}_M$.
To illustrate the quality of the approximation we consider processes defined
on $\mcD = [0,1]^2$ having isotropic Mat\'ern correlation functions with $\vartheta=0.4$.
Figure \ref{fi:exact-approx-corr} displays the correlation functions of $Z(\cdot)$
(solid black lines) when $\nu = 0.5, 1.5$ and $2.5$ (top, middle and bottom panels),
and the corresponding correlation functions of $T_{M_1, M_2}(\cdot)$ (broken red lines).
To compute the latter we used $\Delta = 0.1$, $M_1 = M_2 = 20$ for the left panels and
$M_1 = M_2 = 30$ for the right panels.
These show that the approximations are quite precise for most distances, except when 
$M_1$ and $M_2$ are not large enough. In this case the approximation is poor for large distances 
due to the periodic nature of the spectral approximation.
But as long as $M_1$ and $M_2$ are chosen large enough, the approximation is excellent for 
all distances relevant to the region $\mcD$.

\begin{figure}[t]
\begin{center}
\psfig{figure=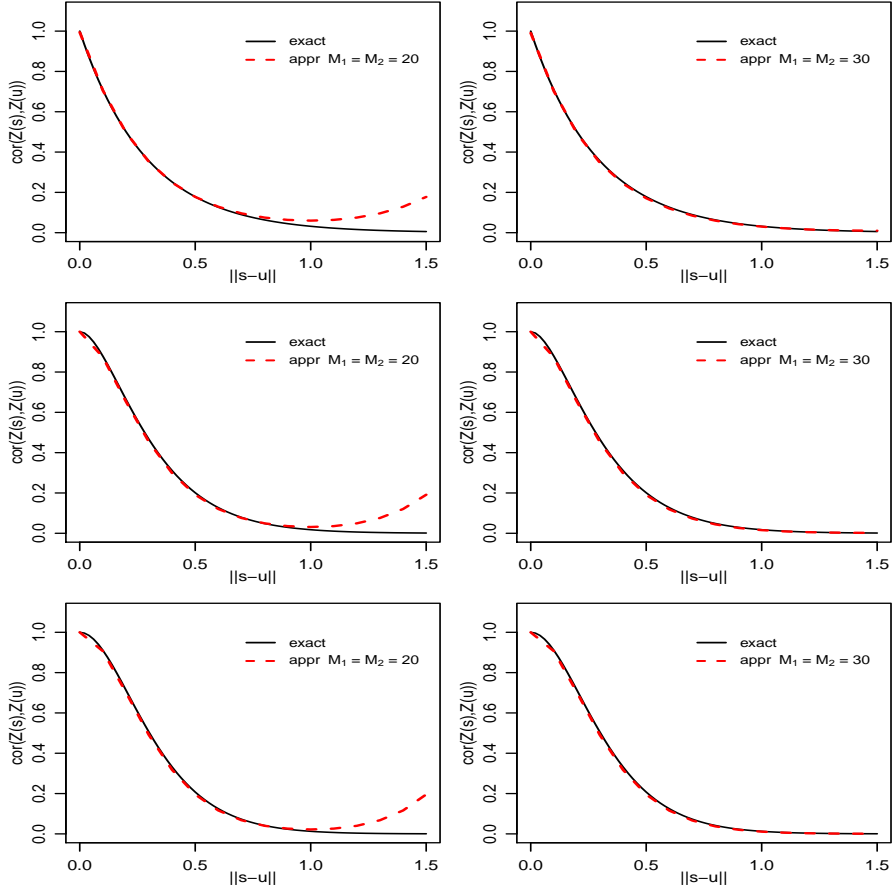, width=12cm, height=12cm} 
\end{center}
\vspace{-0.5 cm}
\caption{Plots of the correlation functions of $Z(\cdot)$ (solid black lines) 
and $T_{M_1, M_2}(\cdot)$ (broken red lines) corresponding to the isotropic Mat\'ern correlation 
functions with $\vartheta=0.4$ and $\nu = 0.5, 1.5$ and $2.5$ (top, middle and bottom panels).
In all $\Delta = 0.1$ was used.}
\label{fi:exact-approx-corr}
\end{figure}

\subsection{Approximate Integrated Likelihood}

Let $\tilde{\boldsymbol{\Sigma}}_{\vartheta} \coloneqq {\rm var}(\tilde{\boldsymbol{z}})$.
Although $\tilde{\boldsymbol{\Sigma}}_{\vartheta} \approx 
\boldsymbol{H}_1 \boldsymbol{G}_{\vartheta}\boldsymbol{H}_1^{\top}$
when $M_1$ and $M_2$ are  both large, 
replacing the former matrix with the latter in (\ref{eq:intlik-org}) or (\ref{eq:intlik-var})
does not generally result in a computationally convenient approximation of the integrated likelihood
of $\bftheta$. 
So we explore the alternative route of computing reference priors from the likelihood of a special 
linear combination of $\tilde{\boldsymbol{z}}$, somewhat similar to what is done for estimation of 
variance components using restricted likelihoods.
In all that follows, the aliased spectral density $f^{\Delta}_{\vartheta}(\bfomega)$ is approximated by 
truncating the series (\ref{eq:spec-den-alias}) so that only the terms for which
$\max\{|l_1|, |l_2|\} \leq T$ are retained, for some $T \in {\mathbb N}$;
this approximation is denoted by $\tilde{f}^{\Delta}_{\vartheta}(\bfomega)$.
Extensive numerical exploration shows that when $T$ is chosen in the range 3--6, 
the contribution of additional terms in (\ref{eq:spec-den-alias}) is negligible, 
so $\tilde{f}^{\Delta}_{\vartheta}(\bfomega)$ is not sensitive to $T$; 
see Section \ref{sec:numerical studies}.

Let $\boldsymbol{L}_1 \coloneqq \boldsymbol{H}_1 (\boldsymbol{H}_1^{\top}\boldsymbol{H}_1)^{-1/2}$, 
where $\boldsymbol{H}_1$ is the matrix defined above, and
$\boldsymbol{V}_1 \coloneqq \boldsymbol{L}_1^{\top} \tilde{\boldsymbol{z}}$.
Because of the regular arrangements of locations ${\bfu}_{i,j}$ and frequencies 
$\bfomega_{m_1, m_2}$, and the orthogonality properties of  cosines and sines, 
it holds that 
\begin{equation}
\boldsymbol{H}^{\top} {\bf 1}_M = {\bf 0}_{M-1} \quad  {\rm and}  \quad
\boldsymbol{H}_1^{\top}\boldsymbol{H}_1 = 
M {\rm diag}\big(1,1,1,1,\underbrace{2,\ldots,2}_{\tiny \mbox{$M - 4$ times}}\big) ;
\label{orthogonality}
\end{equation}
(see for instance \citet[Appendix E]{Bose2018}).
From these facts, direct calculation shows that 
\begin{equation}
\label{eq:spec-approx-covar}
\boldsymbol{V}_1 \stackrel{{\scriptsize \rm approx}}{\sim}  
{\rm N} \left( \boldsymbol{X}_1 \bfbeta, 
\sigma^2 \tilde{\boldsymbol{\Lambda}}_{\vartheta} \right),
\end{equation}
where $\boldsymbol{X}_1 := \boldsymbol{L}_1^{\top}\tilde{\boldsymbol{X}}$ is an $M \times p$ matrix 
with full rank $p$, and 
$\tilde{\boldsymbol{\Lambda}}_{\vartheta}$ is the diagonal matrix 
\begin{align}
\nonumber	
\tilde{\boldsymbol{\Lambda}}_{\vartheta} &= c_{\Delta} {\rm diag}\Big(
\big(\tilde{f}^{\Delta}_{\vartheta}(\bfomega_{m_1,m_2}) : (m_1,m_2) \in I_{\rm C}\big)^{\top}, \
\big(\tilde{f}^{\Delta}_{\vartheta}(\bfomega_{m_1,m_2}) : (m_1,m_2) \in I\big)^{\top}, \\
& \hspace{6cm} \big(\tilde{f}^{\Delta}_{\vartheta}(\bfomega_{m_1,m_2}) : (m_1,m_2) \in I\big)^{\top} \Big) .
\label{eq:Lambda}
\end{align}
Although the components of $\boldsymbol{V}_1$ are not error contrasts in general, its approximate 
covariance matrix is substantially simpler (diagonal) than that of $\tilde{\boldsymbol{z}}$.
So applying the reference prior algorithm described in Section \ref{sec:derivations} based on 
the likelihood of $\boldsymbol{V}_1$ will result in substantial simplifications.
Section \ref{sec:numerical studies} shows that this route delivers close approximations to 
reference priors when $\mathcal{U}_M$ is tuned to the features of the sampling design $\mathcal{S}_n$.

An important special case is that of models with constant mean function, i.e., when $p=1$.
In this case, $\boldsymbol{X}_1 \beta_1 = (\sqrt{M}\beta_1, \; {\bf 0}_{M-1}^{\top})^{\top}$
so the last $M-1$ components of $\boldsymbol{V}_1$ form a set of $M-1$ linearly independent 
error contrasts of $ \tilde{\boldsymbol{z}}$.
As a result, direct calculation from (\ref{eq:spec-approx-covar}) shows that the 
restricted log--likelihood function of $\bftheta$ based on $\tilde{\boldsymbol{z}}$ is, 
up to an additive constant, approximately equal to
\begin{equation}
l^{\rm AI}(\bftheta ; \tilde{\boldsymbol{z}}) \; = \;
-\frac{1}{2} \sum_{j=1}^{M-1} \left(  
\log\Big( c_{\Delta}\sigma^2 \tilde{f}^{\Delta}_{\vartheta}
(\bfomega_j) \Big) 
	+ \frac{V_j^2}{c_{\Delta} \sigma^2 \tilde{f}^{\Delta}_{\vartheta}(\bfomega_j) }\right) ,
\label{eq:approx-restr-loglik}
\end{equation}
where $\bfomega_j$ is a re--indexing of the frequencies $\bfomega_{m_1, m_2}$ appearing 
in (\ref{eq:Lambda}), with $\bfomega_{0,0}$ removed, 
and $V_1,\ldots,V_{M-1}$ are 
the last $M-1$ components of $\boldsymbol{V}_1$. 
In this case, even more substantial simplifications accrue in the computation of 
approximate reference priors since (\ref{eq:approx-restr-loglik}) is a matrix--free expression
and the required expectations are simplified as the $V_j$s are independent with 
$V_j^2 \stackrel{{\tiny \rm approx}}{\sim} 
{\rm Gamma}\big(1/2, 2\sigma^2 c_{\Delta} \tilde{f}^{\Delta}_{\vartheta}(\bfomega_j)\big)$ 
(shape--scale parametrization). 
Additionally, for Mat\'ern correlation functions, differentiation with respect to the 
range parameter $\vartheta$ is simplified since $l^{\rm AI}(\bftheta ; \tilde{\boldsymbol{z}})$ is 
devoid of Bessel functions.
Finally, \cite{Harville1974} showed that the integrated likelihood 
$L^{\rm I}(\bftheta; \tilde{\boldsymbol{z}})$ is proportional to the restricted likelihood function 
of $\bftheta$ based on a set of error contrasts $\boldsymbol{A}^{\top}\tilde{\boldsymbol{z}}$ 
when $\boldsymbol{A}$ satisfies $\boldsymbol{A}^{\top}\boldsymbol{A} 
= \boldsymbol{I}_{M-1}$ and $\boldsymbol{A}\boldsymbol{A}^{\top} = 
\boldsymbol{I}_{M} - \frac{1}{M}{\bf 1}_{M}{\bf 1}_{M}^{\top}$.
Since the ratio of the restricted likelihood functions of $\bftheta$ based on any two
sets of linearly independent error contrasts does not depend on $\bftheta$, it follows that
$\log L^{\rm I}(\bftheta; \tilde{\boldsymbol{z}})$ is, up to an additive constant, 
approximately equal to (\ref{eq:approx-restr-loglik}).

\begin{figure}[t]
\begin{center}
\psfig{figure=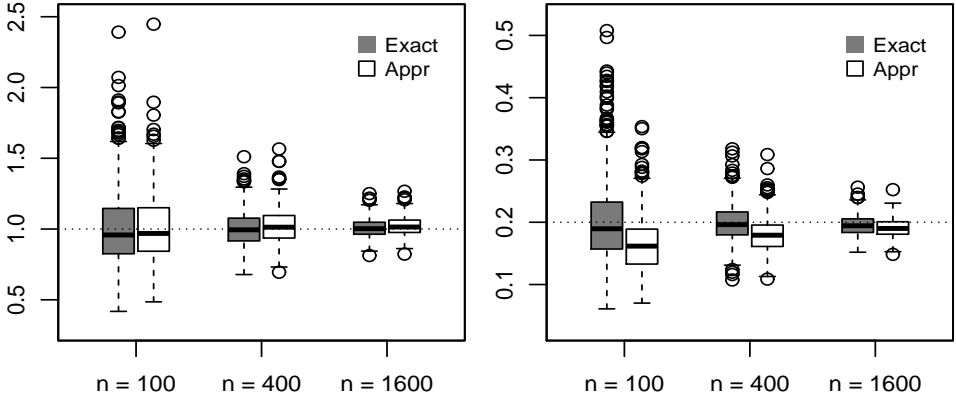, width=13cm, height= 5.5cm} 
\end{center}
\vspace{-0.5 cm}
\caption{Boxplot of exact and approximate REML estimates of $\sigma^2$ (left) and $\vartheta$ (right)
for different sample sizes. The true model is a Gaussian random field with mean 0 and Mat\'ern covariance function
with parameters $\sigma^2 = 1$, $\vartheta = 0.2$ and $\nu = 0.5$.}
\label{fi:reml-sigsqtheta}
\end{figure}

 Consider situations where the sampling design is regular, and the mean function is constant. 
By setting $M=n$ and $\mathcal{U}_M = \mathcal{S}_n$, so $\tilde{\boldsymbol{z}} = \boldsymbol{z}$,
restricted maximum likelihood (REML) estimates of the covariance parameters can be
approximated by maximizing (\ref{eq:approx-restr-loglik}). 
This is appealing when the sample size is large, since in this case the computation of exact REML estimates 
(obtained by maximizing (\ref{eq:intlik-org}) or (\ref{eq:intlik-var}))
may be very time--consuming or even unfeasible.
We ran a small simulation to compare the sampling distributions of exact and approximate REML estimators.
For each sample size $n=100, 400$ and $1600$, the Gaussian random field with mean 0 and 
Mat\'ern covariance function with $\sigma^2 = 1$, $\vartheta = 0.2$ and $\nu = 0.5$ was simulated $500$ times 
over the $\sqrt{n} \; \times \sqrt{n}$ regular lattice with $\Delta = 0.1$, and for each simulated data set 
exact and approximate REML estimates of $\bftheta = (\sigma^2,\vartheta)$ were computed, assuming $\nu$ known.
Figure \ref{fi:reml-sigsqtheta} displays boxplots from the REML estimates of $\sigma^2$ (left) and $\vartheta$ (right).
This suggests the sampling distributions of exact and approximate REML estimators of $\sigma^2$ are close,
even in small samples. 
On the other hand, the sampling distributions of exact and approximate REML estimators of $\vartheta$ 
are close only for large samples. For small samples, approximate REML estimators of $\vartheta$ are
(downward) biased and less variable than their exact counterparts.
The same behaviours were observed for other model settings (not shown).
In terms of computational effort, when $n=1600$ the computation of exact REML estimates took 597 seconds 
on average, while the computation of approximate REML estimates took 1.71 seconds 
(in this work, computation time were reported  
using a MacBook Pro with 2.3 GHz Intel Core i9 processor under the {\tt R} programming language).

\section{Approximate Reference Priors}  \label{sec:appr-ref-prior}

The derivation of the approximate reference prior, to be denoted as $\pi^{\rm AR}(\vartheta)$,
proceeds as follows.
Rather than using the exact integrated likelihood, (\ref{eq:intlik-org}) or (\ref{eq:intlik-var}), 
based on the data $\boldsymbol{z}$ measured at $\mathcal{S}_n$,
%to compute $\pi^{\rm R}(\vartheta)$, 
we use the approximate integrated likelihood derived 
from the potential summary (\ref{eq:spec-approx-covar}) measured at $\mathcal{U}_M$. 
This summary has a substantially simpler (diagonal) covariance matrix which, for 
the Mat\'ern and other families, is also devoid of special functions.
This makes the evaluation and analysis of the resulting approximate reference prior 
much more manageable than those of the exact reference prior.
In what follows, we state expressions for the approximate reference priors and 
establish the propriety of the corresponding approximate reference posteriors.
In these it is assumed that the covariates, if any, are available everywhere.

\begin{thm}[Approximate Reference Prior]
\label{thm:approx-prior-covar}
The approximate reference prior of 
$(\bfbeta, \sigma^2, \vartheta)$ derived from (\ref{eq:spec-approx-covar}) 
is given by 
$\pi^{\rm AR}(\bfbeta, \sigma^2, \vartheta) \propto \frac{\pi^{\rm AR}(\vartheta)}{\sigma^2}$, where 
\begin{equation}
\label{eq:approx-prior-covar}
\pi^{\rm AR}(\vartheta) \propto 
\left\{ {\rm tr}\left[\left\{ 
\left(\frac{\partial}{\partial \vartheta} \tilde{\boldsymbol{\Lambda}}_{\vartheta}\right) 
\tilde{\boldsymbol{Q}}_{\vartheta}\right\}^2 \right] - \frac{1}{M-p} 
	\left[{\rm tr}\left\{\left(\frac{\partial}{\partial \vartheta} \tilde{\boldsymbol{\Lambda}}_{\vartheta}
\right) \tilde{\boldsymbol{Q}}_{\vartheta}\right\} \right]^{2} \right\}^{\frac{1}{2}},
\end{equation}
with $\tilde{\boldsymbol{\Lambda}}_{\vartheta}$ defined in (\ref{eq:Lambda}) and 
$\tilde{\boldsymbol{Q}}_{\vartheta} := \tilde{\boldsymbol{\Lambda}}_{\vartheta}^{-1} - 
\tilde{\boldsymbol{\Lambda}}_{\vartheta}^{-1}\boldsymbol{X}_1 
(\boldsymbol{X}_1^{\top} \tilde{\boldsymbol{\Lambda}}_{\vartheta}^{-1}
\boldsymbol{X}_1)^{-1} \boldsymbol{X}_1^{\top} \tilde{\boldsymbol{\Lambda}}_{\vartheta}^{-1}$. 
 
\medskip

\noindent
Proof. The result follows from (\ref{eq:reference-prior-1}) by replacing $\boldsymbol{X}$ with 
$\boldsymbol{X}_1$ and $\boldsymbol{\Sigma}_{\vartheta}$ with $\tilde{\boldsymbol{\Lambda}}_{\vartheta}$.
\end{thm}

Note that $\tilde{\boldsymbol{\Lambda}}_{\vartheta}$ is a diagonal matrix so its inverse is 
easy to compute.  
The computation of $\tilde{\boldsymbol{Q}}_{\vartheta}$ only involves the inversion of the (small) 
$p \times p$ matrix  $\boldsymbol{X}_1^{\top}\tilde{\boldsymbol{\Lambda}}_{\vartheta}^{-1}\boldsymbol{X}_1$, 
where the matrix $\boldsymbol{X}_1$ needs to be computed only once since $\boldsymbol{H}_1$ is fixed. 
The diagonal elements of the diagonal matrix 
$(\partial \tilde{\boldsymbol{\Lambda}}_{\vartheta}/ \partial \vartheta)
\tilde{\boldsymbol{\Lambda}}_{\vartheta}^{-1}$, namely
$(\partial/\partial \vartheta)\log \tilde{f}^{\Delta}_{\vartheta}(\bfomega_{m_1,m_2})$, 
have a closed--form expression which for the Mat\'ern and other models 
is devoid of special functions. 
As a result, the computation and analysis of $\pi^{\rm AR}(\vartheta)$ is substantially simpler 
than that of $\pi^{\rm R}(\vartheta)$.
Note $\pi^{\rm AR}(\vartheta)$ could also be obtained from (\ref{eq:reference-prior-2}), but 
the resulting expression does not afford computational savings, so it is omitted. 

An important special case of the above result occurs when the mean function is constant,
in which case the approximate reference prior of $\vartheta$ takes an even simpler 
matrix--free form.

\begin{cor}[Constant Mean Case]
\label{cor:approx-prior-rho-nocov}
Consider models with constant mean function.  
In this case, the approximate reference prior of $(\beta_1, \sigma^2, \vartheta)$ 
is $\pi^{\rm AR}(\beta_1, \sigma^2, \vartheta)\propto \frac{\pi^{\rm AR}(\vartheta)}{\sigma^2}$, 
where
\begin{equation}
\label{eq:approx-prior-rho-nocov}
\pi^{\rm AR}(\vartheta) \propto 
	\left\{\sum_{j = 1}^{M-1} 
\left(\frac{\partial}{\partial \vartheta} \log \tilde{f}^{\Delta}_{\vartheta}(\bfomega_j) \right)^2 
- \frac{1}{M-1} \bigg(\sum_{j = 1}^{M-1} 
\frac{\partial}{\partial \vartheta} \log \tilde{f}^{\Delta}_{\vartheta}(\bfomega_j) \bigg)^2 \right\}^{\frac{1}{2}} ,	
\end{equation}
where $\bfomega_j$ is a re--indexing of the frequencies $\bfomega_{m_1, m_2}$ 
in (\ref{eq:Lambda}), with $\bfomega_{0,0} = (0,0)^{\top}$ removed. 
\medskip

\noindent
Proof. See the Appendix.
This result can also be obtained by applying the last step of the reference prior algorithm to the 
approximate log--integrated likelihood (\ref{eq:approx-restr-loglik}) (not shown). 
\end{cor}

It should be noted that, because of isotropy, the (unnormalized) prior $\pi^{\rm AR}(\vartheta)$
can also be computed by including in the sum (\ref{eq:approx-prior-rho-nocov})
all frequencies in $\mathcal{W}_{M} - \{\bfomega_{0,0}\}$,
which is proportional to the sample standard deviation of the derivative w.r.t. the 
range parameter of the log aliased spectral density evaluated at these frequencies. 

To establish the propriety behaviour of the approximate reference prior and 
posterior, we make the following assumptions about the second--order
structure of $Z(\cdot)$:
\begin{enumerate}

\item[(A1)]
The family of (normalized) spectral densities satisfies 
$\int_{{\mathbb R}^2} f_{\vartheta}(\bfomega) d\bfomega  = 1$ for all $\vartheta > 0$,
and has the form
\begin{equation*}
f_{\vartheta}(\bfomega) = \frac{h_1(\bfomega) h_2(\vartheta)} 
{\big(\|\bfomega\|^2 + u(\vartheta) \big)^{a}} ,
\label{general-spec-den}
\end{equation*}
where 

$\bullet$ $h_1(\bfomega)$ is non--negative and continuous in ${\mathbb R}^2$, and $a$ is a constant.

$\bullet$ $h_2(\vartheta)$ and $u(\vartheta)$ are positive and continuously differentiable functions
on $(0,\infty)$. 
		
$\bullet$ $\lim_{\vartheta \rightarrow 0^+} f_{\vartheta}(\bfomega) = 0$ \;
(this implies $\lim_{\vartheta \rightarrow 0^+} K_{\vartheta}(r) = {\bf 1}\{r = 0\}$ (white noise)). 

\item[(A2)]
The correlation matrix $\boldsymbol{\Sigma}_{\vartheta}$ can be expressed as 
\begin{equation*}
\boldsymbol{\Sigma}_{\vartheta} = 
\sum_{i = 0}^{J}q_{i}(\vartheta) \boldsymbol{D}^{(i)} + \boldsymbol{R}(\vartheta), \quad\quad 
{\rm as} \ \ \vartheta \rightarrow \infty ,
\label{sigma-expansion-mure}
\end{equation*}
where $J \in {\mathbb N}$, the $q_{i}(\vartheta)$s are continuous functions on $(0, \infty)$,
the $\boldsymbol{D}^{(i)}$s are fixed symmetric matrices satisfying 
$\cap_{i = 0}^{J}{\rm Ker}(\boldsymbol{D}^{(i)}) = \{{\bf 0}_n\}$, and 
$\boldsymbol{R}(\vartheta)$ is a function from $(0, \infty)$ to the space of $n \times n$ real matrices.
\end{enumerate}

\begin{thm}[Propriety]  \label{thm:asymp-behav-approx-prior-rho}
Assume the mean function $\mu({\bf s})$ has an intercept (so $f_1({\bf s}) \equiv 1$).

(a) If $f_{\vartheta}(\bfomega)$ satisfies assumption (A1), then the approximate 
marginal reference prior $\pi^{\rm AR}(\vartheta)$ in (\ref{eq:approx-prior-covar}) 
is a continuous function satisfying
 \begin{equation}
  \pi^{\rm AR}(\vartheta) = 
  O\left( 
|u'(\vartheta)| \land \frac{|u'(\vartheta)|}{u^2(\vartheta)} \right),  \; \; \;
\text{as} \;\; \vartheta \to 0^{+} \;\; {\rm and} \;\; \vartheta \to \infty ,
\label{asymptotic-behaviour=pi=vartheta}
 \end{equation}
where $x \land y \coloneqq \min\{x, y\}$.  
So, if 
$|u'(\vartheta)| \land |u'(\vartheta)| u(\vartheta)^{-2}$ is integrable on $(0, \infty)$,
$\pi^{\rm AR}(\vartheta)$ is proper.

(b) If the second--order structure of $Z(\cdot)$ satisfies assumptions (A1) and (A2),
then the approximate reference posterior distribution based on the observed data, 
$\pi^{\rm AR}(\bfbeta, \sigma^2, \vartheta ~|~\boldsymbol{z}) \propto
L(\bfbeta, \sigma^2, \vartheta; \boldsymbol{z})\pi^{\rm AR}(\bfbeta, \sigma^2, \vartheta)$, 
is proper.

\medskip

\noindent
Proof. See the Appendix.
\end{thm}

\begin{cor}[Propriety for the Mat\'ern Family]  \label{cor:propriety-matern}
Consider a model determined by a mean function $\mu({\bf s})$ with an intercept and the Mat\'ern family 
of (normalized) spectral densities (\ref{eq:matern-spden}).
Then, $\pi^{\rm AR}(\vartheta)$ is integrable on $(0,\infty)$ and 
$\pi^{\rm AR}(\bfbeta, \sigma^2, \vartheta ~|~\boldsymbol{z})$ is proper.

\medskip

\noindent
Proof.
The Mat\'ern family (\ref{eq:matern-spden}) clearly satisfies (A1),
with $h_1(\bfomega) = \Gamma(\nu + 1) (4\nu)^\nu/\pi \Gamma(\nu)$, $h_2(\vartheta) = \vartheta^{-2\nu}$,
$u(\vartheta) = 4\nu/{\vartheta^2}$ and $a = \nu + 1$. Then
\[
|u'(\vartheta)| \land \frac{|u'(\vartheta)|}{u^2(\vartheta)} 
\ = \ \frac{8\nu}{\vartheta^3} \land \frac{\vartheta}{2\nu}
\ = \
\begin{cases}
o(1) , & \text{as} \ \vartheta \to 0^{+} \\
O(\vartheta^{-3}), & \text{as}\ \vartheta \to \infty  
\end{cases} ,
\]
is integrable on $(0,\infty)$, so by Theorem \ref{thm:asymp-behav-approx-prior-rho}(a)
$\pi^{\rm AR}(\vartheta)$ is proper.
Also, it was shown in %\cite{Berger2001} and 
\cite{Mure2021} that the Mat\'ern family (\ref{eq:matern-spden}) satisfies (A2), 
so the propriety of $\pi^{\rm AR}(\bfbeta, \sigma^2, \vartheta ~|~\boldsymbol{z})$ follows from
Theorem \ref{thm:asymp-behav-approx-prior-rho}(b).
\end{cor}

The preceding results provide a theoretical justification for using the approximate reference prior 
under either constant mean or non--constant mean model with a common intercept. 
Note that, in general, $\pi^{\rm AR}(\vartheta)$ is proper and its tail rate 
as $\vartheta \to \infty$ is the same regardless of the mean function and degree of smoothness 
of the random field.
On the other hand, \cite{Mure2021} showed that, in general, 
$\pi^{\rm R}(\vartheta) = O(1/\vartheta)$ as $\vartheta \rightarrow \infty$, and 
this tail behaviour is sharp for some models.
Mur\'e also showed that for some special models, other tail behaviours hold; 
\citep[Appendix B]{Mure2021}.
Consequently, $\pi^{\rm R}(\vartheta)$ is not always proper and the proposed use
of (exact) reference priors and Bayes factors discussed in \citet[Section 6]{Berger2001} 
for selecting smoothness in correlation families are not valid.
In contrast, $\pi^{\rm AR}(\vartheta)$ can always be used for this purpose; 
this is illustrated in Section \ref{sec:example}.

The marginal prior $\pi^{\rm AR}(\vartheta)$ depends on the tuning constants $M_1$, $M_2$ 
and $\Delta$ that need to be tuned to the sampling design $\mathcal{S}_n$. 
Since these have specific interpretations in terms of the spectral approximation, 
their selection is more straightforward than using a subjectively chosen  prior, for example, 
an inverse gamma prior, since it is unclear how to select the hyperparameters; 
this is discussed in Section \ref{sec:numerical studies}.

\medskip

{\bf Discussion}. 
Assumption (A1) is satisfied by several families of spectral densities proposed in the literature, 
after a reparametrization if needed. 
In addition to the Mat\'ern family, the family proposed by \cite{Laga2017} 
(assuming their parameters b and $\xi$ are known),
is of this form with $h_1(\bfomega) = ({\rm b}^2 + \|\bfomega\|^2)^\xi$, $h_2(\vartheta) \propto 1$, 
$u(\vartheta) = 1/\vartheta^2$ and $a = \nu + 1$.
%${\rm a} = 1/\vartheta$, $a_1 = 2$ 
Also, some of the families of spectral densities studied in 
\cite{Vecchia1985} and \cite{Jones1993} 
are of this form, after they are suitably parametrized.

Assumption (A2) is a more general expansion than that in (\ref{sigma-expansion}).
The latter occurs when $J = 1$, $q_{0}(\vartheta) = 1$, $\boldsymbol{D}^{(0)} = 
{\bf 1}_{n}{\bf 1}_{n}^{\top}$, 
$q_{1}(\vartheta) = q(\vartheta)$, $\boldsymbol{D}^{(1)} = \boldsymbol{D}$ and 
$\boldsymbol{R}(\vartheta) = o(q(\vartheta))$, with $q(\vartheta)$ and $\boldsymbol{D}$
defined circa (\ref{sigma-expansion}).
When $\boldsymbol{D}$ is non--singular, 
${\rm Ker}({\bf 1}_{n}{\bf 1}_{n}^{\top}) \cap {\rm Ker}(\boldsymbol{D}) = \{{\bf 0}_n\}$ 
clearly holds. 
Likewise, a sufficient (but not necessary) condition for 
$\cap_{i = 0}^{J}{\rm Ker}(\boldsymbol{D}^{(i)}) = \{{\bf 0}_n\}$ to hold is that 
at least one matrix $\boldsymbol{D}^{(i)}$ is non--singular.
\cite{Mure2021} checked that assumption (A2) holds for several commonly used families of 
covariance functions, including the Mat\'ern family.

\section{Numerical Studies}  \label{sec:numerical studies}

In this section, we conduct numerical studies to explore how close the marginal priors 
$\pi^{\rm R}(\vartheta)$ and $\pi^{\rm AR}(\vartheta)$ are for various sampling designs and 
model features, and provide empirical guidelines for the selection of the tuning constants 
$M_1, M_2$ and $\Delta$.
Additionally, we also compare the computational efforts for their computation.

We consider two regular designs, a $10 \times 10$ equally spaced grid in $[0, 1]^2$
and a $20 \times 20$ equally spaced grid in $[0, 2]^2$, as well as three irregular sampling designs
in $[0, 1]^2$ of size $n = 100$, to be described below.
For the mean function we consider $\mu({\bf s}) = 1$ and 
$\mu({\bf s}) = 0.15-0.65x-0.1y+0.9x^2-xy+1.2y^2$, with ${\bf s} = (x, y)$,
and for the covariance function we consider the isotropic Mat\'ern model (\ref{eq:matern-cov})
with $\nu = 0.5, 1.5$ and $2.5$.
In all cases the approximate reference priors are computed with $\tilde{f}^{\Delta}_{\vartheta}(\bfomega_j)$ 
obtained by truncating the series (\ref{eq:spec-den-alias}) so that only the terms with 
$\max\{|l_1|, |l_2|\} \leq 5$
%$|l_1 + l_2| \leq 10$ 
are retained.
These approximate reference priors show no sensitivity to the truncation point.

\begin{figure}[t!]
\begin{center}
\psfig{figure=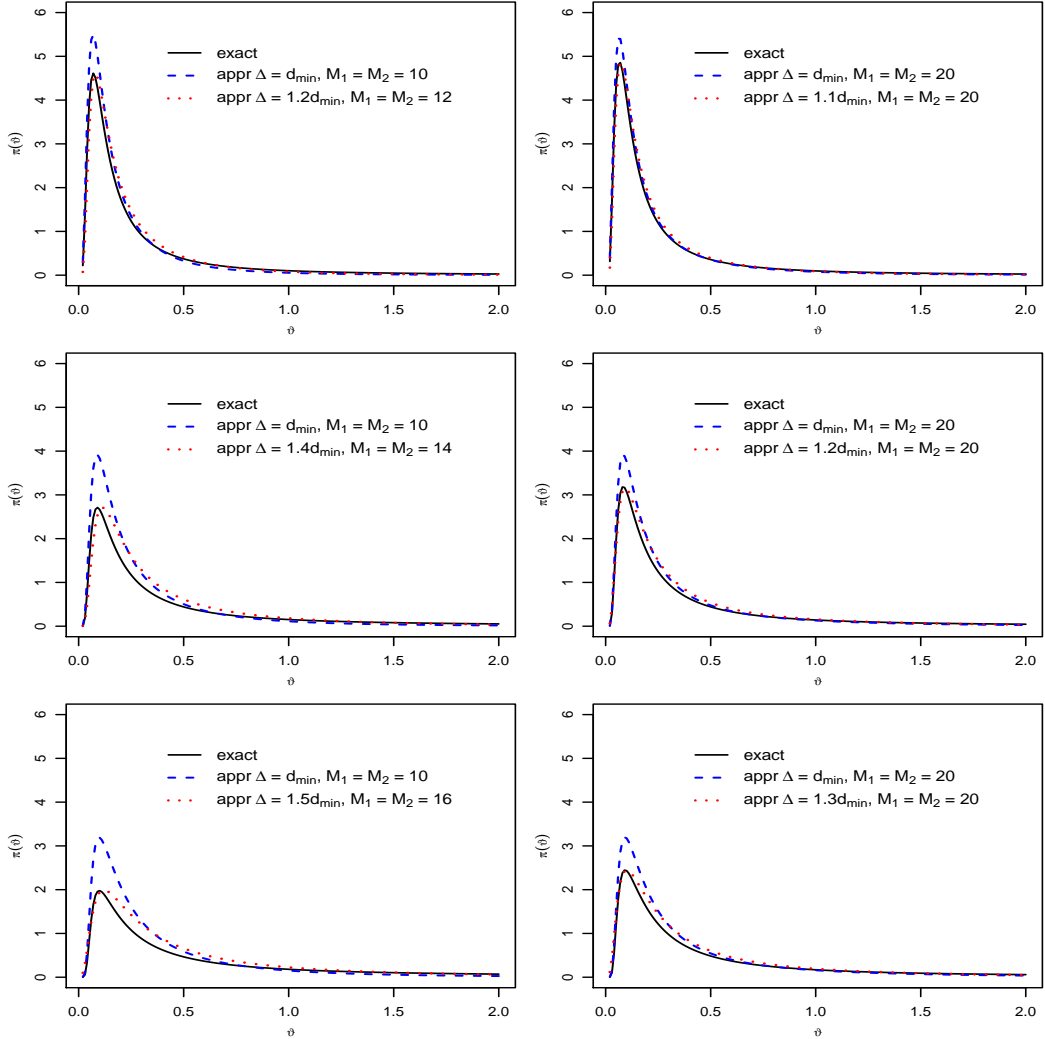, width=14cm,height=14cm} 
\end{center}
\vspace{-0.5 cm}
\caption{Marginal densities of the exact and approximate reference priors of $\vartheta$ for the
constant mean models under different sample designs and covariance smoothness.  
Left: $10\times 10$ equally spaced grid in $[0, 1]^2$.
Right: $20 \times 20$ equally spaced grid in $[0, 2]^2$.
From top to bottom: $\nu = 0.5, 1.5$ and $2.5$.}
\label{fi:prior-rho-nc}
\end{figure}

\begin{figure}[t!]
\begin{center}
\psfig{figure=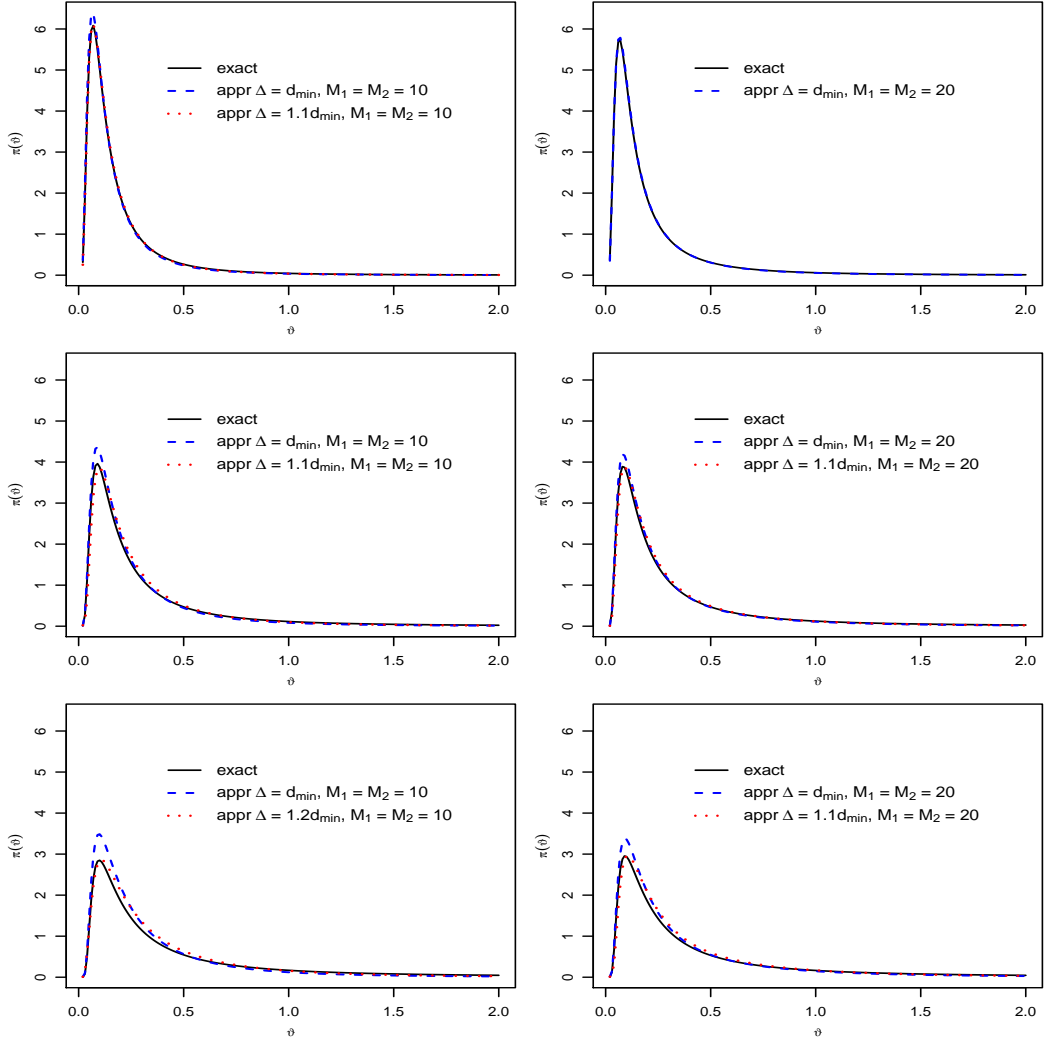, width=14cm,height=14cm} 
\end{center}
\vspace{-0.5 cm}
\caption{Marginal densities of the exact and approximate reference priors of $\vartheta$ for the
non--constant mean models under different sample designs and covariance smoothness.  
Left: $10\times 10$ equally spaced grid in $[0, 1]^2$.
Right: $20 \times 20$ equally spaced grid in $[0, 2]^2$.
From top to bottom: $\nu = 0.5, 1.5$ and $2.5$.}
\label{fi:prior-rho-c}
\end{figure}

Figure \ref{fi:prior-rho-nc} displays the (normalized) reference priors of $\vartheta$ 
based on the regular designs for models with constant mean. 
The left panels are the priors based on the grid in $[0, 1]^2$, the right panels
are the priors based on the grid in $[0, 2]^2$, and the top, middle and bottom panels 
are the priors obtained when $\nu = 0.5, 1.5$ and $2.5$, respectively.
The solid black curves are exact reference priors, and the broken colored curves are approximate reference priors.
As the default choice to compute $\pi^{\rm AR}(\vartheta)$ we use $\mathcal{U}_M = \mathcal{S}_n$, 
i.e., we set $M_1 = M_2 = \sqrt{n}$ and $\Delta = d_{\min}$, 
the distance between adjacent sampling locations.
The resulting approximate reference priors (broken blue curves) display the same shapes as the exact
reference priors, both having about the same mode, but do not provide a very close approximation
in general.
But the approximation improves substantially when $M_1$, $M_2$ and/or $\Delta$ are tuned.
Figure \ref{fi:prior-rho-nc} also displays approximate reference priors (broken red curves)
obtained by setting $M_1 = M_2 \geq \sqrt{n}$ and $\Delta > d_{\min}$ at values indicated in the legends.
Now the approximate reference priors provide close approximations.
Less tuning is needed for the larger sample size, as using the default $M_1 = M_2 = \sqrt{n}$ 
and only tuning $\Delta$ results in good approximations; 
the required tuned value of $\Delta$ increases with the smoothness.
Additionally, and in agreement with part (b) of Lemma \ref{thm:spec-approx-2d}, 
the approximations in the right panels are closer to their exact counterparts than the ones in the 
left panels since $M_1$ and $M_2$ are larger for the former.

Figure \ref{fi:prior-rho-c} displays the (normalized) reference priors of $\vartheta$ based on 
the regular designs for models with non--constant mean, with the same layout used in
Figure \ref{fi:prior-rho-nc}.
The behaviours and conclusions are essentially the same as those in Figure \ref{fi:prior-rho-nc}.
But now the approximate reference priors are even closer to their exact counterparts, and
the default choice provides even closer approximations.
Moreover, setting $M_1 = M_2 = \sqrt{n}$ and only tuning $\Delta$ results in good approximations
also for the small sample size, and even no adjustment at all 
(i.e., also setting $\Delta = d_{\min}$) may provide a good approximation when the sample size is large and the process is not smooth.

The results in Figures \ref{fi:prior-rho-nc} and \ref{fi:prior-rho-c}, 
as well as additional numerical explorations (not shown), 
suggest that, as a rule of thumb, for sampling designs that are small regular grids,  
the adjustment involves setting $M_1$ (for simplicity $M_2 = M_1$) and $\Delta$ to values 10--40\% larger than $\sqrt{n}$ 
and $d_{\min}$, respectively.
For larger  regular grids it may suffice to set $M_1 = \sqrt{n}$ and only adjust $\Delta$ 
to a value about 10\% larger than $d_{\min}$. 
Overall, the approximate reference priors are more sensitive to $\Delta$ than to $M_1$ (the Supplementary Materials provide an illustration of this fact).
Finally, for regular grids with similar $d_{\min}$, both the exact and the approximate 
reference priors of $\vartheta$ are more sensitivity to $\nu$ than to $n$ or $M_1$, and 
the approximations seem to be closer for non--smooth random fields.

\begin{figure}[ht!]
\begin{center}
\psfig{figure=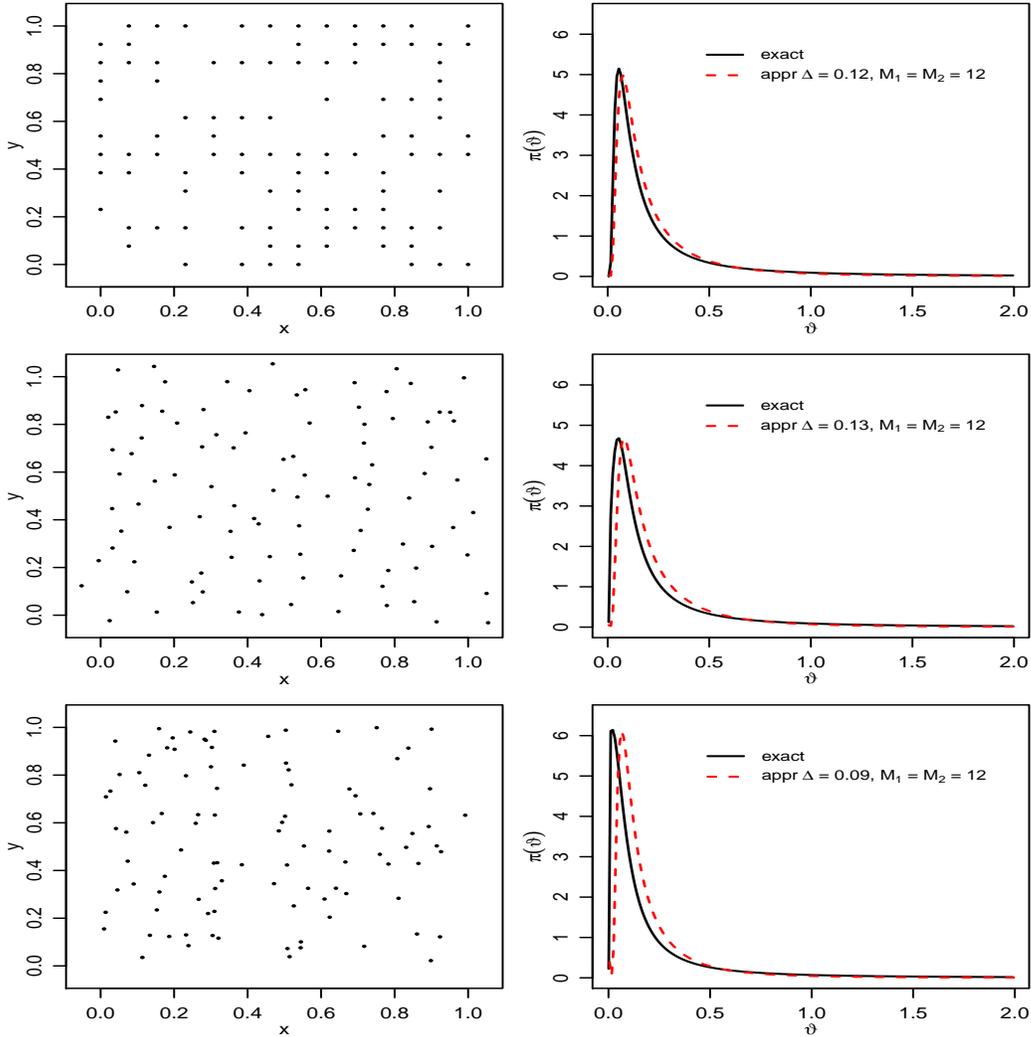,width=14cm,height=14cm} 
\end{center}
\vspace{-0.5 cm}
\caption{Left: Three different sampling designs in $[0,1]^2$. 
Right: Corresponding marginal densities of the exact and approximate reference priors
of $\vartheta$ under constant mean model when $\nu = 0.5$.}
\label{fi:prior-aref-ir}
\end{figure}

Next we consider three irregular sampling designs in $[0, 1]^2$ of size $n = 100$:
an incomplete $14\times 14$ regular grid, a hybrid design generated by the method proposed in 
\cite{Bachoc2014}, with $\epsilon = 0.499$, and a random sample from the ${\rm unif}((0, 1)^2)$ 
distribution. These designs are displayed in Figure \ref{fi:prior-aref-ir} (left panels).
Figure \ref{fi:prior-aref-ir} (right panels) displays the exact and approximate reference priors
of $\vartheta$ (solid black and broken red curves, respectively) based on these irregular designs
for the model with constant mean and Mat\'ern correlation with $\nu = 0.5$.
The approximate reference priors can still provide satisfactory approximations for practical purposes, although the discrepancy between the two priors increases with the degree of irregularity of the design.
The results for models with a non--constant mean and other degrees of smoothness displayed
similar behaviours (not shown).
The tuning of $M_1$ and $M_2$ was done similarly as described above for regular designs, but
now $\Delta$ is selected based on the distances to the nearest neighbors, 
$d_{i} = \min\{\|{\bf s}_{i} - {\bf s}_{j} \| : j \neq i \}$. 
It was empirically found that setting $\Delta$ at a value between the 75 to 95 percentiles of 
$\{d_{i}\}_{i = 1}^{n}$ provides reasonable approximations under the above designs. 
%More precise rules to tune $M_1, M_2$ and $\Delta$ will be explored elsewhere.
Overall, the numerical explorations reported in 
Figures \ref{fi:prior-rho-nc}--\ref{fi:prior-aref-ir} indicate that approximate 
reference priors, after properly tuned, provide satisfactory approximations to
exact reference priors for a variety of sampling designs and models. 
For large sample sizes when an approximation is most needed,
the tuning simplifies as we can set $M_1 = M_2 \approx \sqrt{n}$ and only select $\Delta$ 
using the aformentioned guideline.
The Supplementary Materials provide a more detailed comparison of the tail behaviour
of these priors, showing that approximate reference priors tend to have lighter tails than
their exact counterparts
(the text after Corollary \ref{cor:propriety-matern} explains the reason for this behaviour).

\begin{table}[t]
\centering
\caption{Computational time (in seconds) for $500$ evaluations of the exact and approximate 
reference priors of $\vartheta$ for regular grid sampling designs of different sample sizes and several model features.}
%These computations used $\Delta = 0.1$ and $M=n$.}}
\label{ta:prior-time-2}
\begin{tabular}{lllcccccc}
\toprule
$\nu$ & $p$ & Reference Prior & & $n = 100$   & $n = 400 $ & $n = 1600$ & $n = 10000$ \\
 \cline{1-3} \cline{5-8}
\multirow{4}{*}{$0.5$} & 
\multirow{2}{*}{$1$} & 
Exact &   & $1.18$   &  $41.56$  & $2652.75$  & -- \\
& & Approximate  &   & $0.21$   & $0.87$   & $2.05$   & $10.21$   \\
 \cline{2-3} \cline{5-8}
 & 
\multirow{2}{*}{$6$} & 
Exact &   &  $2.48$    & $61.49$ & $4098.96$   & --  \\
& & Approximate  &   & $0.68$   & $2.62$  
& $32.08$   & $712.09$  \\
\midrule
\multirow{4}{*}{$1$} & 
\multirow{2}{*}{$1$} & 
Exact &   & $3.80$   &  $81.46$  & $3459.14$    & -- \\
& & Approximate  &   & $0.14$   & $0.77$   & $2.10$   & $9.69$   \\
 \cline{2-3} \cline{5-8}
 & 
\multirow{2}{*}{$6$} & 
Exact &   & $8.01$    & $104.44$  & $5047.18$  & -- \\
& & Approximate  &   & $0.81$   & $3.04$   & $29.94$  & $694.33$  \\
\bottomrule
\end{tabular}
\end{table}

To discuss the computational complexity of exact and approximate reference priors,
we consider for simplicity a regular grid $\sqrt{n} \; \times \sqrt{n}$ sampling design and use $M = n$.
The computation of the exact reference prior $\pi^{\rm R}(\vartheta)$ in (\ref{eq:reference-prior-1}) 
requires $O(M^3)$ operations due to the need of numerically invert $\boldsymbol{\Sigma}_{\vartheta}$. 
On the other hand, for processes with constant mean the computation of the approximate reference prior 
$\pi^{\rm AR}(\vartheta)$ in (\ref{eq:approx-prior-rho-nocov}) only requires $O(M)$ operations.
For processes with a non--constant mean function the computation of $\pi^{\rm AR}(\vartheta)$ in
(\ref{eq:approx-prior-covar}) requires $O(M^2)$ operations (and does not involve the evaluation of 
special functions). 
This is so due to the need to compute the matrix $\boldsymbol{X}_1$ (only once), 
with $O(M^2)$ computational complexity, and then computing 
$\boldsymbol{X}_1 (\boldsymbol{X}_1^{\top} \tilde{\boldsymbol{\Lambda}}_{\vartheta}^{-1}
\boldsymbol{X}_1)^{-1} \boldsymbol{X}_1^{\top} \tilde{\boldsymbol{\Lambda}}_{\vartheta}^{-1}$,
which also has $O(M^2)$ computational complexity.

Table \ref{ta:prior-time-2} reports the timings for 500 evaluations of both marginal reference priors 
of $\vartheta$ under regular $\sqrt{n} \; \times \sqrt{n}$ sampling designs for models with constant and
non--constant mean functions and Mat\'ern covariance functions with $\nu = 0.5$ and $1$. 
For the evaluation of the approximate reference prior we used $\Delta = 0.1$ and $M=n$.
The evaluation of approximate reference priors is between one and two orders of magnitude faster than 
that of exact reference priors, and the computational time gap increases substantially with sample size.
In particular, the computation of exact reference priors becomes computationally unfeasible 
when $n=10000$. 

The Supplementary Materials report results from a simulation study to compare frequentist properties 
of Bayesian procedures based on approximate and exact reference priors (under two types of the sampling designs), as well as frequentist properties of 
a purely likelihood--based procedure (under 
the regular lattice design for illustrative purposes).
The results suggest that the credible intervals for the covariance parameters based on these two priors 
have similar and satisfactory frequentist coverage, and their expected lengths are also about the same in most case scenarios. 
In addition, the mean absolute errors of the Bayesian estimators of the range parameter based on these two priors 
are about the same, and these are smaller than the mean absolute error of maximum likelihood estimators.

\section{Example}  \label{sec:example}

We illustrate the application of default Bayesian analysis based on exact and approximate reference priors
with a data set analyzed by \cite{Diggle2010}, available in the {\tt R} package {\tt PrevMap}. 
%\citep{Giorgi2017}. 
The data set, which came about in the monitoring of lead pollution in Galicia, northern Spain, 
consists of measurements of lead concentrations in moss samples (in micrograms per gram dry weight). 
Data from two survey times were analyzed by \cite{Diggle2010}, 
one in October 1997 and the other in July 2000. 
Here we use the July 2000 data, since the 1997 data were collected using a preferential sampling design. 
The analysis uses the log--transformation of the original measurements to eliminate their 
variance--mean relationship, which renders the homoscedastic Gaussian assumption appropriate. 
A summary of the data is plotted in Figure \ref{fi:galicia-locs} (left),  showing
$132$ sampling locations where the unit of distance is 100 km. 

There are no covariates available and an exploratory analysis reveals no apparent spatial trend,
so the mean function is assumed constant.
Figure \ref{fi:galicia-locs} (right) displays the empirical semivariogram and the fitted 
(by least squares) semivariogram function $0.21\big(1 - \exp(-\sqrt{2}r/0.26)\big)$.
This corresponds the Mat\'ern covariance function (\ref{eq:matern-cov}) with $\nu=0.5$ which,
except for a slight reparametrization, is the exponential semivariogram model used in \cite{Diggle2010}.
The fit appears appropriate for the data and suggests the data contain no measurement error (no nugget).

The sampling locations are close to form a regular grid, but they are not strictly aligned. 
For most sites, the distances to their nearest neighbors are similar 
(the $90$ percentile is about $15$ km and the maximum is $21$ km). 
To compute the approximate reference prior we set $M_1 = M_2 = 16$ and $\Delta = 0.2$, and $\tilde{f}^{\Delta}_{\vartheta}(\bfomega_j)$ is obtained by setting $\max\{|l_1|, |l_2|\} \leq 5$. 
Two Bayesian analyses were carried out based on the exact and approximate reference priors,
where samples of size $10^4$ from the corresponding posteriors  of $(\beta_1,\sigma^2,\vartheta)$
were simulated using the Monte Carlo algorithm described in the Supplementary Materials. 
The acceptance rate in the ratio--of--uniforms step was about $75\%$.  

\begin{figure}[t!]
\begin{center}
\psfig{figure=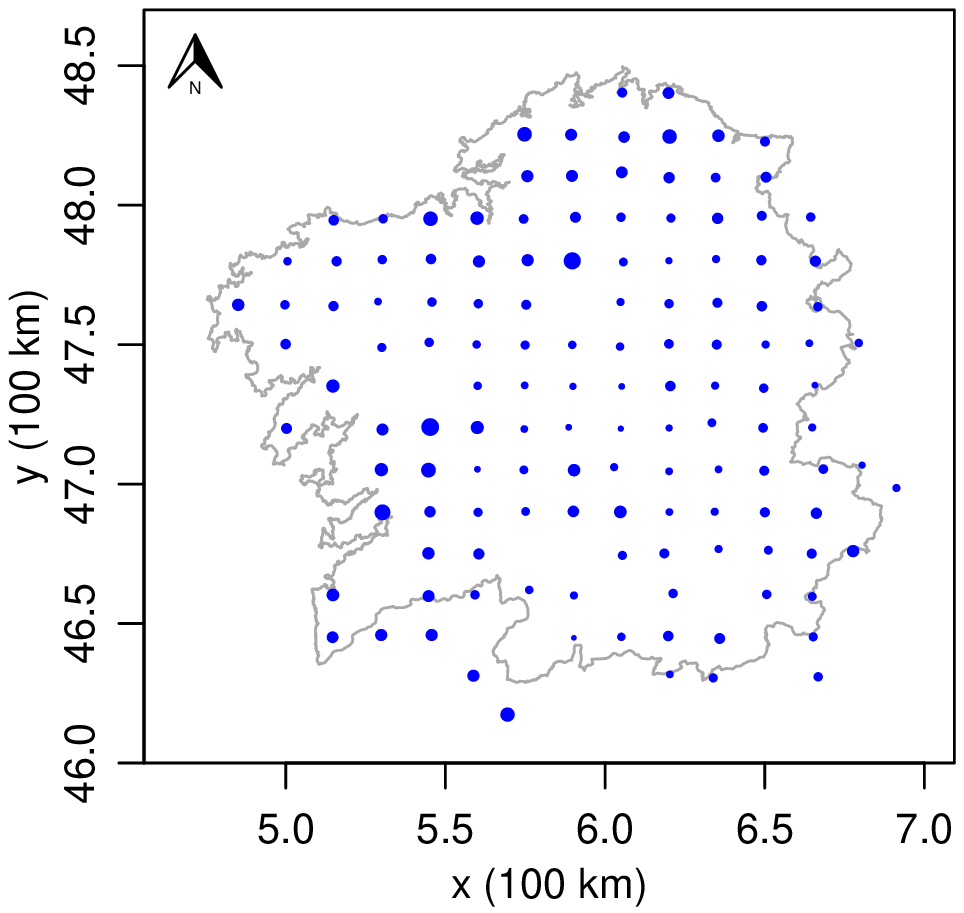, width=7cm,height=6cm} 
\psfig{figure=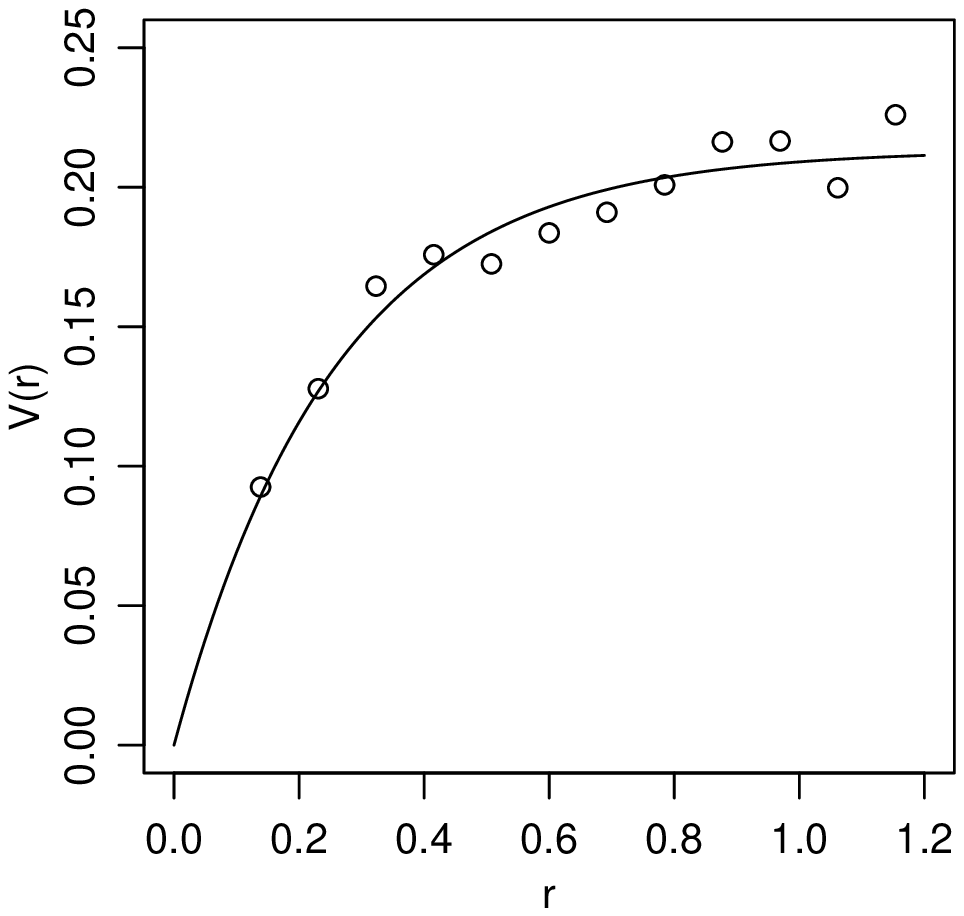, width=7cm,height=6cm} 
\end{center}
\caption{Left: Sampling locations of the lead concentration data in year 2000 and response bubble plot. 
Right: empirical semivariogram of the log--transformed lead concentration data and its least squares fit.}
\label{fi:galicia-locs}
\end{figure}

Figure \ref{fi:galicia-prior-post-int-lik-nu} (left) displays the normalized exact and approximate 
reference priors of $\vartheta$, $\pi^{\rm R}(\vartheta)$ and $\pi^{\rm AR}(\vartheta)$, 
as well as their corresponding marginal posteriors. Both posterior distributions are quite close.
Table \ref{ta:solar-exp} reports the Bayesian estimators of the model parameters and their 
corresponding $95\%$ highest posterior density (HPD) credible intervals based on both posteriors, 
showing that both inferences are essentially the same, as expected from the findings in 
Figure \ref{fi:galicia-prior-post-int-lik-nu} (left). 
The modes of the two posteriors of $\vartheta$ are almost indistinguishable, and 
the estimates of $\beta_1$ and $\sigma^2$ are also very close. 
The analyses suggest that the approximate reference posterior has slightly lighter tails than 
the exact reference posterior, and as a result, the credible intervals from the former are slightly narrower.

\begin{figure}[t!]
\begin{center}
\psfig{figure=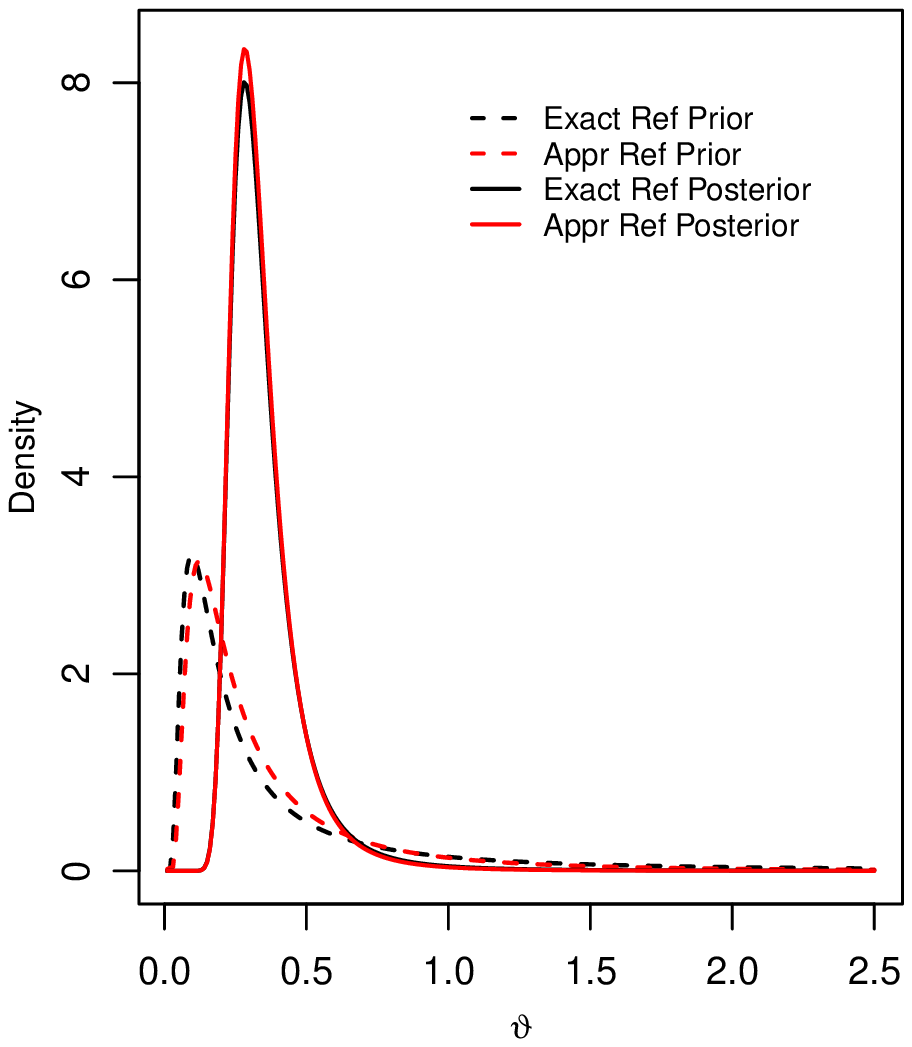, width=7cm,height=6cm} 
\psfig{figure=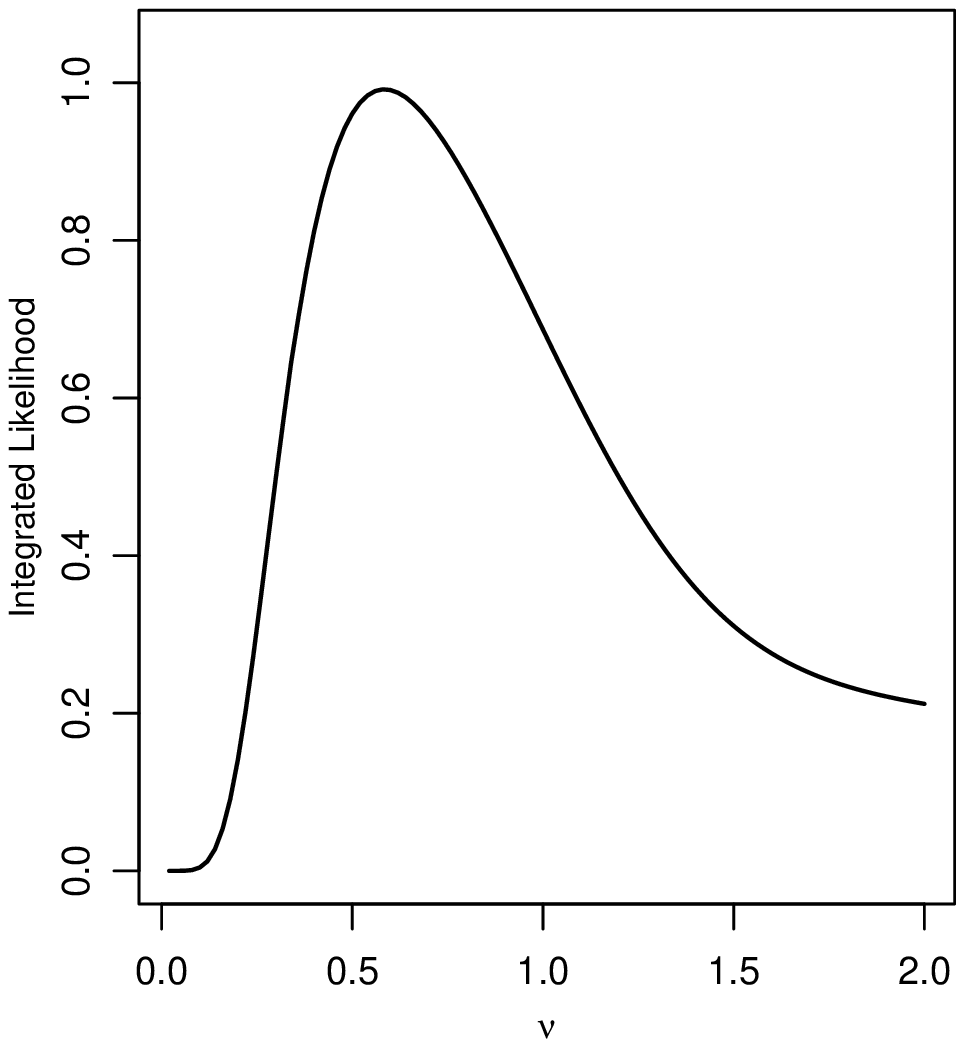, width=7cm,height=6cm}
\end{center}
\caption{Left: Densities of exact and approximate marginal reference priors and 
posteriors of $\vartheta$ for the lead concentration data.
Right: Integrated likelihood of $\nu$ for the lead concentration data.}
\label{fi:galicia-prior-post-int-lik-nu}
\end{figure}

\begin{table}[b]
\begin{center}
\caption{Parameter estimates from the lead concentration data using exact and approximate reference priors. 
The estimate $\hat{\vartheta}$ is the posterior mode,  
$\hat{\sigma}^2$ is the posterior median and $\hat{\beta}_1$ is the posterior mean. 
The 95\% credible intervals are the HPD.}
\label{ta:solar-exp}
\begin{tabular}{lccc}
\toprule
\multirow{2}{*}{Prior} & $\hat{\beta}_1$ & $\hat{\sigma}^2$ & $\hat{\vartheta}$ \\ 
 &  $(95\%~{\rm CI})$  & $(95\%~{\rm CI})$   & $(95\%~{\rm CI})$ \\
 \midrule
\multirow{2}{*}{Exact Reference Prior} 
& $0.734$ & $0.233$ & $0.283$  \\ 
 & $(0.462, 1.013)$ & $(0.133, 0.368)$ & $(0.168, 0.613)$ \\ \midrule
\multirow{2}{*}{Approximate Reference Prior} & $0.732$   & $0.228$ &  $0.283$ \\ 
 & $(0.465, 0.999)$ &  $(0.135, 0.359)$ & $(0.177, 0.603)$ \\ \bottomrule
\end{tabular}
\end{center}
\end{table}

Note that the evaluation of the exponential covariance function and its derivative w.r.t.
$\vartheta$ is devoid of Bessel functions.  
The computation time to draw $10^4$ posterior samples based on the 
approximate reference prior was about $78$ seconds, while the time to do the same task 
based on exact reference prior was $122$ seconds.
In both the exact likelihood was used so the time difference is due to prior evaluations.
For a Mat\'ern model with $\nu \neq m + 1/2$, where $m$ is a non--negative integer,
the evaluation of the covariance function and its derivative w.r.t. $\vartheta$ involve
Bessel functions, and in this case the computation time to draw a posterior sample 
of the same size jumps to $277$ seconds.

It is worth pointing out that the selection of the family of covariance functions is in general 
a difficult problem, and this is so in particular for the selection of the smoothness of 
the covariance family. 
Graphical summaries such as the one reported in 
Figure \ref{fi:galicia-locs} (right) are often used to aid in this task, but they are of limited value 
due to the lack of measurements separated by small distances. 
\cite{Berger2001} suggested choosing the smoothness by inspecting its integrated likelihood.
When the smoothness parameter is not assumed known, the (conditional) approximate reference prior 
is written as
\[
\pi^{\rm AR}(\bfbeta, \sigma^2, \vartheta ~|~\nu) = \frac{C(\nu)\pi^{\rm AR}(\vartheta ~|~\nu)}{\sigma^2} ,
\]
where $\pi^{\rm AR}(\vartheta ~|~\nu)$ is given in (\ref{eq:approx-prior-covar}), with the dependence 
on the smoothness parameter now being explicit, and
$C(\nu) \coloneqq \big( \int_{0}^{\infty} \pi^{\rm AR}(\vartheta ~|~\nu) d \vartheta\big)^{-1}$ is the
normalizing constant.
If $\bfvartheta = (\vartheta, \nu)$ denotes the correlation parameters of the Mat\'ern model,
then the integrated likelihood of $\nu$ is given by
\begin{align*}
m(\boldsymbol{z} ~|~\nu) &= \int_{{\mathbb R}^p \times (0,\infty)^2}
L(\bfbeta,\sigma^2,\vartheta,\nu ; \boldsymbol{z}) \pi^{\rm AR}(\bfbeta, \sigma^2, \vartheta ~|~\nu)	
d \bfbeta d\sigma^2 d\vartheta \\
& \propto \int_{0}^{\infty} |\boldsymbol{\Sigma}_{\bfvartheta}|^{-\frac{1}{2}} 
|\boldsymbol{X}^{\top}\boldsymbol{\Sigma}_{\bfvartheta}^{-1}\boldsymbol{X}|^{- \frac{1}{2}}
(S^{2}_{\bfvartheta})^{-\frac{n-p}{2}}
C(\nu)\pi^{\rm AR}(\vartheta ~|~\nu) d \vartheta .
\end{align*}
Note that this integrated likelihood is not well defined when the exact reference prior 
$\pi^{\rm R}(\vartheta ~|~\nu)$ is used, since $C(\nu)$ does not exist for some $\nu \geq 1$,
while it is well defined for the approximate reference prior; see Section \ref{sec:appr-ref-prior}.
The smoothness parameter can now be chosen as the value that maximizes $m(\boldsymbol{z} ~|~\nu)$.
Figure \ref{fi:galicia-prior-post-int-lik-nu} (right) displays the integrated likelihood of $\nu$ 
for the Galicia lead concentration data, showing that the choice $\nu=0.5$ was about right 
(the maximum occurs at $\nu=0.58$).

The Supplementary Materials report the results of data analysis for a (simulated) 
data set with different features from those of the above lead concentration data:
an irregular sampling design of larger size and data from a  smoother model.
The results of the exact and approximate reference analyses were also in this case practically equivalent.

\section{Conclusions}  \label{sec:conclusions}

This work has derived and studied approximate reference priors for a class of geostatistical models, 
namely for isotropic Gaussian random fields whose covariance functions depend on unknown 
variance and range parameters.
The methodology relies on a spectral approximation to the integrated likelihood of the 
covariance parameters, which produces close approximations to exact reference priors 
for a variety of sampling designs and model features.

The approximate reference priors derived in this work have a number of beneficial features 
that make them attractive for practical use.
First, they can be evaluated in a fraction of the time required to evaluate exact reference priors
because they do not involve inversion of large or ill--conditioned matrices nor 
the evaluation of special functions. 
For random fields with constant mean, the approximate reference prior has a simple matrix--free expression.
Second, for many families of correlation functions, including the widely used Mat\'ern family,
the approximate marginal reference prior for the correlation parameter is proper, which is not
always the case for models with smooth covariance function 
in the reference prior. 
This enables the use of approximate reference priors for covariance function selection 
using Bayes factors, as described in \cite{Berger2001}. 
This is a very helpful property since very few tools are available for this purpose, and 
covariance selection is often done casually.
Finally, results from simulation experiments reported in the Supplementary Materials 
show that inferences based on these approximate reference posteriors have satisfactory frequentist properties that are as good as those based on exact reference posteriors, and sometimes better than those based on purely likelihood--based inferences.

The proposed approximate reference prior depends on an auxiliary regular grid set up by the user.
Tuning this grid allows the attainment of close approximations to exact reference priors for 
a variety of sampling designs and model features. 
The numerical studies in Section \ref{sec:numerical studies} provide useful guidelines for 
the setting of the tuning constants, and a default way for their determination will be investigated in future work. 
The approximation can be computed for random field models with explicit (normalized) 
spectral density that have the general form stated in Section \ref{sec:appr-ref-prior}.
%functions that are positive, differentiable w.r.t. the correlation parameter, and 
The isotropic Mat\'ern family of correlation functions was used for illustration, but the 
methodology is equally applicable for other families, such as some families in \cite{Vecchia1985}, 
\cite{Jones1993} and \cite{Laga2017}, possibly after a reparametrization and once 
some of their parameters are fixed.

It should be noted that the approximate reference priors derived here do not seem 
to converge to their exact counterparts in any strict mathematical sense.
One reason is that exact reference priors depend on the sample size $n$, while approximate reference priors 
depend on $\Delta$, $M_1$ and $M_2$ that are in principle unrelated to $n$.
Another reason is that, although it holds that 
${\rm cov}\{T_{M_1, M_2}({\bf u}_{i,j}), T_{M_1, M_2}({\bf u}_{i',j'})\} \rightarrow 
\sigma^2 K_{\vartheta}(||{\bf u}_{i,j} - {\bf u}_{i',j'}||)$ as $\min\{M_1,M_2\} \to \infty$, 
a comparable result in the spectral domain may not hold due to the aliasing effect ($\Delta > 0$ is fixed).
Nevertheless, approximate reference priors provide useful working approximations, since they share 
the main properties of exact reference priors, and can be computed much faster in situations where 
the latter cannot.

The proposed methodology could be extended to models with more complex correlation functions.
One possible extension is to approximate the reference priors derived in 
\cite{Kazianka2012} and \cite{Ren2012} for isotropic correlation functions with unknown 
range and nugget parameters, which would describe situations when geostatistical data contain 
measurement error.
Another possible extension is to non--isotropic separable correlation functions that 
depend on several range parameters, which are commonly used in the analysis of data from 
computer experiments \citep{Paulo2005}.
As pointed out by a reviewer, when the number of range parameters is large, 
the spectral approximation to the random field may not be as accurate for these models compared to the 
isotropic models, or may require a much larger $M$ to achieve a good approximation. 
For these models, \cite{Gu2019} proposed an approximation to the (joint) reference prior of  
the range parameters aimed at matching the tail behaviours of their exact counterparts. 
A benefit of the latter approximation is that it can perform input selection, 
in the sense of identifying `inert inputs'.
These extensions are currently being developed and will be reported elsewhere.

Finally, in recent years several methods have been proposed in the literature to approximate Gaussian 
likelihoods that include, but are not limited to, spectral approximations \citep{Paciorek2007}, 
composite likelihood approximations \citep{Varin2011},
low--rank approximations \citep{Heaton2018} and Vecchia approximations
\citep{Katzfuss2021}.
The combination of one of these with the approximate reference prior developed in this work 
would make it feasible to carry out default Bayesian analyses of large geostatistical data sets. 
This will be explored in future work.

%\section*{Acknowledgment}
%We want to express our deep appreciation for the Associate Editor and two anonymous reviewers for their 
%comments and suggestions, which lead to a much improved article in terms of clarity and depth. This work is partially supported by ``the Fundamental Research Funds for the Central Universities, China'' in UIBE (CXTD11-05). 

%\section*{Supporting information}
%Supplementary materials for this article is available online including:
%S1: Proof of Lemma 1. S2: Details of the construction of the matrix $\boldsymbol{H}_1$. S3: Numerical comparison of the tail behaviours of exact and approximate reference priors of $\vartheta$. S4: Sensitivity of the approximate reference prior to the tuning constants. S5: Comparison of frequentist properties of Bayesian inferences based on several default priors and MLE. S6: Analysis of a data set simulated using an irregular sampling design. S7: Additional References.

%\bibliography{priorgpbib.bib}  

\bigskip
\bigskip

%\noindent{\bf Corresponding Author}: 
%Zifei Han. School of Statistics, University of International Business and Economics, China. 
%{\tt Email: zifeihan@uibe.edu.cn}

\section*{Appendix: Proofs of the Main Results in Section \ref{sec:appr-ref-prior}}

\noindent To prove the main results in Section \ref{sec:appr-ref-prior}, we use the following lemmas. 

\begin{lemma}\label{th:lemma-inequality}
Let $\{x_i\}_{i = 1}^{\infty}$ and $\{y_i\}_{i = 1}^{\infty}$ be sequences of
positive real numbers. 
If $\max\limits_{i} \left\{{x_i}/{y_i} \right\} = u < \infty$, then 
$\left(\sum_{i=1}^{\infty} x_i\right)/\left(\sum_{i=1}^{\infty} y_i\right) \leq u$.
\end{lemma}
\begin{proof}
From the assumption follows that $x_i \leq u y_i$ for all $i$, so summing over all $i$
provides the result.
\end{proof}

\begin{lemma}
\label{th:lemma-Chipman}
Let $\boldsymbol{A}$ be an $n \times p$ real--valued matrix $(n > p)$ 
with rank $p-m > 0$ and $m > 0$, and 
$\boldsymbol{B}$ an $m \times p$ real--valued matrix with rank $m$ 
whose rows are linearly independent of the rows in  $\boldsymbol{A}$. 
Then 
$\boldsymbol{B}(\boldsymbol{A}^{\top} \boldsymbol{A} + 
\boldsymbol{B}^{\top} \boldsymbol{B})^{-1} \boldsymbol{B}^{\top} = 
\boldsymbol{I}_{m}$
\end{lemma}
\begin{proof}
See \citet[Theorem 1.1]{Chipman1964}
or \citet[Lemma 1.1]{Mohammadi2016}. 
\end{proof}

\begin{lemma}
\label{th:lemma-hatmatrix}
Let $\boldsymbol{\Lambda}$ be a diagonal matrix with positive diagonal entries,
$\boldsymbol{X}$ an $n \times p$ real--valued matrix with rank $p$ and $n > p$,
and $\boldsymbol{P} = \boldsymbol{X} 
(\boldsymbol{X}^{\top} \boldsymbol{\Lambda}^{-1} \boldsymbol{X})^{-1}
\boldsymbol{X}^{\top} \boldsymbol{\Lambda}^{-1}$
a projection matrix.	
If $\boldsymbol{P}_{ij}$ and $\boldsymbol{X}_{ij}$ are the $(i, j)^{th}$ entries of 
$\boldsymbol{P}$ and $\boldsymbol{X}$, respectively, then 
 
\medskip
 
(1) $\boldsymbol{P}_{ij} \boldsymbol{P}_{ji} \geq 0$
for all $i, j  = 1, \ldots, n$. 

(2) The diagonal elements of $\boldsymbol{P}$ 
satisfies $0 \leq \boldsymbol{P}_{ii} \leq 1$
for all $i = 1, \ldots, n$. 

(3) Let $\boldsymbol{X}_{(-1)}$ denote the matrix $\boldsymbol{X}$ with its first row removed and
$\boldsymbol{0}_{n-1}$ the zero column vector of length $n-1$.
If $\boldsymbol{0}_{n-1}$ is a column of $\boldsymbol{X}_{(-1)}$ and $\boldsymbol{X}_{11} \neq 0$, 
then $\boldsymbol{P}_{11} = 1$. 
\end{lemma}
\begin{proof}
(1) This follows by noting that 
$\boldsymbol{X} (\boldsymbol{X}^{\top} \boldsymbol{\Lambda}^{-1} \boldsymbol{X})^{-1}\boldsymbol{X}^{\top}$ 
is symmetric and $\boldsymbol{\Lambda}^{-1}$ diagonal matrix with positive diagonal entries. 
(2) $\boldsymbol{P}_{ii} \geq 0$ follows from the fact that
$(\boldsymbol{X}^{\top} \boldsymbol{\Lambda}^{-1} \boldsymbol{X})^{-1}$ is positive definite.
And from $\boldsymbol{P}^2 = \boldsymbol{P}$ and (1) we have that $\boldsymbol{P}_{ii}(1-\boldsymbol{P}_{ii}) \geq 0$,
so the result follows. 
(3) Let $\boldsymbol{X}^{*} = \label{derivations} \boldsymbol{\Lambda}^{-\frac{1}{2}} \boldsymbol{X}$ and 
$\boldsymbol{P}^{*} = \boldsymbol{X}^{*} 
(\boldsymbol{X}^{*\top} \boldsymbol{X}^{*})^{-1} \boldsymbol{X}^{*\top}$, so it holds that
$\boldsymbol{P}_{11} = \boldsymbol{P}^{*}_{11}$. 
From the inequalities 
${\rm rank}(\boldsymbol{X}^{*}) \leq \min \{{\rm rank}(\boldsymbol{\Lambda}^{-\frac{1}{2}}), 
{\rm rank}(\boldsymbol{X})\} \leq p$ and 
${\rm rank}(\boldsymbol{X}^{*}) \geq {\rm rank}(\boldsymbol{\Lambda}^{-\frac{1}{2}}) + 
{\rm rank}(\boldsymbol{X}) - n \geq p$ follow that 
${\rm rank}(\boldsymbol{X}^{*}) = p$.  
When $\boldsymbol{0}_{n-1}$ is a column of $\boldsymbol{X}_{(-1)}$, 
it is also a column of $\boldsymbol{X}^{*}_{(-1)}$, 
so ${\rm rank}(\boldsymbol{X}^{*}_{(-1)}) = p-1$.  
Let $\boldsymbol{x}^{*}_{1}$ be the first row of $\boldsymbol{X}^{*}$, 
which is clearly linearly independent from the rows of $\boldsymbol{X}^{*}_{(-1)}$. 
Then by Lemma \ref{th:lemma-Chipman}, 
$\boldsymbol{P}_{11} = \boldsymbol{P}_{11}^{*} = 
\boldsymbol{x}^{*}_1 (\boldsymbol{X}_{(-1)}^{*\top} \boldsymbol{X}_{(-1)}^{*} + 
\boldsymbol{x}^{*\top}_1 \boldsymbol{x}^{*}_1)^{-1} \boldsymbol{x}^{*\top}_1 = 1$.
\end{proof}

\begin{proof}[\textbf{Proof of Corollary \ref{cor:approx-prior-rho-nocov}}]
It follows from (\ref{orthogonality}) and direct calculation that
\begin{align*}
\tilde{\boldsymbol{Q}}_{\vartheta} &=
\frac{1}{c_{\Delta}} {\rm diag}\Big(
	\big(1/\tilde{f}^{\Delta}_{\vartheta}(\bfomega_{m_1,m_2}) : (m_1,m_2) \in I_{\rm C}\big)^{\top}, \
\big(1/\tilde{f}^{\Delta}_{\vartheta}(\bfomega_{m_1,m_2}) : (m_1,m_2) \in I\big)^{\top}, \\
& \hspace{6.5cm} 
\big(1/\tilde{f}^{\Delta}_{\vartheta}(\bfomega_{m_1,m_2}) : (m_1,m_2) \in I\big)^{\top} \Big) \\
& \hspace{2cm} - \
\frac{\sqrt{M}}{c_{\Delta}} \Big(1/\tilde{f}^{\Delta}_{\vartheta}(\bfomega_{0,0}), \; {\bf 0}_{M-1}^{\top}\Big)^{\top}
\!  \cdot \frac{c_{\Delta} \tilde{f}^{\Delta}_{\vartheta}(\bfomega_{0,0})}{M} \cdot
\frac{\sqrt{M}}{c_{\Delta}} \Big(1/\tilde{f}^{\Delta}_{\vartheta}(\bfomega_{0,0}), \; {\bf 0}_{M-1}^{\top}\Big) \\
&= \frac{1}{c_{\Delta}} {\rm diag}\Big(0,
	\big(1/\tilde{f}^{\Delta}_{\vartheta}(\bfomega_{m_1,m_2}) : (m_1,m_2) \in I_{\rm C} - \{(0,0)^\top\}\big)^{\top}, \\
&  \hspace{2cm} \big(1/\tilde{f}^{\Delta}_{\vartheta}(\bfomega_{m_1,m_2}) : (m_1,m_2) \in I\big)^{\top}, \
\big(1/\tilde{f}^{\Delta}_{\vartheta}(\bfomega_{m_1,m_2}) : (m_1,m_2) \in I\big)^{\top} \Big) ,
\end{align*}
and hence
\begin{align*}
\left(\frac{\partial}{\partial \vartheta} \tilde{\boldsymbol{\Lambda}}_{\vartheta}\right) 
\tilde{\boldsymbol{Q}}_{\vartheta} &= 
{\rm diag}\Big(0,  
\Big(\frac{\partial}{\partial \vartheta} \log \tilde{f}^{\Delta}_{\vartheta}(\bfomega_{m_1,m_2}) : 
	(m_1,m_2) \in I_{\rm C} - \{(0,0)^\top\}\Big)^{\top}, \\
& \quad
\Big(\frac{\partial}{\partial \vartheta} \log \tilde{f}^{\Delta}_{\vartheta}(\bfomega_{m_1,m_2}) : 
(m_1,m_2) \in I\Big)^{\top}, \
\Big(\frac{\partial}{\partial \vartheta} \log \tilde{f}^{\Delta}_{\vartheta}(\bfomega_{m_1,m_2}) : 
(m_1,m_2) \in I\Big)^{\top} \Big) .
\end{align*}
Then the first term in (\ref{eq:approx-prior-covar}) becomes
\[
{\rm tr}\left[\left\{ \left(\frac{\partial}{\partial \vartheta} \tilde{\boldsymbol{\Lambda}}_{\vartheta}\right) 
\tilde{\boldsymbol{Q}}_{\vartheta}\right\}^2 \right] = 
\sum_{(m_1,m_2)}^{} 
\left(\frac{\partial }{\partial \vartheta} \log \tilde{f}^{\Delta}_{\vartheta}(\bfomega_{m_1,m_2})\right)^2 ,
\]
where $(m_1,m_2) \in (I_{\rm C} \cup I) - \{(0,0)^\top\}$ and each frequency corresponding to an index in $I$
appears twice in the sum. 
The second term in (\ref{eq:approx-prior-covar}) is computed similarly and the result follows.
\end{proof}

\begin{proof}[\textbf{Proof of Theorem \ref{thm:asymp-behav-approx-prior-rho}}]
(a) 
We first show that $\pi^{\rm AR}(\vartheta)$ is integrable on $(0,\infty)$
when the mean function of $Z(\cdot)$ is constant. 
From (\ref{eq:approx-prior-rho-nocov}) we have in this case that 
$(\pi^{\rm AR}(\vartheta))^2$ is proportional to the sample variance of 
$\big\{\frac{\partial}{\partial \vartheta} \log \tilde{f}^{\Delta}_{\vartheta}(\bfomega_j) \big\}_{j = 1}^{M-1}$,
which can be alternatively written as
\begin{equation}
\label{eq:approx-prior-rho2}
(\pi^{\rm AR}(\vartheta))^2 \ \propto \
\sum\limits_{j = 1}^{M-1}\sum\limits_{j' =1}^{M-1}
\left(
\frac{\partial}{\partial \vartheta} \log \tilde{f}^{\Delta}_{\vartheta}(\bfomega_j)
- \frac{\partial}{\partial \vartheta}  \log \tilde{f}^{\Delta}_{\vartheta}(\bfomega_{j'})\right)^{2} .
\end{equation}
From (\ref{eq:spec-den-alias}), direct calculation shows that for $j=1,\ldots,M-1$
\begin{equation*}
\frac{\partial}{\partial \vartheta} \log \tilde{f}^{\Delta}_{\vartheta}(\bfomega_j) 
= \frac{h'_2(\vartheta)}{h_2(\vartheta)} - a \cdot u'(\vartheta) 
\left(\frac{\sum_{\boldsymbol{l}\in {{\mathcal T}_2}} 
\frac{h_1(\bfomega_{j(l_1, l_2)})}
{\big(\|\bfomega_{j(l_1, l_2)}\|^2 \; + \; u(\vartheta) \big)^{a + 1}}}
{\sum_{\boldsymbol{l}\in {{\mathcal T}_2}} \frac{h_1(\bfomega_{j(l_1, l_2)})}
{\big(\|\bfomega_{j(l_1, l_2)}\|^2 \; + \; u(\vartheta) \big)^{a}}} \right) ,
\label{der-log-spden-ali}
\end{equation*}
where $\boldsymbol{l} = (l_1, l_2)$,  
$\mathcal{T}_2 = [-T, T]^2 \cap {\mathbb Z}^2$ and
$\bfomega_{j(l_1, l_2)} \coloneqq \bfomega_{j} + \frac{2\pi}{\Delta}\boldsymbol{l}$. 
Then for any fixed $\bfomega_j$,
$\frac{\partial}{\partial \vartheta} \log \tilde{f}^{\Delta}_{\vartheta}(\bfomega_j)$
is a continuous function in $(0,\infty)$.
After some expansion and simplification we have that
\begin{align}
& \quad\quad \Big|\frac{\partial}{\partial \vartheta}  \log \tilde{f}^{\Delta}_{\vartheta}(\bfomega_j)
- \frac{\partial}{\partial \vartheta}  \log \tilde{f}^{\Delta}_{\vartheta}(\bfomega_{j'}) \Big| \nonumber \\
& \ \leq \  a \cdot u'(\vartheta) 
\frac{\sum_{\boldsymbol{k} \in {{\mathcal T}_4}} 
\frac{h_1(\bfomega_{j(k_1, k_2)}) h_1(\bfomega_{j'(k_3, k_4)})
\big| \|\bfomega_{j'(k_3, k_4)}\|^2 - \|\bfomega_{j(k_1, k_2)}\|^2 \big|}
{\big(g(\vartheta, \bfomega_{j(k_1, k_2)}, \bfomega_{j'(k_3, k_4)})\big)^{a + 1}}}
{\sum\limits_{\boldsymbol{k} \in {{\mathcal T}_4}} 
\frac{h_1(\bfomega_{j(k_1, k_2)}) h_1(\bfomega_{j'(k_3, k_4)})}
{\big(g(\vartheta, \bfomega_{j(k_1, k_2)}, \bfomega_{j'(k_3, k_4)})\big)^{a}}}
\nonumber \\
& \ \leq \ C \cdot a \cdot u'(\vartheta) 
\frac{\sum_{\boldsymbol{k} \in {{\mathcal T}_4}} 
\frac{h_1(\bfomega_{j(k_1, k_2)}) h_1(\bfomega_{j'(k_3, k_4)})
}{\big(g(\vartheta, \bfomega_{j(k_1, k_2)}, \bfomega_{j'(k_3, k_4)})\big)^{a + 1}}}
{\sum\limits_{\boldsymbol{k} \in {{\mathcal T}_4}} 
\frac{h_1(\bfomega_{j(k_1, k_2)}) h_1(\bfomega_{j'(k_3, k_4)})}
{\big(g(\vartheta, \bfomega_{j(k_1, k_2)}, \bfomega_{j'(k_3, k_4)})\big)^{a}}},
\label{eq:logfj-logfk-rho}
\end{align}
where $\boldsymbol{k} = (k_1, k_2, k_3, k_4)$, 
$\mathcal{T}_4 = [-T, T]^4 \cap {\mathbb Z}^4$, $C > 0$ is a constant and 
\begin{equation*}
\label{eq:g-omega}
g(\vartheta, \bfomega_{j(k_1, k_2)}, \bfomega_{j'(k_3, k_4)}) = 
\left(\|\bfomega_{j(k_1, k_2)}\|^2 + u(\vartheta) \right)
\left(\|\bfomega_{j'(k_3, k_4)}\|^2 +  u(\vartheta) \right).
\end{equation*}
To bound (\ref{eq:logfj-logfk-rho}) 
we note that the maximum of the ratios of the general
terms in the numerator and denominator sums is 
\begin{align*}
\max_{\boldsymbol{k} \in \mathcal{T}_4}
\frac{1}{g(\vartheta, \bfomega_{j(k_1, k_2)}, \bfomega_{j'(k_3, k_4)})} 
& \leq \frac{1} {\left(\|\bfomega_{j}\|^2 + u(\vartheta) \right) 
\left(\|\bfomega_{j'}\|^2 +  u(\vartheta) \right)} \\
& \leq \min\left\{\frac{1}{u^2(\vartheta)}, 
\frac{1}{\|\bfomega_{j}\|^2 \|\bfomega_{j'}\|^2}\right\} ,
\end{align*}
where the first inequality holds because for 
any $(k_1,k_2) \in {\mathcal T}_{2}$,
$\|\bfomega_{j(k_1,k_2)}\| \geq \|\bfomega_{j(0,0)}\| = \|\bfomega_{j}\|$, and the 
second inequality holds because $\|\bfomega_{j}\|^2 > 0$ 
for $j = 1, \ldots, M-1$ and $u(\vartheta) > 0$. 
By Lemma \ref{th:lemma-inequality} we have
\[
\left( \frac{\partial}{\partial \vartheta}  \log \tilde{f}^{\Delta}_{\vartheta}(\bfomega_j)
- \frac{\partial}{\partial \vartheta}  \log \tilde{f}^{\Delta}_{\vartheta}(\bfomega_{j'}) \right)^2
\leq C_{jj'} \min \left\{ (u'(\vartheta))^2, \frac{(u'(\vartheta))^2}{u^4(\vartheta)}\right\},
\]
for some $C_{jj'} > 0$, and replacing this in (\ref{eq:approx-prior-rho2}) we obtain 
(\ref{asymptotic-behaviour=pi=vartheta}).

Now we use the above to prove the result for the case of non--constant mean functions.
Recall 
$\tilde{\boldsymbol{X}}$ is the $M \times p$ matrix whose entries involve the covariates 
measured at the locations in $\mathcal{U}_M$. 
Since $\boldsymbol{H}_1^{\top} {\bf 1}_{M} = (M, \; {\bf 0}_{M-1}^{\top})^{\top}$, 
$\boldsymbol{H}_1^{\top}\boldsymbol{H}_1$ is an $M \times M$ diagonal matrix with 
its first diagonal element $M$, and clearly 
$\boldsymbol{L}_1^{\top}{\bf 1}_{M} = (\sqrt{M}, \; {\bf 0}_{M-1}^{\top})^{\top}$. 
Since the first column of $\tilde{\boldsymbol{X}}$ is ${\bf 1}_{M}$, the first column of 
$\boldsymbol{X}_1$ is $(\sqrt{M}, \; {\bf 0}_{M-1}^{\top})^{\top}$. 
Let $\boldsymbol{P}_{\vartheta} = \boldsymbol{X}_1 (\boldsymbol{X}_1^{\top}  
\tilde{\boldsymbol{\Lambda}}_{\vartheta}^{-1} \boldsymbol{X}_1)^{-1}\boldsymbol{X}_1^{\top} 
\tilde{\boldsymbol{\Lambda}}^{-1}_{\vartheta}$, 
which is the projection matrix under a weighted least square  setting, and let $\boldsymbol{P}_{ij}$ denote 
the $(i, j)^{th}$ element in $\boldsymbol{P}_{\vartheta}$, where the dependence on $\vartheta$ is suppressed 
to simplify the notation. 
Recall that $\tilde{\boldsymbol{Q}}_{\vartheta} = 
\tilde{\boldsymbol{\Lambda}}_{\vartheta}^{-1} (\boldsymbol{I}_{M} - \boldsymbol{P}_{\vartheta})$, and let  
\[
\boldsymbol{\Psi}_{\vartheta} :=  
\left(\frac{\partial}{\partial \vartheta} \tilde{\boldsymbol{\Lambda}}_{\vartheta}\right) \tilde{\boldsymbol{Q}}_{\vartheta} 
= {\rm diag}(\boldsymbol{\gamma}_{\vartheta}) (\boldsymbol{I}_{M} - \boldsymbol{P}_{\vartheta}), 
\]
where $\boldsymbol{\gamma}_{\vartheta}$ is a length $M$ 
vector of the diagonal elements in   
$(\partial/\partial \vartheta)\log \tilde{\boldsymbol{\Lambda}}_{\vartheta}$,  
and $\gamma_i$ is the $i^{th}$ component in $\boldsymbol{\gamma}_{\vartheta}$. Then 
\[
{\rm tr}(\boldsymbol{\Psi}_{\vartheta}) = \sum_{i = 1}^{M} \gamma_i (1-\boldsymbol{P}_{ii}) ,
\]
and
\begin{align*}
{\rm tr}(\boldsymbol{\Psi}_{\vartheta}^2) &= {\bf 1}_{M}^{\top}\left(\boldsymbol{\Psi}_{\vartheta} \circ \boldsymbol{\Psi}_{\vartheta}^{\top}\right){\bf 1}_{M}  
= \sum_{i = 1}^{M} \gamma_i^2 (1-\boldsymbol{P}_{ii})^2 + \sum_{i=1}^{M}\sum_{j \neq i}^{M} \gamma_i \gamma_j \boldsymbol{P}_{ij} \boldsymbol{P}_{ji} \\
& \leq \sum_{i = 1}^{M} \gamma_i^2 \bigg[ (1-\boldsymbol{P}_{ii})^2 + \sum_{j \neq i}^{M}  \boldsymbol{P}_{ij} \boldsymbol{P}_{ji} \bigg] 
 = \sum_{i = 1}^{M} \gamma_i^2  (1-\boldsymbol{P}_{ii}), 
\end{align*}
where $\circ$ is the Hadamard product. 
The inequality holds because 
$\sum_{i=1}^{M}\sum_{j \neq i}^{M} (\gamma_i^2 - \gamma_i \gamma_j )  \boldsymbol{P}_{ij} \boldsymbol{P}_{ji} =
\sum_{i = 1}^{M} \sum_{j > i}^{M} (\gamma_i - \gamma_j)^2 \boldsymbol{P}_{ij} \boldsymbol{P}_{ji}$, which 
is non--negative because of Lemma \ref{th:lemma-hatmatrix}(1), and
the last equality follows from the fact that $\boldsymbol{P}_{\vartheta}$ is idempotent. 
Then the approximate reference prior $\pi^{\rm AR}(\vartheta)$ for non--constant mean case satisfies 
\[
\big(\pi^{\rm AR}(\vartheta)\big)^2 \propto {\rm tr}(\boldsymbol{\Psi}_{\vartheta}^2) 
- \frac{{\rm tr}^2(\boldsymbol{\Psi}_{\vartheta})}{M-p} 
\ \leq \ \sum_{i = 1}^{M} \gamma_i^2  (1-\boldsymbol{P}_{ii}) 
- \frac{\Big(\sum_{i = 1}^{M} \gamma_i (1-\boldsymbol{P}_{ii})\Big)^2}{M-p} .
\]
Let $\pi^{0,\rm AR}(\vartheta) :=  
\left[\sum_{i = 2}^{M} \gamma_{i}^{2} - \frac{\left( \sum_{i = 2}^{M} \gamma_i \right)^2}{M-1}\right]^{\frac{1}{2}}$. 
Note that this expression is proportional to (\ref{eq:approx-prior-rho-nocov}), since it does not involve
$\gamma_{1} = (\partial/\partial \vartheta)\log \tilde{f}_{\vartheta}(\bfomega_{0,0})$, so
from the first part of the proof we have that $\pi^{0,\rm AR}(\vartheta)$ is integrable in $(0,\infty)$.
We now proceed to show that
$\pi^{\rm AR}(\vartheta) \leq C \pi^{0,\rm AR}(\vartheta)$ for some positive constant $C$. 
According to Lemma \ref{th:lemma-hatmatrix}(3),
$\boldsymbol{P}_{11} = 1$. It is sufficient to show
\begin{equation}
\label{eq:proof-condition}
\sum_{i = 2}^{M} \gamma_{i}^{2} - \frac{\left( \sum_{i = 2}^{M} \gamma_i \right)^2}{M-1} - 
\sum_{i = 2}^{M} \gamma_i^2  (1-\boldsymbol{P}_{ii}) + \frac{\left(\sum_{i = 2}^{M} \gamma_i (1-\boldsymbol{P}_{ii})\right)^2}{M-p} \geq 0. 
\end{equation}
Based on the standard properties of the projection matrix and Lemma \ref{th:lemma-hatmatrix}(3), 
${\rm tr}(\boldsymbol{P}_{\vartheta}) = p$ and $\sum_{i = 2}^{M} \boldsymbol{P}_{ii} = p-1$. 
Since Lemma \ref{th:lemma-hatmatrix}(2) guarantees $P_{ii} \geq 0$ ($i = 2, \ldots, M$),
Cauchy--Schwartz inequality is applicable, which results in 
\begin{equation}
\label{eq:cauchy-schwartz}
\sum_{i = 2}^{M} \gamma_{i}^{2}\boldsymbol{P}_{ii}  =   \frac{ \big(\sum_{i=2}^{M}\boldsymbol{P}_{ii}\big) \big(\sum_{i=2}^{M}\boldsymbol{P}_{ii} \gamma_{i}^2\big) }{p-1} 
 \geq \frac{ \big( \sum_{i =2}^{M} \boldsymbol{P}_{ii} \gamma_i\big)^2}{p-1}.
\end{equation}
Furthermore, applying the Sedrakyan's inequality 
\citep{Sedrakyan1997} we have 
\begin{equation}
\label{eq:sedrakyan}
\frac{ \big( \sum_{i =2}^{M} \boldsymbol{P}_{ii} \gamma_i\big)^2}{p-1} + 
\frac{\left(\sum_{i = 2}^{M} \gamma_i (1-\boldsymbol{P}_{ii})\right)^2}{M-p} -
\frac{ \big( \sum_{i =2}^{M} \gamma_i\big)^2}{M-1} \geq 0.
\end{equation}
Plugging (\ref{eq:cauchy-schwartz}) into (\ref{eq:proof-condition}),  then using (\ref{eq:sedrakyan}), 
it is easy to verify the condition (\ref{eq:proof-condition}) holds. 
Therefore, there exists a constant $C > 0$ such that 
$\pi^{\rm AR}(\vartheta) \leq C \pi^{0,AR}(\vartheta)$.
As a result, $\pi^{\rm AR}(\vartheta)$ is also proper and has the same limiting behaviour as 
the approximate reference prior for the constant--mean case. 

\medskip

\noindent
(b) 
The integrated likelihood (\ref{eq:intlik-1}) is clearly a continuous function on $(0,\infty)$.
Assumption (A1) implies that $\lim_{\vartheta \rightarrow 0^+} K_{\vartheta}(r) = {\bf 1}\{r = 0\}$,
so as $\vartheta \rightarrow 0^+$, 
$\boldsymbol{\Sigma}_{\vartheta} \rightarrow \boldsymbol{I}_n$ and  
$L^{\rm I}(\vartheta; \boldsymbol{z}) \rightarrow 
|\boldsymbol{X}^{\top}\boldsymbol{X}|^{- \frac{1}{2}} (S^{2}_{0})^{-\frac{n-p}{2}} > 0$, where
$S^{2}_{0} = (\boldsymbol{z} - \boldsymbol{X} \hat{\bfbeta}_0)^{\top} 
(\boldsymbol{z} - \boldsymbol{X} \hat{\bfbeta}_0)$  and
$\hat{\bfbeta}_0 = (\boldsymbol{X}^{\top} \boldsymbol{X})^{-1} \boldsymbol{X}^{\top} \boldsymbol{z}$. 
\cite{Berger2001} showed that when the mean function includes an intercept and the 
correlation function satisfies (\ref{sigma-expansion}) (non--smooth covariance models), 
$L^{\rm I}(\vartheta; \boldsymbol{z}) = O(1)$ as $\vartheta \rightarrow \infty$.
For smooth covariance  models, \cite{Mure2021} showed that if 
$v_{1}(\vartheta) \geq \ldots  \geq v_{n-p}(\vartheta) > 0$ are the ordered eigenvalues of 
$\boldsymbol{\Sigma}^{W}_{\vartheta}$ and assumption (A2) holds, then
\begin{equation*}
L^{\rm I}(\vartheta; \boldsymbol{z}) = 
\left\{ \prod_{i=1}^{n-p} \frac{O(v_{n-p}(\vartheta))}{v_{i}(\vartheta)} \right\}^{\frac{1}{2}} ,
\quad\quad {\rm as} \ \vartheta \rightarrow \infty .
%\label{int-prof-lik-vartheta-inf}
\end{equation*}
So $L^{\rm I}(\vartheta; \boldsymbol{z})$ is bounded on $(0,\infty)$ in both cases. 
Combining this with the result in (a) imply that the integral (\ref{integrability}) is finite
when $\pi^{\rm R}(\vartheta)$ is replaced with $\pi^{\rm AR}(\vartheta)$, and therefore
$\pi^{\rm AR}(\bfbeta, \sigma^2, \vartheta ~|~\boldsymbol{z})$ is proper.
\end{proof}

\newpage 

\setcounter{table}{0}
\setcounter{figure}{0}
\setcounter{equation}{0}

\renewcommand{\thefigure}{S\arabic{figure}}
\renewcommand{\thetable}{S\arabic{table}}
\renewcommand{\thetable}{S\arabic{table}}
\renewcommand{\theequation}{S\arabic{equation}}

\begin{center}
\Large{\textbf{Supplementary Materials for the Manuscript: \\ 
Approximate Reference Priors for Gaussian Random Fields \\
by Victor De Oliveira and Zifei Han }}
\end{center}

\medskip

This document provides the proofs of results stated in the manuscript indicated in the title,
as well as some additional numerical results. It consists of seven parts:
\begin{enumerate}
\item[S1]
Proof of Lemma 1.

\item[S2]
Details of the construction of the matrix $\boldsymbol{H}_1$.

\item[S3]
 Numerical comparison of the tail behaviours of exact and approximate 
reference priors of $\vartheta$.

\item[S4]
Sensitivity of the approximate reference prior to the tuning constants.

\item[S5]
Comparison of frequentist properties of Bayesian inferences based on
		several default priors and MLE.

\item[S6] Analysis of a data set simulated on an irregular sampling design.

\item[S7] Additional references.
\end{enumerate}

\noindent{\bf S1. Proof of Lemma \ref{thm:spec-approx-2d}} 

\medskip

\noindent
(a) 
To guide the finding of conditions on $A_{m_1,m_2}$ and $B_{m_1,m_2}$ for $T_{M_1, M_2}(\bfu_{i, j})$ 
to be real--valued, the terms of the sum are grouped based on the spatial frequencies as
\begin{equation}
T_{a}(\bfu_{i, j}) + T_{b}(\bfu_{i, j}) + T_{c}(\bfu_{i, j}) + T_{d}(\bfu_{i, j}) 
+ T _{e}(\bfu_{i, j}) + T_{f}(\bfu_{i, j}) + T_{g}(\bfu_{i, j}) ,
\label{T-split}
\end{equation}
where 
\begin{align*}
T_{a}(\bfu_{i, j}) & =  
\sum_{m_1 = 1}^{\frac{M_1}{2}-1} \exp\big( \mi \bfomega_{m_1, 0}^{\top} \bfu_{i,j} \big) U_{m_1, 0} \ + 
\sum_{m_1 = -\frac{M_1}{2}+1}^{-1} \exp\big( \mi \bfomega_{m_1, 0}^{\top} \bfu_{i,j} \big) U_{m_1, 0}
\\
T_{b}(\bfu_{i, j})    & =  
  \sum_{m_2 = 1}^{\frac{M_2}{2}-1} \exp\big( \mi \bfomega_{0, m_2}^{\top} \bfu_{i,j} \big) U_{0, m_2} \ +
  \sum_{m_2 = -\frac{M_2}{2}+1}^{-1} \exp\big( \mi \bfomega_{0, m_2}^{\top} \bfu_{i,j} \big) U_{0, m_2}
 \\
T_{c}(\bfu_{i, j})   & =  
  \sum_{m_1 = 1}^{\frac{M_1}{2}-1}   \sum_{m_2 = 1}^{\frac{M_2}{2}-1} 
\exp\big( \mi \bfomega_{m_1, m_2}^{\top} \bfu_{i,j} \big) U_{m_1, m_2} \ +
\sum_{m_1 = -\frac{M_1}{2}+1}^{-1} \sum_{m_2 = -\frac{M_2}{2}+1}^{-1}  
 \exp\big( \mi \bfomega_{m_1, m_2}^{\top} \bfu_{i,j} \big) U_{m_1, m_2} 
 \\
T_{d}(\bfu_{i, j})   & =  \sum_{m_1 = 1}^{\frac{M_1}{2}-1} \sum_{m_2 = -\frac{M_2}{2}+1}^{-1}  
 \exp\big( \mi \bfomega_{m_1, m_2}^{\top} \bfu_{i,j} \big) U_{m_1, m_2} \ + 
  \sum_{m_1 = -\frac{M_1}{2}+1}^{-1}   \sum_{m_2 = 1}^{\frac{M_2}{2}-1} 
\exp\big( \mi \bfomega_{m_1, m_2}^{\top} \bfu_{i,j} \big) U_{m_1, m_2}  
\\
T_{e}(\bfu_{i, j})   & =   
\sum_{m_1 = 1}^{\frac{M_1}{2}-1} \exp\big( \mi \bfomega_{m_1, \frac{M_2}{2}}^{\top} \bfu_{i,j} \big) U_{m_1, \frac{M_2}{2}} 
\ + 
\sum_{m_1 = -\frac{M_1}{2}+1}^{-1} \exp\big( \mi \bfomega_{m_1, \frac{M_2}{2}}^{\top} \bfu_{i,j} \big) U_{m_1, \frac{M_2}{2}} 
  \\
T_{f}(\bfu_{i, j})   &=   
  \sum_{m_2 = 1}^{\frac{M_2}{2}-1} \exp\big( \mi \bfomega_{\frac{M_1}{2}, m_2}^{\top} \bfu_{i,j} \big) U_{\frac{M_1}{2}, m_2} \ + 
  \sum_{m_2 = -\frac{M_2}{2}+1}^{-1} \exp\big( \mi \bfomega_{\frac{M_1}{2}, m_2}^{\top} \bfu_{i,j} \big) U_{\frac{M_1}{2}, m_2}
  \\
T_{g}(\bfu_{i, j})   & =   U_{0,0}+  
\exp\big( \mi \bfomega_{\frac{M_1}{2}, 0}^{\top} \bfu_{i,j} \big) U_{\frac{M_1}{2}, 0}+  
\exp\big( \mi \bfomega_{0, \frac{M_2}{2}}^{\top} \bfu_{i,j} \big) U_{0, \frac{M_2}{2}}+  
\exp\big( \mi \bfomega_{\frac{M_1}{2}, \frac{M_2}{2}}^{\top} \bfu_{i,j} \big) U_{\frac{M_1}{2}, \frac{M_2}{2}} ; 
\end{align*}
the dependence of $T_{a}(\bfu_{i, j}), \ldots, T_{g}(\bfu_{i, j})$ on $M_1, M_2$ is suppressed to simplify notation. 
A graphical illustration of the groupings when $M_1 = M_2 = 6$ is shown in Figure \ref{fi:spec-design-group}. 
For instance, the terms in $T_{M_1, M_2}(\bfu_{i, j})$ involving the spectral points on the $x$ axis 
in the second quadrant labeled as `a' are matched with the terms involving the spectral points 
on the $x$ axis in the first quadrant also labeled as `a'. 
The same goes for all the other terms involving the spectral points in Figure \ref{fi:spec-design-group}, 
where the matched frequencies are letter--coded.

Now, note that for $m_1 = 1,\ldots,M_1/2 - 1$ it holds that $\bfomega_{-m_1, 0} = -\bfomega_{m_1, 0}$  
so if we set $U_{-m_1, 0} \coloneqq \bar{U}_{m_1, 0}$, the term $T_{a}(\bfu_{i, j})$ can be written as
%Each term $T_{a}$, $T_{b}$, $T_{c}$ and $T_{d}$ is made out of two sums, where for any frequency $\bfomega$ 
%that appears in the first sum, the frequency $-\bfomega$ appears in the second sum. 
%Combining the two summations in each $T_{a}$, $T_{b}$, $T_{c}$ and $T_{d}$, 
\begin{align}
T_a(\bfu_{i, j}) & = 
\sum_{m_1 = 1}^{\frac{M_1}{2}-1} \Big(\exp\big( \mi \bfomega_{m_1, 0}^{\top} \bfu_{i,j} \big) U_{m_1, 0} + 
\exp\big(- \mi \bfomega_{m_1, 0}^{\top} \bfu_{i,j} \big) U_{-m_1, 0} \Big) \nonumber \\
& = \sum_{m_1 = 1}^{\frac{M_1}{2}-1} \Big( \exp\big( \mi \bfomega_{m_1, 0}^{\top} \bfu_{i,j}\big) U_{m_1, 0} + 
\compconj{\exp\big( \mi \bfomega_{m_1, 0}^{\top} \bfu_{i,j} \big) U_{m_1, 0}} \Big)  \nonumber \\
& = 2  \sum_{m_1 = 1}^{\frac{M_1}{2}-1}  {\rm Re} \left( \exp\big( \mi \bfomega_{m_1, 0}^{\top} \bfu_{i,j}\big) U_{m_1, 0}\right) \nonumber \\
& = 2 \sum_{m_1 = 1}^{\frac{M_1}{2}-1} 
\left( A_{m_1, 0} \cos\big(\bfomega_{m_1, 0}^{\top} \bfu_{i,j}\big) - 
B_{m_1, 0} \sin (\bfomega_{m_1, 0}^{\top} \bfu_{i,j}) \right), 
\label{Ta}
\end{align}
which is real--valued, where ${\rm Re}(z)$ denotes the real part of the complex number $z$. 
Likewise we have that, if for $m_1 = 1,\ldots,M_1/2 - 1$ and $m_2 = 1,\ldots,M_2/2 - 1$
we set $U_{0,-m_2} \coloneqq \compconj{U}_{0, m_2}$ and $U_{-m_1,-m_2} \coloneqq \compconj{U}_{m_1, m_2}$, 
we obtain that
\begin{align}
T_b(\bfu_{i,j}) &= 2 \sum_{m_2 = 1}^{\frac{M_2}{2}-1} 
\left(A_{0,m_2} \cos\big(\bfomega_{0,m_2}^{\top} \bfu_{i,j}\big) - 
B_{0,m_2} \sin \big(\bfomega_{0, m_2}^{\top} \bfu_{i,j}\big) \right) 
\label{Tb} \\
T_c(\bfu_{i,j}) &= 2 \sum_{m_1 = 1}^{\frac{M_1}{2}-1}  \sum_{m_2 = 1}^{\frac{M_2}{2}-1} 
\left( A_{m_1,m_2} \cos\big(\bfomega_{m_1,m_2}^{\top} \bfu_{i,j}\big) - 
B_{m_1,m_2} \sin \big(\bfomega_{m_1, m_2}^{\top} \bfu_{i,j}\big) \right) , 
\label{Tc}
\end{align}
are also real--valued. 
And if for $m_1 = 1,\ldots,M_1/2 - 1$ and $m_2 = -M_2/2 + 1,\ldots,-1$
we set $U_{-m_1,-m_2} 
\coloneqq \compconj{U}_{m_1, m_2}$, we have
\begin{equation}
T_d(\bfu_{i,j}) = 2 \sum_{m_1 = 1}^{\frac{M_1}{2}-1}  \sum_{m_2 = -\frac{M_2}{2}+1}^{-1}
\left( A_{m_1,m_2} \cos\big(\bfomega_{m_1,m_2}^{\top} \bfu_{i,j}\big) - 
B_{m_1,m_2} \sin \big(\bfomega_{m_1, m_2}^{\top} \bfu_{i,j}\big) \right) . 
\label{Td}
\end{equation} 
By noting that for any $m_1, i, j \in {\mathbb Z}$ and $x \in {\mathbb R}$ it holds that
$\bfomega_{m_1, \frac{M_2}{2}}^{\top} \bfu_{i,j}  = \frac{2\pi m_1 i}{M_1} + \pi j$,
$\cos(x + j\pi) = \cos(x - j\pi)$ and $\sin(x + j\pi) = \sin(x - j\pi)$, we have
\begin{align*}
\exp\big(\mi  \bfomega_{-m_1, \frac{M_2}{2}}^{\top} \bfu_{i,j}  \big) 
& =   \cos\big(\bfomega_{-m_1,\frac{M_2}{2}}^{\top} \bfu_{i,j}\big) + 
\mi \sin \big(\bfomega_{-m_1,\frac{M_2}{2}}^{\top}\bfu_{i,j}\big) \\
& =  \cos\left(-\frac{2\pi m_1 i }{M_1} + \pi j \right) + \mi \sin \left(-\frac{2\pi m_1 i }{M_1} + \pi j\right) \\
& =  \cos\left(-\frac{2\pi m_1 i }{M_1} - \pi j \right) + \mi \sin \left(-\frac{2\pi m_1 i }{M_1} - \pi j\right) \\
& =  \cos\left(\frac{2\pi m_1 i }{M_1} + \pi j \right) - \mi \sin \left(\frac{2\pi m_1 i }{M_1} + \pi j\right) \\
& = \cos\left(\bfomega_{m_1, \frac{M_2}{2}}^{\top} \bfu_{i,j}\right) - \mi \sin \left(\bfomega_{m_1, \frac{M_2}{2}}^{\top} \bfu_{i,j}\right)\\
& = \exp\left( -\mi  \bfomega_{m_1, \frac{M_2}{2}}^{\top} \bfu_{i,j} \right) .
\end{align*}
Then, if for $m_1 = 1,\ldots,M_1/2 - 1$ we set $U_{-m_1, \frac{M_2}{2}} \coloneqq \bar{U}_{m_1, \frac{M_2}{2}}$,
we have 
\begin{align}
T_e(\bfu_{i,j}) & = 
\sum_{m_1 = 1}^{\frac{M_1}{2}-1} \Big( \exp\big( \mi \bfomega_{m_1, \frac{M_2}{2}}^{\top} \bfu_{i,j} \big) 
U_{m_1, \frac{M_2}{2}} + 
\exp\big(\mi \bfomega_{-m_1, \frac{M_2}{2}}^{\top} \bfu_{i,j} \big) U_{-m_1, \frac{M_2}{2}} \Big) \nonumber \\
& = 2 \sum_{m_1 = 1}^{\frac{M_1}{2}-1} 
\left( A_{m_1, \frac{M_2}{2}} \cos\big(\bfomega_{m_1, \frac{M_2}{2}}^{\top} \bfu_{i,j}\big) - 
B_{m_1, \frac{M_2}{2}} \sin \big(\bfomega_{m_1, \frac{M_2}{2}}^{\top} \bfu_{i,j}\big) \right).
\label{Te}
\end{align}
By a similar argument, if for $m_2 = 1,\ldots,M_2/2 - 1$ we set 
$U_{\frac{M_1}{2},-m_2} \coloneqq \bar{U}_{\frac{M_1}{2}, m_2}$,
\begin{equation}
T_f(\bfu_{i,j}) = 2 \sum_{m_2 = 1}^{\frac{M_2}{2}-1} 
\left( A_{ \frac{M_1}{2}, m_2} \cos\big(\bfomega_{ \frac{M_1}{2}, m_2}^{\top} \bfu_{i,j}\big) - 
B_{\frac{M_1}{2}, m_2} \sin \big(\bfomega_{\frac{M_1}{2}, m_2}^{\top} \bfu_{i,j}\big) \right) ,
\label{Tf}
\end{equation}
so the last two terms are also real--valued.
Finally, by setting $B_{0, 0} = B_{\frac{M_1}{2}, 0} = B_{0, \frac{M_2}{2}} = B_{\frac{M_1}{2}, \frac{M_2}{2}} = 0$,
%and noting that $\bfomega_{\frac{M_1}{2}, 0}^{\top} \bfu_{i,j} = i\pi$, 
%$\bfomega_{0, \frac{M_2}{2}}^{\top} \bfu_{i,j}= j\pi$ and 
%$\bfomega_{\frac{M_1}{2}, \frac{M_2}{2}}^{\top} \bfu_{i,j}= (iintroduced+j)\pi$, 
we have 
\begin{equation}
T_{g}(\bfu_{i,j}) = 
A_{0, 0} + \cos\big(\bfomega_{\frac{M_1}{2}, 0}^{\top} \bfu_{i,j}\big)A_{\frac{M_1}{2}, 0} 
+ \cos\big(\bfomega_{0, \frac{M_2}{2}}^{\top} \bfu_{i,j}\big)A_{0, \frac{M_2}{2}}
+ \cos\big(\bfomega_{\frac{M_1}{2}, \frac{M_2}{2}}^{\top} \bfu_{i,j}\big)A_{\frac{M_1}{2}, \frac{M_2}{2}} .
\label{Tg}
\end{equation}
By putting all of the above together we have that the aforementioned restrictions on 
the real and imaginary parts of $U_{m_1, m_2}$ imply that
\begin{align*}
& T_{M_1, M_2} (\bfu_{i,j}) \\
&= 2 \Bigg\{ \sum_{m_1 = 1}^{\frac{M_1}{2}-1} 
\left( A_{m_1, 0} \cos\big(\bfomega_{m_1, 0}^{\top} \bfu_{i,j}\big) - 
B_{m_1, 0} \sin (\bfomega_{m_1, 0}^{\top} \bfu_{i,j}) \right) \\
& \quad + \sum_{m_2 = 1}^{\frac{M_2}{2}-1} 
\left(A_{0,m_2} \cos\big(\bfomega_{0,m_2}^{\top} \bfu_{i,j}\big) - 
B_{0,m_2} \sin \big(\bfomega_{0, m_2}^{\top} \bfu_{i,j}\big) \right) \\
& \quad + \sum_{m_1 = 1}^{\frac{M_1}{2}-1}  \sum_{m_2 = 1}^{\frac{M_2}{2}-1} 
\left( A_{m_1,m_2} \cos\big(\bfomega_{m_1,m_2}^{\top} \bfu_{i,j}\big) - 
B_{m_1,m_2} \sin \big(\bfomega_{m_1, m_2}^{\top} \bfu_{i,j}\big) \right) \\
& \quad + \sum_{m_1 = 1}^{\frac{M_1}{2}-1}  \sum_{m_2 = -\frac{M_2}{2}+1}^{-1}
\left( A_{m_1,m_2} \cos\big(\bfomega_{m_1,m_2}^{\top} \bfu_{i,j}\big) - 
B_{m_1,m_2} \sin \big(\bfomega_{m_1, m_2}^{\top} \bfu_{i,j}\big) \right) \\
& \quad + \sum_{m_1 = 1}^{\frac{M_1}{2}-1} 
\left( A_{m_1, \frac{M_2}{2}} \cos\big(\bfomega_{m_1, \frac{M_2}{2}}^{\top} \bfu_{i,j}\big) - 
B_{m_1, \frac{M_2}{2}} \sin \big(\bfomega_{m_1, \frac{M_2}{2}}^{\top} \bfu_{i,j}\big) \right) \\
& \quad + \sum_{m_2 = 1}^{\frac{M_2}{2}-1} 
\left( A_{ \frac{M_1}{2}, m_2} \cos\big(\bfomega_{ \frac{M_1}{2}, m_2}^{\top} \bfu_{i,j}\big) - 
B_{\frac{M_1}{2}, m_2} \sin \big(\bfomega_{\frac{M_1}{2}, m_2}^{\top} \bfu_{i,j}\big) \right) \Bigg\} \\
& \quad + A_{0, 0} + \cos\big(\bfomega_{\frac{M_1}{2}, 0}^{\top} \bfu_{i,j}\big)A_{\frac{M_1}{2}, 0} 
+ \cos\big(\bfomega_{0, \frac{M_2}{2}}^{\top} \bfu_{i,j}\big)A_{0, \frac{M_2}{2}}
+ \cos\big(\bfomega_{\frac{M_1}{2}, \frac{M_2}{2}}^{\top} \bfu_{i,j}\big)A_{\frac{M_1}{2}, \frac{M_2}{2}} \\
&= A_{0,0} + A_{\frac{M_1}{2},0} \cos\big(\bfomega_{\frac{M_1}{2},0}^{\top} {\bfu}_{i,j}\big)
+ A_{0,\frac{M_2}{2}} \cos\big(\bfomega_{0, \frac{M_2}{2}}^{\top} {\bfu}_{i,j}\big)
+ A_{\frac{M_1}{2},\frac{M_2}{2}} \cos\big(\bfomega_{\frac{M_1}{2}, \frac{M_2}{2}}^{\top} {\bfu}_{i,j}\big) \\
& \quad + \ 2 \sum_{(m_1,m_2) \in I} 
\big( A_{m_1,m_2} \cos(\bfomega_{m_1, m_2}^{\top} {\bfu}_{i,j}) 
- B_{m_1,m_2} \sin(\bfomega_{m_1, m_2}^{\top} {\bfu}_{i,j}) \big) ,
\end{align*}	
is a real--valued random variable.  %which also proves (\ref{eq:spectral-approx2}).
Since for each $\bfu_{i,j} \in \mathcal{U}_M$, $T_{M_1, M_2} (\bfu_{i,j})$ is a linear combination  
of the elements of $\bfg$ in (\ref{g-vector}), which has a zero--mean multivariate normal distribution, 
the result (a) follows.

\bigskip

\noindent
(b) 
For any $\bfu_{i, j}, \bfu_{i', j'} \in \mathcal{U}_{M}$, it follows from (\ref{T-split}) that 
\begin{align}
& {\rm cov}\{T_{M_1, M_2} (\bfu_{i,j}), T_{M_1, M_2} (\bfu_{i',j'})\} \nonumber \\	
&= {\rm cov}\{T_{a} (\bfu_{i,j}), T_{a} (\bfu_{i',j'})\} + {\rm cov}\{T_{b} (\bfu_{i,j}), T_{b} (\bfu_{i',j'})\}  
+ {\rm cov}\{T_{c} (\bfu_{i,j}), T_{c} (\bfu_{i',j'})\} \nonumber \\
& \quad	+ {\rm cov}\{T_{d} (\bfu_{i,j}), T_{d} (\bfu_{i',j'})\} 
+ {\rm cov}\{T_{e} (\bfu_{i,j}), T_{e} (\bfu_{i',j'})\} + {\rm cov}\{T_{f} (\bfu_{i,j}), T_{f} (\bfu_{i',j'})\} 
\nonumber \\
& \quad + {\rm cov}\{T_{g} (\bfu_{i,j}), T_{g} (\bfu_{i',j'})\} ,
\label{cov-T-split}
\end{align}
and from the assumed variances of $A_{m_1,m_2}$ and $B_{m_1,m_2}$, (\ref{Ta}) and 
a standard trigonometric identity we have
\begin{align*}
&{\rm cov} \left\{  T_{a}(\bfu_{i,j}),  T_{a}(\bfu_{i',j'})  \right\} \\  
&= 4\sum_{m_1 = 1}^{\frac{M_1}{2}- 1} \left[
\cos(\bfomega_{m_1, 0}^{\top} \bfu_{i,j} ) \cos(\bfomega_{m_1, 0}^{\top} \bfu_{i',j'} ){\rm var}(A_{m_1, 0}) + 
\sin(\bfomega_{m_1, 0}^{\top} \bfu_{i,j} ) \sin(\bfomega_{m_1, 0}^{\top} \bfu_{i',j'} ){\rm var}(B_{m_1, 0}) 
\right] \\ 
&= 2 \frac{c_{\Delta} \sigma^2}{M}
\sum_{m_1 = 1}^{\frac{M_1}{2}- 1} \cos\big( \bfomega_{m_1, 0}^{\top} \bfh \big)
f^{\Delta}_{\vartheta}(\bfomega_{m_1, 0}),
\end{align*}
where $\bfh \coloneqq \bfu_{i,j} -\bfu_{i',j'} = \Delta(i-i', j-j')^{\top}$; 
the dependence of $\bfh$ on $(i-i', j-j')$ is suppressed to simplify the notation.
By similar computations we have from (\ref{cov-T-split}) and (\ref{Tb})--(\ref{Tg}) that 
\begin{align*}
&{\rm cov} \left\{  T_{M_1,M_2}(\bfu_{i,j}),  T_{M_1,M_2}(\bfu_{i',j'})  \right\} \\
&= 2 \frac{c_{\Delta} \sigma^2}{M} \Bigg\{
\sum_{m_1 = 1}^{\frac{M_1}{2}- 1} \cos\big( \bfomega_{m_1, 0}^{\top} \bfh \big)
f^{\Delta}_{\vartheta}(\bfomega_{m_1, 0}) 
+ \sum_{m_2 = 1}^{\frac{M_2}{2}- 1} \cos\big( \bfomega_{0, m_2}^{\top} \bfh \big)
f^{\Delta}_{\vartheta}(\bfomega_{0, m_2}) \\
& \quad +\sum_{m_1= 1}^{\frac{M_1}{2}- 1}\sum_{m_2 = 1}^{\frac{M_2}{2}- 1} 
\cos \big( \bfomega_{m_1, m_2}^{\top} \bfh \big)
f^{\Delta}_{\vartheta}(\bfomega_{m_1, m_2}) 
+ \sum_{m_1= 1}^{\frac{M_1}{2}- 1}\sum_{m_2 = -\frac{M_2}{2}+1}^{- 1} 
\cos \big( \bfomega_{m_1, m_2}^{\top} \bfh \big)
f^{\Delta}_{\vartheta}(\bfomega_{m_1, m_2}) \\
& \quad + \sum_{m_1 = 1}^{\frac{M_1}{2}-1} 
\cos \big( \bfomega_{m_1, \frac{M_2}{2}}^{\top} \bfh \big)
f^{\Delta}_{\vartheta}(\bfomega_{m_1, \frac{M_2}{2}}) 
+ \sum_{m_2 = 1}^{\frac{M_2}{2}-1} 
\cos\big( \bfomega_{\frac{M_1}{2}, m_2}^{\top} \bfh \big)
f^{\Delta}_{\vartheta}(\bfomega_{\frac{M_1}{2}, m_2})  \\
& \quad + \frac{1}{2} \bigg( 
\cos \big(\bfomega_{0,0}^{\top} \bfh \big) f^{\Delta}_{\vartheta}(0, 0)
+  \cos \big(\bfomega_{\frac{M_1}{2}, 0}^{\top} \bfh \big)
f^{\Delta}_{\vartheta}\left({\frac{\pi}{\Delta}, 0}\right)  
+  \cos \big(\bfomega_{0,\frac{M_2}{2}}^{\top} \bfh \big)
f^{\Delta}_{\vartheta}\left( 0, \frac{\pi}{\Delta} \right)	\\ 
& \quad + \cos \big(\bfomega_{\frac{M_1}{2},\frac{M_2}{2}}^{\top} \bfh \big) 
f^{\Delta}_{\vartheta}\left(\frac{\pi}{\Delta}, \frac{\pi}{\Delta} \right) \bigg) \Bigg\} .
\end{align*}
Now, by splitting in two halves of all the terms in the above expression other than the ones in 
the double sums, and recalling what the spatial frequencies in $\mathcal{W}_M$ are, 
the above covariance can be written as
\begin{equation}
{\rm cov} \left\{  T_{M_1,M_2}(\bfu_{i,j}),  T_{M_1,M_2}(\bfu_{i',j'})  \right\}	
= \sigma^2 (V^{(1)}_{M_1,M_2} + V^{(2)}_{M_1,M_2}) , 
\label{cov-T}
\end{equation}
where
\begin{align*}
& V^{(1)}_{M_1,M_2} \\
&= \frac{1}{2} \frac{c_{\Delta}}{M}  \Bigg[  
\cos \big((0,0) \bfh \big) f^{\Delta}_{\vartheta}(0, 0)
	+  \cos \Big(\big(\frac{\pi}{\Delta},0\big) \bfh \Big)
f^{\Delta}_{\vartheta}\left({\frac{\pi}{\Delta}, 0}\right)  \\ 
& \quad +  \cos \Big(\big(0,\frac{\pi}{\Delta}\big) \bfh \Big)
f^{\Delta}_{\vartheta}\left( 0, \frac{\pi}{\Delta} \right)
+ \cos \Big(\big(\frac{\pi}{\Delta},\frac{\pi}{\Delta}\big) \bfh \Big) 
f^{\Delta}_{\vartheta}\left(\frac{\pi}{\Delta}, \frac{\pi}{\Delta} \right) \\	
& \quad + 2 \sum_{m_1 = 1}^{\frac{M_1}{2}- 1} \cos\Big( \Big(\frac{2\pi m_1}{\Delta M_1},0\Big) \bfh \Big) 
f^{\Delta}_{\vartheta}\Big(\frac{2\pi m_1}{\Delta M_1}, 0\Big) 
 + 2 \sum_{m_2 = 1}^{\frac{M_2}{2}- 1} \cos\Big( \Big(0,\frac{2\pi m_2}{\Delta M_2}\Big) \bfh \Big) 
f^{\Delta}_{\vartheta}\Big(0,\frac{2\pi m_2}{\Delta M_2}\Big) \\
& \quad + 2 
\sum_{m_1 = 1}^{\frac{M_1}{2}- 1} \cos\Big( \Big(\frac{2\pi m_1}{\Delta M_1},\frac{\pi}{\Delta}\Big) \bfh \Big) 
f^{\Delta}_{\vartheta}\Big(\frac{2\pi m_1}{\Delta M_1}, \frac{\pi}{\Delta}\Big) 
+ 2 \sum_{m_2 = 1}^{\frac{M_2}{2}- 1} \cos\Big( \Big(\frac{\pi}{\Delta},\frac{2\pi m_2}{\Delta M_2}\Big) \bfh \Big) 
f^{\Delta}_{\vartheta}\Big(\frac{\pi}{\Delta},\frac{2\pi m_2}{\Delta M_2}\Big) \\
& \quad + 4 \sum_{m_1= 1}^{\frac{M_1}{2}- 1}\sum_{m_2 = 1}^{\frac{M_2}{2}-1} 
\cos\Big( \Big(\frac{2\pi m_1}{\Delta M_1},\frac{2\pi m_2}{\Delta M_2}\Big) \bfh \Big) 
f^{\Delta}_{\vartheta}\Big(\frac{2\pi m_1}{\Delta M_1},\frac{2\pi m_2}{\Delta M_2}\Big)  \Bigg] \\
&= 2 \cdot \frac{c_{\Delta}}{M}
\sum_{m_1=0}^{\frac{M_1}{2}}\sum_{m_2=0}^{\frac{M_2}{2}} a_{m_1,m_2}
\cos\Bigg(\left(\frac{2\pi m_1}{\Delta M_1} , \frac{2\pi m_2}{\Delta M_2} \right)
\bfh \Bigg)
f^{\Delta}_{\vartheta}\left(\frac{2\pi m_1}{\Delta M_1} , \frac{2\pi m_2}{\Delta M_2} \right) \\
&= 2 \cdot Q^{(1)}_{M_1,M_2}, \ {\rm say} ,
\end{align*}
with
\[
a_{m_1,m_2} = \left\{ \begin{array}{ll} 
\frac{1}{4} & \mbox{if $(m_1,m_2) \in I_C$ \ (a `corner' frequency)} \\
	\frac{1}{2} & \mbox{if $(m_1,m_2) \in I_B$ \ (a `boundary' frequency)} \\
1 & \mbox{if $(m_1,m_2) \in I_I$ \ (an `interior' frequency)}
\end{array} \right. ,
\]
and  
\begin{align*}
& V^{(2)}_{M_1,M_2} \\
&= \frac{1}{2} \frac{c_{\Delta}}{M}  \Bigg[  
\cos \big((0,0) \bfh \big) f^{\Delta}_{\vartheta}(0, 0)
+  \cos \Big(\big(\frac{\pi}{\Delta},0\big) \bfh \Big)
f^{\Delta}_{\vartheta}\left({\frac{\pi}{\Delta}, 0}\right)  \\ 
& \quad +  \cos \Big(\big(0,\frac{\pi}{\Delta}\big) \bfh \Big)
f^{\Delta}_{\vartheta}\left( 0, \frac{\pi}{\Delta} \right)
+ \cos \Big(\big(\frac{\pi}{\Delta},\frac{\pi}{\Delta}\big) \bfh \Big) 
f^{\Delta}_{\vartheta}\left(\frac{\pi}{\Delta}, \frac{\pi}{\Delta} \right) \\	
& \quad + 2 \sum_{m_1 = 1}^{\frac{M_1}{2}- 1} \cos\Big( \Big(\frac{2\pi m_1}{\Delta M_1},0\Big) \bfh \Big) 
f^{\Delta}_{\vartheta}\Big(\frac{2\pi m_1}{\Delta M_1}, 0\Big) 
 + 2 \sum_{m_2 = 1}^{\frac{M_2}{2}- 1} \cos\Big( \Big(0,\frac{2\pi m_2}{\Delta M_2}\Big) \bfh \Big) 
f^{\Delta}_{\vartheta}\Big(0,\frac{2\pi m_2}{\Delta M_2}\Big) \\
& \quad + 2 
\sum_{m_1 = 1}^{\frac{M_1}{2}- 1} \cos\Big( \Big(\frac{2\pi m_1}{\Delta M_1},\frac{\pi}{\Delta}\Big) \bfh \Big) 
f^{\Delta}_{\vartheta}\Big(\frac{2\pi m_1}{\Delta M_1}, \frac{\pi}{\Delta}\Big) 
+ 2 \sum_{m_2 = 1}^{\frac{M_2}{2}- 1} \cos\Big( \Big(\frac{\pi}{\Delta},\frac{2\pi m_2}{\Delta M_2}\Big) \bfh \Big) 
f^{\Delta}_{\vartheta}\Big(\frac{\pi}{\Delta},\frac{2\pi m_2}{\Delta M_2}\Big) \\
& \quad + 4 \sum_{m_1= 1}^{\frac{M_1}{2}- 1}\sum_{m_2 = -\frac{M_2}{2}+1}^{-1} 
\cos\Big( \Big(\frac{2\pi m_1}{\Delta M_1},\frac{2\pi m_2}{\Delta M_2}\Big) \bfh \Big) 
f^{\Delta}_{\vartheta}\Big(\frac{2\pi m_1}{\Delta M_1},\frac{2\pi m_2}{\Delta M_2}\Big)  \Bigg] \\
   &= \frac{1}{2} \frac{c_{\Delta}}{M}  \Bigg[  
\cos \big((0,0) \bfh \big) f^{\Delta}_{\vartheta}(0, 0)
+  \cos \Big(\big(\frac{\pi}{\Delta},0\big) \bfh \Big)
f^{\Delta}_{\vartheta}\left({\frac{\pi}{\Delta}, 0}\right)  \\ 
& \quad +  \cos \Big(\big(0,-\frac{\pi}{\Delta}\big) \bfh \Big)
f^{\Delta}_{\vartheta}\left(0, -\frac{\pi}{\Delta} \right)
+ \cos \Big(\big(\frac{\pi}{\Delta},-\frac{\pi}{\Delta}\big) \bfh \Big) 
f^{\Delta}_{\vartheta}\left(\frac{\pi}{\Delta}, -\frac{\pi}{\Delta} \right) \\	
& \quad + 2 \sum_{m_1 = 1}^{\frac{M_1}{2}- 1} \cos\Big( \Big(\frac{2\pi m_1}{\Delta M_1},0\Big) \bfh \Big) 
f^{\Delta}_{\vartheta}\Big(\frac{2\pi m_1}{\Delta M_1}, 0\Big) 
 + 2 \sum_{m_2 = -\frac{M_2}{2}+1}^{- 1} \cos\Big( \Big(0,\frac{2\pi m_2}{\Delta M_2}\Big) \bfh \Big) 
f^{\Delta}_{\vartheta}\Big(0,\frac{2\pi m_2}{\Delta M_2}\Big) \\
& \quad + 2 
\sum_{m_1 = 1}^{\frac{M_1}{2}- 1} \cos\Big( \Big(\frac{2\pi m_1}{\Delta M_1},-\frac{\pi}{\Delta}\Big) \bfh \Big) 
f^{\Delta}_{\vartheta}\Big(\frac{2\pi m_1}{\Delta M_1}, -\frac{\pi}{\Delta}\Big) 
+ 2 \sum_{m_2 = -\frac{M_2}{2}+1}^{-1} \cos\Big( \Big(\frac{\pi}{\Delta},\frac{2\pi m_2}{\Delta M_2}\Big) \bfh \Big) 
f^{\Delta}_{\vartheta}\Big(\frac{\pi}{\Delta},\frac{2\pi m_2}{\Delta M_2}\Big) \\
& \quad + 4 \sum_{m_1= 1}^{\frac{M_1}{2}- 1}\sum_{m_2 = -\frac{M_2}{2}+1}^{-1} 
\cos\Big( \Big(\frac{2\pi m_1}{\Delta M_1},\frac{2\pi m_2}{\Delta M_2}\Big) \bfh \Big) 
f^{\Delta}_{\vartheta}\Big(\frac{2\pi m_1}{\Delta M_1},\frac{2\pi m_2}{\Delta M_2}\Big)  \Bigg] \\
&= 2 \cdot \frac{c_{\Delta}}{M}
\sum_{m_1=0}^{\frac{M_1}{2}}\sum_{m_2=-\frac{M_2}{2}}^{0} b_{m_1,m_2}
\cos\Bigg(\left(\frac{2\pi m_1}{\Delta M_1} , \frac{2\pi m_2}{\Delta M_2} \right) \bfh \Bigg)
f^{\Delta}_{\vartheta}\left(\frac{2\pi m_1}{\Delta M_1} , \frac{2\pi m_2}{\Delta M_2} \right) \\
&= 2 \cdot Q^{(2)}_{M_1,M_2}, \ {\rm say} ,
\end{align*}
where $b_{m_1,m_2} \coloneqq a_{m_1,-m_2}$, and for the second identity it was used that 
$f^{\Delta}_{\vartheta}(\cdot)$ is an isotropic (radial) function,
$\cos\Big( \Big(\frac{2\pi m_1}{\Delta M_1},-\frac{\pi}{\Delta}\Big) \bfh \Big) 
= \cos\Big( \Big(\frac{2\pi m_1}{\Delta M_1},\frac{\pi}{\Delta}\Big) \bfh \Big)$
and $\cos\Big( \Big(\frac{\pi}{\Delta},-\frac{2\pi m_2}{\Delta M_2}\Big) \bfh \Big)
= \cos\Big( \Big(\frac{\pi}{\Delta},\frac{2\pi m_2}{\Delta M_2}\Big) \bfh \Big)$. 
Based on the trapezoidal product rule with   
$2\pi/\Delta M_1$ and $2\pi/\Delta M_2$ grid spacings, $Q^{(1)}_{M_1,M_2}$ approximates 
%Note that $Q^{(1)}_{M_1,M_2}$ is the trapezoidal product rule, with grid spacings 
%$2\pi/\Delta M_1$ and $2\pi/\Delta M_2$, to approximate 
$\int_{0}^{\frac{\pi}{\Delta}} \int_{0}^{\frac{\pi}{\Delta}} 
\cos \big(\bfomega^{\top} \bfh \big) f^{\Delta}_{\vartheta}(\bfomega) d\bfomega$
(Dahlquist and Bj\"{o}rck, 2008). 
Since $f^{\Delta}_{\vartheta}(\bfomega)$ is a continuous function, it holds that
as $M_1, M_2 \rightarrow \infty$
 \[
Q^{(1)}_{M_1,M_2} \rightarrow 
\int_{0}^{\frac{\pi}{\Delta}} \int_{0}^{\frac{\pi}{\Delta}} 
\cos \big(\bfomega^{\top} \bfh \big) f^{\Delta}_{\vartheta}(\bfomega) d\bfomega ,
\]
and hence
\begin{equation}
V^{(1)}_{M_1,M_2} \rightarrow 
\int_{0}^{\frac{\pi}{\Delta}} \int_{0}^{\frac{\pi}{\Delta}} 
\cos \big(\bfomega^{\top} \bfh \big) f^{\Delta}_{\vartheta}(\bfomega) d\bfomega 	
+ \int_{-\frac{\pi}{\Delta}}^{0} \int_{-\frac{\pi}{\Delta}}^{0} 
\cos \big(\bfomega^{\top} \bfh \big) f^{\Delta}_{\vartheta}(\bfomega) d\bfomega ,	
\label{v1}
\end{equation}
since $\cos \big(\bfomega^{\top} \bfh \big) f^{\Delta}_{\vartheta}(\bfomega)$
is an even function of $\bfomega$.
Likewise, $Q^{(2)}_{M_1,M_2}$ is the trapezoidal product rule to approximate
$\int_{-\frac{\pi}{\Delta}}^{0} \int_{0}^{\frac{\pi}{\Delta}} 
\cos \big(\bfomega^{\top} \bfh \big) f^{\Delta}_{\vartheta}(\bfomega) d\bfomega$,
so by the same argument it holds that  %as $\Delta M_1, \Delta M_2 \rightarrow \infty$
\begin{equation}
V^{(2)}_{M_1,M_2} \rightarrow 
\int_{-\frac{\pi}{\Delta}}^{0} \int_{0}^{\frac{\pi}{\Delta}} 
\cos \big(\bfomega^{\top} \bfh \big) f^{\Delta}_{\vartheta}(\bfomega) d\bfomega	
+ \int_{0}^{\frac{\pi}{\Delta}} \int_{-\frac{\pi}{\Delta}}^{0} 
\cos \big(\bfomega^{\top} \bfh \big) f^{\Delta}_{\vartheta}(\bfomega) d\bfomega	.
\label{v2}
\end{equation}
Finally, from (\ref{cov-T}), (\ref{v1}) and (\ref{v2}) follow that 
as $M_1,M_2 \rightarrow \infty$
\begin{eqnarray*}
{\rm cov}\left\{T_{M_1, M_2} (\bfu_{i,j}), T_{M_1, M_2} (\bfu_{i',j'})\right\} & \rightarrow &
\sigma^2 \int_{-\frac{\pi}{\Delta}}^{\frac{\pi}{\Delta}} \int_{-\frac{\pi}{\Delta}}^{\frac{\pi}{\Delta}} 
\cos \big(\bfomega^{\top} \bfh \big) f^{\Delta}_{\vartheta}(\bfomega) d\bfomega \\
& = & \sigma^2 K_{\vartheta}(||\bfu_{i,j} - \bfu_{i',j'}||),
\end{eqnarray*}
where the last identity follows from the spectral representation of the covariance function 
of the discrete index random field $Z_{\Delta}(\cdot) - E\{Z_{\Delta}(\cdot)\}$.
This shows result (b).

\bigskip
 
\noindent{\bf S2. Details of the Construction of the Matrix $\boldsymbol{H}_1$} 

\medskip

\noindent
The columns of the $M \times M$ matrix $\boldsymbol{H}_1$ are constructed as follows.
The first four columns are formed by the vectors
\[
\big( \cos(\bfomega_{m_1,m_2}^{\top}\bfu_{1,1}),  \cos(\bfomega_{m_1,m_2}^{\top}\bfu_{1,2}),  
 \ldots, \cos(\bfomega_{m_1,m_2}^{\top}\bfu_{M_1,M_2}) \big)^{\top} ,
\]
obtained for $(m_1,m_2) \in I_C$.
Since $\bfomega_{0,0} = (0,0)^\top$, the first column of $\boldsymbol{H}_1$ is ${\bf 1}_M$.
The next $M/2 - 2$ columns are formed by the vectors 
\[
2\big( \cos(\bfomega_{m_1,m_2}^{\top}\bfu_{1,1}),  \cos(\bfomega_{m_1,m_2}^{\top}\bfu_{1,2}),  \ldots, 
\cos(\bfomega_{m_1,m_2}^{\top}\bfu_{M_1,M_2}) \big)^{\top}, 
\]
obtained for $(m_1,m_2) \in I$, and the last $M/2 - 2$ columns are formed by the vectors
\[
-2 \big( \sin(\bfomega_{m_1,m_2}^{\top}\bfu_{1,1}),  \sin(\bfomega_{m_1,m_2}^{\top}\bfu_{1,2}),  \ldots, 
\sin(\bfomega_{m_1,m_2}^{\top}\bfu_{M_1,M_2}) \big)^{\top},
\]
obtained for $(m_1,m_2) \in I$.
This assures that $\boldsymbol{H}_1 \bfg$ results in the vector that collects all
the right hand sides of (\ref{eq:spectral-approx2}) for $\bfu_{i,j} \in \mathcal{U}_M$.

\bigskip

\noindent
{\bf S3. Numerical Comparison of the Tail Behaviours of Exact and Approximate Reference Priors} 

\medskip

\noindent
To better compare the tail behaviours of the exact and approximate reference priors of $\vartheta$, 
we consider a subset of the set up in Figures \ref{fi:prior-rho-nc} and \ref{fi:prior-rho-c}, but 
we now compare $\log \pi^{\rm R}(\vartheta)$ and $\log \pi^{\rm AR}(\vartheta)$ in the interval $[0.01, 5]$.
Figure \ref{fi:prior-logrho} displays these log prior densities for a $10 \times 10$ sampling design
in $[0,1]^2$, where the left panels correspond to models with constant mean and the right panels to 
models with non--constant mean.
For most scenarios the approximate reference priors have lighter tails, and this is more so in models
with constant mean. 
For the Mat\'ern model, $\pi^{\rm AR}(\vartheta) = O(\vartheta^{-3})$ as $\vartheta \rightarrow \infty$
(Corollary \ref{cor:propriety-matern}), while 
in general for $\nu \geq 1$, $\pi^{\rm R}(\vartheta) = O(\vartheta^{-1})$ \citep[Appendix B]{Mure2021}. 
The latter rate is not tight for some models though, 
as the tight rate varies with the mean function and degree of smoothness. This partially explains the 
different degree of discrepancy between 
$\log \pi^{\rm R}(\vartheta)$ and $\log \pi^{\rm AR}(\vartheta)$ in different models. 

\begin{figure}[t!]
\begin{center}
\psfig{figure=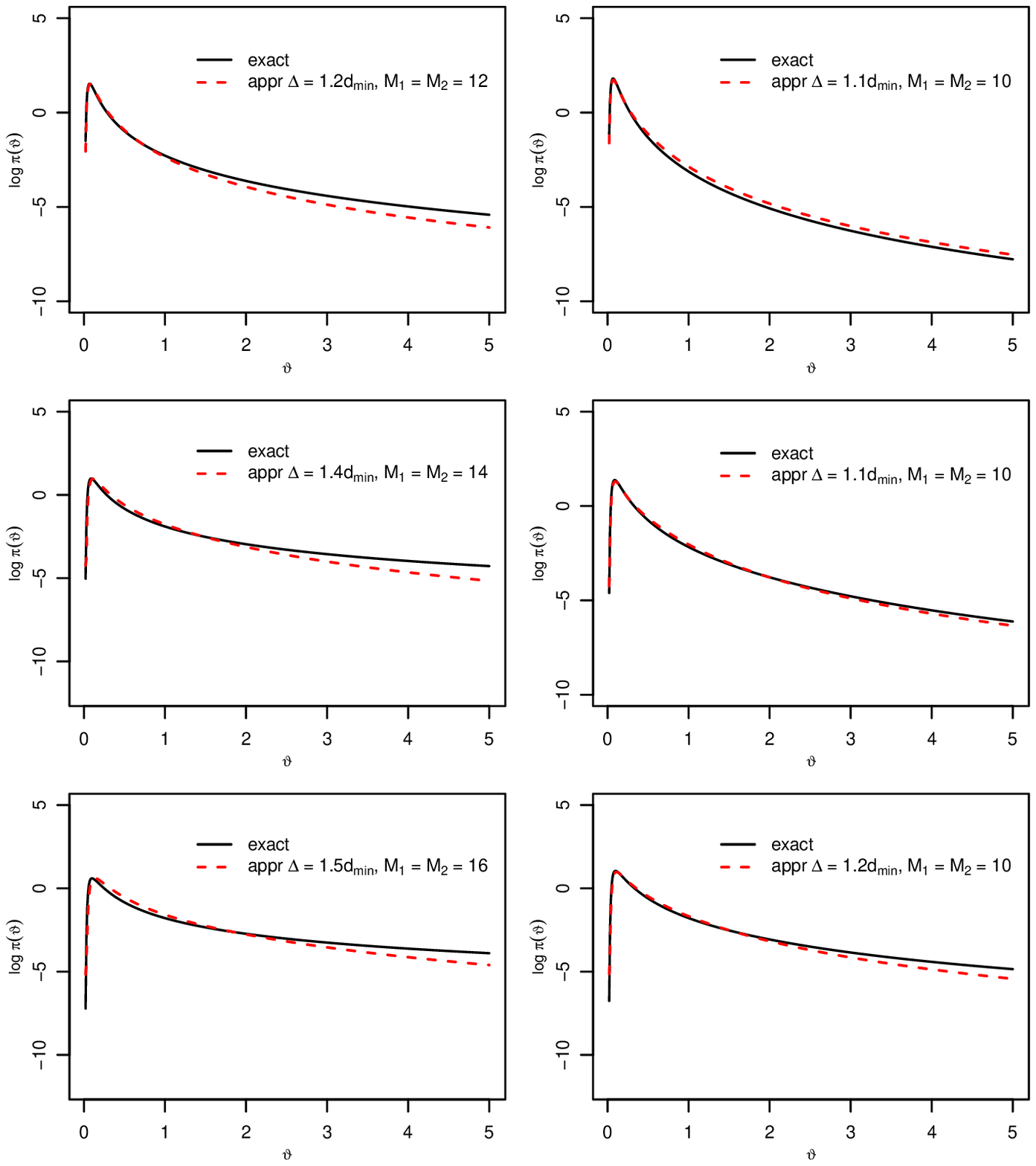, width=12cm, height=12cm} 
\end{center}
\vspace{-0.5 cm}
\caption{Exact and approximate log prior densities of $\vartheta$ for the Mat\'ern covariance function
based on the $10 \times 10$ sampling design in $[0,1]^2$.
The left panels are for models with constant mean, the right panels for models with non--constant mean,
and from top to bottom $\nu = 0.5, 1.5$ and $2.5$.}
\label{fi:prior-logrho}
\end{figure} 

\bigskip

\noindent 
{\bf S4. Sensitivity of the Approximate Reference Priors to the Tuning Constants} 

\medskip

\noindent
To illustrate how sensitive approximate reference priors are to the tuning constants $\Delta$ and $M_1$ 
(for simplicity $M_2 = M_1$), consider the model with mean zero and Mat\'ern correlation function with $\nu = 0.5$.
Figure \ref{fi:sen-deltaM} displays approximate reference priors of $\vartheta$ for different values of $\Delta$ 
and $M_1 = 12$ fixed (left), as well as approximate reference priors for different values of $M_1$ and 
$\Delta = 0.09$ fixed (right).
These show that the approximations are more sensitive to $\Delta$ than to $M_1$, so the tuning of the former is
more important. Section \ref{sec:numerical studies} provides some guidelines for the selection of these tuning constants.
It may also be noted that these constants could be viewed as hyperparameters and estimated using empirical Bayes.

\begin{figure}[ht]
\begin{center}
\psfig{figure=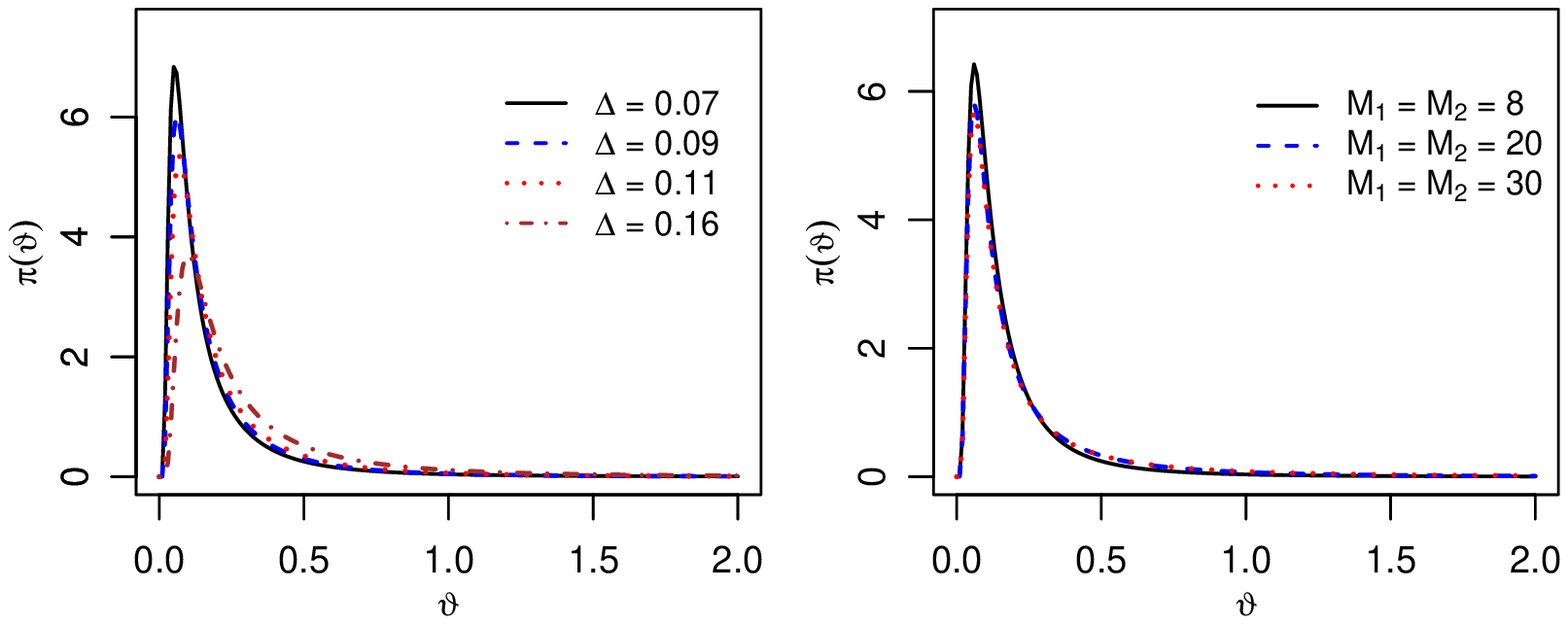, width=13cm, height=6cm} 
\end{center}
\vspace{-0.5 cm}
\caption{Sensitivity assessment of the 
approximate reference priors to tuning constants. 
Left: $M_1 = M_2 = 12$ fixed. Right: $\Delta = 0.09$ fixed.}
\label{fi:sen-deltaM}
\end{figure}

\bigskip

\noindent{\bf S5. Comparison of Frequentist Properties of Bayesian Inferences Based on
Several Default Priors and MLE}

A useful way to evaluate default priors is through the study of frequentist properties  
of the resulting Bayesian inferences (Ghosh and Mukerjee, 1992).
It has been found for a variety of models that reference priors yield credible intervals
with satisfactory frequentist coverage and estimators with competitive mean square errors, 
and this has also been the case for spatial models; 
see \cite{Berger2001,Ren2013,Gu2018} and the references therein.

We use a simulation experiment to compare frequentist properties of Bayesian procedures 
based on the exact and approximate reference priors, as well as those based on a default prior 
suggested by Gu (2019) called ``joint robust prior''.
This was proposed for the (transformed) range parameters of separable correlation functions 
in the context of computer model emulation and calibration; for isotropic correlation functions, it reduces to an inverse gamma prior. 
We use $\pi^{\rm JR}(\vartheta) = {\rm IG}(0.5, \sqrt{2}/100)$, where the hyperparameters 
were set at values recommended by Gu (2019)  
for the design described below. 
In addition, we also compare frequentist properties of purely likelihood--based inferences.

The numerical experiment is based on data in the region $\mcD = [0,1]^2$ simulated from 
Gaussian random fields with several mean and covariance functions at $n=100$ sampling locations.
We consider the $10 \times 10$ regular design and the irregular design displayed at the
bottom left panel of Figure  \ref{fi:prior-aref-ir}.
For the mean function we use $\mu({\bf s}) = 1$ ($p = 1$) and 
$\mu({\bf s}) = 0.15-0.65x-0.1y+0.9x^2-xy+1.2y^2$ ($p = 6$), while for the correlation function  
we use the Mat\'ern model (\ref{eq:matern-cov}) with $\sigma^2 = 1$, 
range parameter $\vartheta = 0.2, 0.4$ and $0.7$ and smoothness parameter $\nu = 0.5$ and $1.5$. 
This setup provides a variety of scenarios in terms of trend, strength of correlation and smoothness.  
For each of these 12 scenarios, $3000$ data sets were simulated, and for each data set 
we generated $10^4$ posterior samples by the Monte Carlo algorithm described at the end of this part.

We compare the following frequentist properties of Bayesian procedures based on the three default priors
and a purely likelihood--based procedure to make inferences about the covariance parameters: 

\medskip

\noindent
(1) Let $[L(\boldsymbol{Z}), U(\boldsymbol{Z})]$ be either a Bayesian highest probability density 
credible interval or a confidence interval for a covariance parameter $\eta$. 
Its frequentist coverage and expected log--length are estimated from the simulated data by
\[
\frac{1}{3000} \sum_{j=1}^{3000} {\bf 1}\{L(\boldsymbol{z}_j) < \eta < U(\boldsymbol{z}_j)\} 
\quad {\rm and} \quad 
\frac{1}{3000} \sum_{j=1}^{3000} \big( \log U(\boldsymbol{z}_j) - \log L(\boldsymbol{z}_j) \big) ,
\]
respectively, where $\eta = \vartheta$ or $\sigma^2$, and 
$\boldsymbol{z}_j$ is the $j{\rm th}$ simulated data set. 
The confidence interval is obtained by evaluating the profile likelihood of $\eta$ and 
inverting a likelihood ratio test (Meeker and Escobar, 1995).

\noindent
(2) Let $\hat{\eta}_j$ be either a Bayesian estimator or the MLE of $\eta$ based on $\boldsymbol{z}_j$.
Its mean absolute error, $E\big(|\hat{\eta} - \eta|\big)$, is estimated by
\[
\frac{1}{3000} \sum_{j=1}^{3000} |\hat{\eta}_j - \eta| .
\]
For Bayesian estimation, $\hat{\eta}_j$ is the mode of $\pi(\eta ~|~ \boldsymbol{z}_j)$
when $\eta = \vartheta$, while
$\hat{\eta}_j$ is the median of $\pi(\eta ~|~ \boldsymbol{z}_j)$ when $\eta = \sigma^2$.

\medskip

\noindent
{\bf Results for the Regular Lattice Design}
To compute the approximate reference prior for models with constant mean, 
we used $M_1 = M_2 = 12$ and $\Delta = 0.133$ when $\nu = 0.5$, and 
$M_1 = M_2 = 14$ and $\Delta = 0.156$ when $\nu = 1.5$.
For models with non--constant mean, we used $M_1 = M_2 = 10$ and $\Delta = 0.122$
for both smoothness parameters. 
These choices were informed by the findings in Section \ref{sec:numerical studies}.
 
\begin{table}
\centering
\caption{Frequentist coverage probability and $[$average log--length$]$ of Bayesian 
95\% highest probability density credible intervals of $\vartheta$ based on the three default priors,
and the 95\% profile likelihood confidence interval of $\vartheta$. Regular design.}
\label{ta:ci-rho-hpd}
\begin{tabular}{lcccccccc} 
\toprule
\multirow{2}{*}{}   &   &       & $p=1$   &       &         &       & $p=6$   &        \\ 
\cline{3-5}\cline{7-9}
 &  & $\vartheta = 0.2$   & $\vartheta = 0.4$  & $\vartheta = 0.7$ & & $\vartheta = 0.2$ & $\vartheta = 0.4$  & $\vartheta = 0.7$  \\ 
\cline{3-5}\cline{7-9}
            &   &       &       &       & $\nu = 0.5$ &       &       &        \\
Inverse Gamma Prior&   & $0.963$   &  $0.968$   & $0.976$   &   & $0.952$  &   $0.980$     &    $0.983$  \\
                              &   & $[2.110]$ & $[3.878]$ & $[4.789]$ &   &  $[5.519]$ & $[5.885]$   &  $[6.030]$ \\

Exact Reference Prior          &   & $0.960$   &  $0.955$   & $0.952$   &   & $0.955$  &   $0.934$     &    $0.921$  \\
                              &   & $[1.632]$ & $[2.541]$ & $[3.066]$ &   &  $[2.079]$ & $[2.335]$   &  $[2.424]$   \\

Appr Reference Prior &   & $0.964$   &  $0.955$   & $0.946$   &   & $0.958$  &   $0.945$     &    $0.935$  \\
                              &   & $[1.535]$ & $[2.008]$ & $[2.313]$ &   &  $[2.181]$ &  $[2.388]$   &  $[2.459]$   \\
                              
MLE &   & $0.933$   &  $0.914$   & $0.885$   &   & $0.623$  &   $0.407$     &    $0.221$  \\
                              &   & $[4.847]$ & $[2.791]$ & $[4.940]$ &   &  $[9.440]$ &  $[5.889]$   &  $[4.335]$   \\ \hline
            &   &       &       &       & $\nu = 1.5$ &       &       &        \\
Inverse Gamma Prior &   & $0.955$   &  $0.952$   & $0.950$   &   & $0.956$  &   $0.967$     &    $0.977$  \\
                              &   & $[0.598]$ & $[0.784]$ & $[1.111]$ &   &  $[0.758]$ & $[1.883]$   &  $[3.926]$ \\
Exact Reference Prior          &   & $0.956$   &  $0.953$   & $0.951$   &   & $0.959$  &   $0.963$     &    $0.957$  \\
                              &   & $[0.596]$ & $[0.781]$ & $[1.106]$ &   &  $[0.742]$ & $[1.227]$   &  $[1.857]$ \\
Appr Reference Prior  &   & $0.960$   &  $0.952$   & $0.944$   &   & $0.960$  &   $0.967$     &    $0.960$  \\
                              &   & $[0.594]$ & $[0.765]$ & $[0.992]$ &   &  $[0.748]$ & $[1.199]$   &  $[1.703]$ \\ 
MLE &   & $0.944$   &  $0.930$   & $0.901$   &   & $0.772$  &   $0.603$     &    $0.395$  \\
                              &   & $[0.684]$ & $[0.714]$ & $[1.021]$ &   &  $[0.755]$ &  $[0.837]$   &  $[0.798]$   \\
\bottomrule
\end{tabular}
\end{table}

Table \ref{ta:ci-rho-hpd} reports for different models the frequentist coverage and [average log--length]  
of the 95\% highest probability density credible intervals (HPDCI) based on the three default priors
and of the profile likelihood confidence interval (PLCI) for $\vartheta$.
In almost all situations the HPDCI based on the three default priors have coverage probabilities 
close to the target $0.95$, with those from the approximate reference priors almost never being inferior 
to those from the exact reference priors.
The coverage probabilities of HPDCI based on the inverse gamma prior are a bit too large in some situations.
On the other hand, the coverage probabilities of PLCI tend to be smaller than the target $0.95$, and 
this is substantially so for models with non--constant mean.
The log--lengths of the HPDCI for $\vartheta$ based on the exact and approximate reference priors are 
similar in most situations, with those from approximate reference priors tending to be slightly shorter
than those from exact reference priors. The log--lengths of the HPDCI based on the inverse gamma prior 
tend to be larger, sometimes substantially so, and the same holds for PLCI.

\begin{table}[t]
\centering
\caption{Frequentist coverage probability and $[$average log--length$]$ of Bayesian 
95\% highest probability density credible intervals of $\sigma^2$ based on the three default priors, 
and the 95\% profile likelihood confidence interval of $\sigma^2$. Regular design.}
\label{ta:ci-sigsq-hpd}
\begin{tabular}{lcccccccc} 
\toprule
\multirow{ 2}{*}{}   &   &       & $p=1$   &       &         &       & $p=6$   &        \\ 
\cline{3-5}\cline{7-9}
            &  & $\vartheta = 0.2$   & $\vartheta = 0.4$  & $\vartheta = 0.7$ & & $\vartheta = 0.2$ & $\vartheta = 0.4$  & $\vartheta = 0.7$  \\ \cline{3-5}\cline{7-9}
            &   &       &       &       & $\nu = 0.5$ &       &       &        \\
Inverse Gamma Prior &   & $0.968$   &  $0.973$   & $0.976$   &   & $0.960$  &   $0.982$     &    $0.984$  \\
                              &   & $[1.577]$ & $[3.565]$ & $[4.541]$ &   &  $[3.723]$ & $[5.061]$   &  $[5.395]$ \\

Exact Reference Prior           &   & $0.968$   &  $0.960$   & $0.954$   &   & $0.964$  &   $0.934$     &    $0.927$  \\
                              &   & $[1.246]$ & $[2.261]$ & $[2.835]$ &   &  $[1.513]$ & $[1.932]$   &  $[2.070]$   \\

Appr Reference Prior &   & $0.971$   &  $0.959$   & $0.943$   &   & $0.963$  &   $0.945$     &    $0.933$  \\
                              &   & $[1.173]$ & $[1.796]$ & $[2.115]$ &   &  $[1.574]$ &  $[1.974]$   &  $[2.095]$   \\
MLE &   & $0.941$   &  $0.914$   & $0.892$   &   & $0.630$  &   $0.390$     &    $0.202$  \\
                              &   & $[1.130]$ & $[1.837]$ & $[3.863]$ &   &  $[1.307]$ &  $[1.238]$   &  $[1.958]$   \\\hline
            &   &       &       &       & $\nu = 1.5$ &       &       &        \\
Inverse Gamma Prior &   & $0.953$   &  $0.953$   & $0.948$   &   & $0.957$  &   $0.961$     &    $0.977$  \\
                              &   & $[1.068]$ & $[1.961]$ & $[3.127]$ &   &  $[1.385]$ & $[4.982]$   &  $[7.962]$ \\
Exact Reference Prior         &   & $0.954$   &  $0.955$   & $0.945$   &   & $0.956$  &   $0.958$     &    $0.954$  \\
                              &   & $[1.075]$ & $[1.954]$ & $[3.108]$ &   &  $[1.366]$ & $[3.142]$   &  $[5.184]$ \\
Appr Reference Prior &   & $0.957$   &  $0.956$   & $0.941$   &   & $0.957$  &   $0.961$     &    $0.955$  \\
                              &   & $[1.076]$ & $[1.911]$ & $[2.801]$ &   &  $[1.382]$ & $[3.080]$   &  $[4.768]$ \\ 
MLE &   & $0.940$   &  $0.932$   & $0.908$   &   & $0.746$  &   $0.577$     &    $0.376$  \\
                              &   & $[1.095]$ & $[1.781]$ & $[2.939]$ &   &  $[1.128]$ & $[1.872]$   &  $[2.962]$ \\ 
\bottomrule
\end{tabular}
\end{table}

\begin{table}[htbp!]
\centering
\caption{Mean absolute error of the posterior mode of $\vartheta$ based on the three default priors, 
and of the MLE of $\vartheta$ for the regular design.}
\label{ta:mae-rho}
\begin{tabular}{lcccccccc} 
\toprule
\multirow{2}{*}{}   &   &       & $p=1$   &       &         &       & $p=6$   &        \\ 
\cline{3-5}\cline{7-9}
            &  & $\vartheta = 0.2$   & $\vartheta = 0.4$  & $\vartheta = 0.7$ & & $\vartheta = 0.2$ & $\vartheta = 0.4$  & $\vartheta = 0.7$  \\ \cline{3-5}\cline{7-9}
            &   &       &       &       & $\nu = 0.5$ &       &       &        \\
Inverse Gamma Prior &   & $0.046$ & $0.121$ & $0.287$ &         & $0.059$ & $0.162$ & $0.404$  \\
Exact Reference Prior          &   & $0.044$ & $0.120$ & $0.296$ &         & $0.055$ & $0.175$ & $0.440$  \\
Appr Reference Prior &   & $0.043$ & $0.117$ & $0.294$ &         & $0.056$ & $0.169$ & $0.431$  \\  
MLE &   & $0.064$ & $0.131$ & $0.412$ & & $0.082$ & $0.195$ & $0.560$  \\  \hline
            &   &       &       &       & $\nu = 1.5$ &       &       &        \\
Inverse Gamma Prior &   & $0.023$ & $0.056$ & $0.128$ &         & $0.028$ & $0.073$ & $0.183$  \\
Exact Reference Prior         &   & $0.022$ & $0.055$ & $0.127$ &         & $0.027$ & $0.072$ & $0.181$  \\
Appr Reference Prior &   & $0.022$ & $0.055$ & $0.126$ &         & $0.027$ & $0.071$ & $0.176$  \\ 
MLE &   & $0.029$ & $0.057$ & $0.136$ &  & $0.036$ & $0.087$ & $0.287$  \\  
\bottomrule
\end{tabular}
\end{table}

\begin{table}[htbp!]
\centering
\caption{Mean absolute error of the posterior median of $\sigma^2$ based on the three default priors, 
and of the MLE of $\sigma^2$ for the regular design.}
\label{ta:mae-sigsq}
\begin{tabular}{lcccccccc} \toprule
\multirow{ 2}{*}{}   &   &       & $p=1$   &       &         &       & $p=6$   &        \\ 
\cline{3-5}\cline{7-9}
            &  & $\vartheta = 0.2$   & $\vartheta = 0.4$  & $\vartheta = 0.7$ & & $\vartheta = 0.2$ & $\vartheta = 0.4$  & $\vartheta = 0.7$  \\ \cline{3-5}\cline{7-9}
            &   &       &       &       & $\nu = 0.5$ &       &       &        \\
Inverse Gamma Prior &   & $0.240$ & $0.618$ & $0.843$ &         & $0.650$ & $0.988$ & $0.690$  \\
Exact Reference Prior         &   & $0.207$ & $0.356$ & $0.413$ &         & $0.238$ & $0.283$ & $0.431$  \\
Appr Reference Prior &   & $0.198$ & $0.273$ & $0.307$ &         & $0.256$ & $0.279$ & $0.414$  \\  
MLE &   & $0.205$ & $0.289$ & $0.442$ & & $0.333$ & $0.506$ & $0.613$  \\ \hline
            &   &       &       &       & $\nu = 1.5$ &       &       &        \\
Inverse Gamma Prior&   & $0.212$ & $0.389$ & $0.656$ &         & $0.287$ & $1.263$ & $3.013$  \\
Exact Reference  Prior  &   & $0.213$ & $0.386$ & $0.645$ &         & $0.282$ & $0.639$ & $0.892$  \\
Appr Reference Prior &   & $0.213$ & $0.373$ & $0.536$ &         & $0.287$ & $0.632$ & $0.863$  \\ 
MLE &   & $0.209$ & $0.348$ & $0.503$ & & $0.287$ & $0.524$ & $0.720$  \\ 
\bottomrule
\end{tabular}
\end{table}

Table \ref{ta:ci-sigsq-hpd} reports the frequentist coverage and [average log--length]  
of the 95\% HPDCI based on the three default priors and of the PLCI for $\sigma^2$.
The findings are very similar to those in Table \ref{ta:ci-rho-hpd} for $\vartheta$. 
In particular, the coverage probabilities of HPDCI based on the exact and approximate reference priors are
close to the target, while those of PLCI are substantially smaller than the target when the mean is non-constant.
Also, the log--lengths of the HPDCI for $\sigma^2$ based on the inverse gamma prior tend to be the largest 
and those based on the approximate reference priors tend to be the smallest.

Tables \ref{ta:mae-rho} and \ref{ta:mae-sigsq} report the mean absolute errors (MAE) 
of the Bayesian estimators based on the three default priors and the MAE of the MLE. 
The MAE of the three Bayesian estimates of $\vartheta$ are very close to each other in all scenarios,
while that of the MLE is larger.
The MAE of the Bayesian estimates of $\sigma^2$ based on the exact and approximate reference priors
are similar in most scenarios.
On the other hand, the MAE of the Bayesian estimates based on the inverse gamma prior tend to be larger than
the other two, sometimes substantially so, due to the tendency to overestimate $\sigma^2$. 
The same holds for the MLE, but due to the tendency 
%of this estimator 
to underestimate $\sigma^2$.
%The slight difference in the behaviours of the mean absolute errors of estimators for 
%$\vartheta$ and $\sigma^2$ is due to the use of the posterior mode as the estimator of $\vartheta$ and 
%the posterior median as the estimator of $\sigma^2$.

\medskip

\noindent
{\bf Results for the Irregular Design}

To compute the approximate reference priors for models with constant mean, we used 
$M_1 = M_2 = 12$ and $\Delta = 0.09$ when $\nu = 0.5$, and 
$M_1 = M_2 = 12$ with $\Delta = 0.1$ when $\nu = 1.5$.
For models with non--constant mean, we used $M_1 = M_2 = 12$ and $\Delta = 0.09$ 
for both smoothness parameters. The sampling design of 
this study is based on the complete random design 
shown in the bottom left of Figure \ref{fi:prior-aref-ir}. 

Tables \ref{ta:ci-rho-hpd-rand} and \ref{ta:ci-sigsq-hpd-rand} report the frequentist coverage and 
[average log--length] of the 95\% HPDCI for $\vartheta$ and $\sigma^2$, respectively.
For most scenarios the results are about the same as those for the regular design. 
The coverage probabilities based on the three default priors are satisfactory as they are close to the target 0.95. 
They all tend to be slightly lower than nominal though, when the correlation is strong ($\vartheta = 0.7$), 
and the coverage probabilities based on the inverse gamma prior are a bit too large in some situations.
Also, similar to the findings in the regular design, the inverse gamma prior can yield substantially wider 
confidence intervals. 
Tables \ref{ta:mae-rho-rand} and \ref{ta:mae-sigsq-rand} report the MAE  of the Bayesian estimators of 
$\vartheta$ and $\sigma^2$, respectively. 
The results are similar to the results shown  in Tables \ref{ta:mae-rho} and \ref{ta:mae-sigsq} for the 
regular lattice design. 
Again, the MAE of $\hat{\vartheta}$ based on the three default priors are about the same, but 
the MAE of $\hat{\sigma}^2$ based on the inverse gamma prior tends to be larger than the ones based on 
the reference priors due to the tendency of overestimation.

\begin{table}[htbp!]
\centering
\caption{Frequentist coverage probability and $[$average log--length$]$ of Bayesian 
95\% highest probability density credible intervals of $\vartheta$ based on 
the three default priors for the irregular design.}
\label{ta:ci-rho-hpd-rand}
\begin{tabular}{lcccccccc} 
\toprule
\multirow{2}{*}{Prior}   &   &       & $p=1$   &       &         &       & $p=6$   &        \\ 
\cline{3-5}\cline{7-9}
 &  & $\vartheta = 0.2$   & $\vartheta = 0.4$  & $\vartheta = 0.7$ & & $\vartheta = 0.2$ & $\vartheta = 0.4$  & $\vartheta = 0.7$  \\ 
\cline{3-5}\cline{7-9}
            &   &       &       &       & $\nu = 0.5$ &       &       &        \\
Inverse Gamma &   & $0.954$   &  $0.968$   & $0.980$   &   & $0.969$  &   $0.980$     &    $0.971$  \\
                              &   & $[2.316]$ & $[4.076]$ & $[4.853]$ &   &  $[4.995]$ & $[5.613]$   &  $[5.872]$ \\

Exact Reference            &   & $0.956$   &  $0.947$   & $0.956$   &   & $0.957$  &   $0.929$     &    $0.899$  \\
                              &   & $[1.777]$ & $[2.672]$ & $[3.155]$ &   &  $[2.133]$ & $[2.362]$   &  $[2.427]$   \\

Appr Reference  &   & $0.952$   &  $0.935$   & $0.920$   &   & $0.957$  &   $0.931$     &    $0.911$  \\
                              &   & $[1.524]$ & $[1.999]$ & $[2.282]$ &   &  $[2.167]$ &  $[2.385]$   &  $[2.461]$   \\  \hline
            &   &       &       &       & $\nu = 1.5$ &       &       &        \\
Inverse Gamma &   & $0.949$   &  $0.956$   & $0.964$   &   & $0.952$  &   $0.971$     &    $0.977$  \\
                              &   & $[0.585]$ & $[0.819]$ & $[1.185]$ &   &  $[0.755]$ & $[2.146]$   &  $[3.784]$ \\
Exact Reference           &   & $0.954$   &  $0.955$   & $0.964$   &   & $0.953$  &   $0.963$     &    $0.953$  \\
                              &   & $[0.584]$ & $[0.813]$ & $[1.183]$ &   &  $[0.738]$ & $[1.329]$   &  $[1.937]$ \\
Appr Reference  &   & $0.951$   &  $0.953$   & $0.951$   &   & $0.948$  &   $0.963$     &    $0.953$  \\
                              &   & $[0.581]$ & $[0.784]$ & $[1.013]$ &   &  $[0.741]$ & $[1.253]$   &  $[1.733]$ \\ 
\bottomrule
\end{tabular}
\end{table}

\begin{table}[htbp!]
\centering
\caption{Frequentist coverage probability and $[$average log--length$]$ of Bayesian 
95\% highest probability density credible intervals of $\sigma^2$ based on
the three default priors for the irregular design.}
\label{ta:ci-sigsq-hpd-rand}
\begin{tabular}{lcccccccc} 
\toprule
\multirow{ 2}{*}{Prior}   &   &       & $p=1$   &       &         &       & $p=6$   &        \\ 
\cline{3-5}\cline{7-9}
            &  & $\vartheta = 0.2$   & $\vartheta = 0.4$  & $\vartheta = 0.7$ & & $\vartheta = 0.2$ & $\vartheta = 0.4$  & $\vartheta = 0.7$  \\ \cline{3-5}\cline{7-9}
            &   &       &       &       & $\nu = 0.5$ &       &       &        \\
Inverse Gamma &   & $0.967$   &  $0.967$   & $0.980$   &   & $0.973$  &   $0.981$     &    $0.990$  \\
                              &   & $[1.946]$ & $[3.804]$ & $[4.634]$ &   &  $[4.074]$ & $[5.165]$   &  $[5.482]$ \\

Exact Reference            &   & $0.966$   &  $0.957$   & $0.958$   &   & $0.960$  &   $0.923$     &    $0.907$  \\
                              &   & $[1.454]$ & $[2.329]$ & $[2.949]$ &   &  $[1.683]$ & $[2.017]$   &  $[2.118]$   \\

Appr Reference  &   & $0.965$   &  $0.940$   & $0.917$   &   & $0.963$  &   $0.932$     &    $0.921$  \\
                              &   & $[1.230]$ & $[1.869]$ & $[2.106]$ &   &  $[1.718]$ &  $[2.038]$   &  $[2.130]$   \\  \hline
            &   &       &       &       & $\nu = 1.5$ &       &       &        \\
Inverse Gamma &   & $0.950$   &  $0.951$   & $0.960$   &   & $0.953$  &   $0.971$     &    $0.977$  \\
                              &   & $[1.198]$ & $[2.139]$ & $[3.395]$ &   &  $[1.593]$ & $[5.870]$   &  $[9.847]$ \\
Exact Reference           &   & $0.955$   &  $0.951$   & $0.958$   &   & $0.953$  &   $0.963$     &    $0.950$  \\
                              &   & $[1.201]$ & $[2.125]$ & $[3.390]$ &   &  $[1.556]$ & $[3.557]$   &  $[5.493]$ \\
Appr Reference  &   & $0.954$   &  $0.947$   & $0.947$   &   & $0.955$  &   $0.964$     &    $0.945$  \\
                              &   & $[1.197]$ & $[2.042]$ & $[2.897]$ &   &  $[1.561]$ & $[3.356]$   &  $[4.927]$ \\ 
\bottomrule
\end{tabular}
\end{table}

\begin{table}[htbp!]
\centering
\caption{ Mean absolute error of the posterior mode of $\vartheta$ based on the three default priors 
for the irregular design.}
\label{ta:mae-rho-rand}
\begin{tabular}{lcccccccc} 
\toprule
\multirow{ 2}{*}{Prior}   &   &       & $p=1$   &       &         &       & $p=6$   &        \\ 
\cline{3-5}\cline{7-9}
            &  & $\vartheta = 0.2$   & $\vartheta = 0.4$  & $\vartheta = 0.7$ & & $\vartheta = 0.2$ & $\vartheta = 0.4$  & $\vartheta = 0.7$  \\ \cline{3-5}\cline{7-9}
            &   &       &       &       & $\nu = 0.5$ &       &       &        \\
Inverse Gamma &   & $0.048$ & $0.129$ & $0.301$ &         & $0.061$ & $0.173$ & $0.432$  \\
Exact Reference           &   & $0.047$ & $0.131$ & $0.313$ &         & $0.058$ & $0.190$ & $0.466$  \\
Appr Reference &   & $0.046$ & $0.129$ & $0.326$ &         & $0.060$ & $0.186$ & $0.460$  \\  \hline
            &   &       &       &       & $\nu = 1.5$ &       &       &        \\
Inverse Gamma &   & $0.022$ & $0.059$ & $0.131$ &         & $0.028$ & $0.079$ & $0.195$  \\
Exact Reference           &   & $0.023$ & $0.059$ & $0.132$ &         & $0.028$ & $0.077$ & $0.192$  \\
Appr Reference &   & $0.022$ & $0.058$ & $0.130$ &         & $0.027$ & $0.076$ & $0.190$  \\ 
\bottomrule
\end{tabular}
\end{table}

\begin{table}[htbp!]
\centering
 \caption{Mean absolute error of the posterior median of $\sigma^2$ based on the three default priors 
for the irregular design.}
\label{ta:mae-sigsq-rand}
\begin{tabular}{lcccccccc} \toprule
\multirow{ 2}{*}{Prior}   &   &       & $p=1$   &       &         &       & $p=6$   &        \\ 
\cline{3-5}\cline{7-9}
            &  & $\vartheta = 0.2$   & $\vartheta = 0.4$  & $\vartheta = 0.7$ & & $\vartheta = 0.2$ & $\vartheta = 0.4$  & $\vartheta = 0.7$  \\ \cline{3-5}\cline{7-9}
            &   &       &       &       & $\nu = 0.5$ &       &       &        \\
Inverse Gamma &   & $0.299$ & $0.680$ & $0.859$ &         & $0.824$ & $0.975$ & $0.614$  \\
Exact Reference           &   & $0.244$ & $0.377$ & $0.426$ &         & $0.267$ & $0.303$ & $0.484$  \\
Appr Reference &   & $0.218$ & $0.269$ & $0.337$ &         & $0.281$ & $0.298$ & $0.473$  \\  \hline
            &   &       &       &       & $\nu = 1.5$ &       &       &        \\
Inverse Gamma &   & $0.240$ & $0.443$ & $0.717$ &         & $0.326$ & $1.403$ & $3.339$  \\
Exact Reference           &   & $0.240$ & $0.437$ & $0.699$ &         & $0.314$ & $0.769$ & $0.935$  \\
Appr Reference &   & $0.239$ & $0.415$ & $0.557$ &         & $0.317$ & $0.741$ & $0.892$  \\ 
\bottomrule
\end{tabular}
\end{table}

\medskip

\noindent
{\bf Algorithm to Sample from the Posterior Distribution}.
We describe a non--iterative Monte Carlo algorithm to draw independent samples from the 
exact or approximate reference posterior distribution of $(\bfbeta,\sigma^2,\vartheta)$, 
which by its nature is more efficient than alternative iterative Markov Chain Monte Carlo algorithms.
It is based on factoring the posterior as
\[
\pi(\bfbeta, \sigma^2, \vartheta ~|~ \boldsymbol{z}) = \pi(\bfbeta ~|~ \sigma^2, \vartheta, \boldsymbol{z}) 
\pi(\sigma^2 ~|~ \vartheta, \boldsymbol{z}) \pi(\vartheta ~|~ \boldsymbol{z}) ,
\]
where 
\begin{align*}
\label{eq:marg-post-rho}
	\pi(\bfbeta ~|~ \sigma^2, \vartheta, \boldsymbol{z}) &= {\rm N}_{p}\big(\hat{\bfbeta}_{\vartheta} , 
\sigma^2 (\boldsymbol{X}^{\top} \boldsymbol{\Sigma}^{-1}_{\vartheta}\boldsymbol{X})^{-1}\big) \nonumber \\
\pi(\sigma^2 ~|~ \vartheta, \boldsymbol{z}) &= {\rm IG} \left( \frac{n-p}{2},  \frac{S^2_{\vartheta}}{2} \right) \nonumber \\
\pi(\vartheta ~|~ \boldsymbol{z}) & \propto \pi(\vartheta)
|\boldsymbol{\Sigma}_{\vartheta}|^{-\frac{1}{2}} 
|\boldsymbol{X}^{\top}\boldsymbol{\Sigma}_{\vartheta}^{-1}\boldsymbol{X}|^{- \frac{1}{2}} (S^{2}_{\vartheta})^{-\frac{n-p}{2}} ,
\end{align*} 
and $\pi(\vartheta)$ is $\pi^{\rm R}(\vartheta)$, $\pi^{\rm AR}(\vartheta)$ or the 
inverse gamma prior; 
see Section (\ref{sec:derivations}) for the definition of the terms involved.
So a draw from $\pi(\bfbeta, \sigma^2, \vartheta ~|~ \boldsymbol{z})$ is obtained by sampling in turn
from $\pi(\vartheta ~|~ \boldsymbol{z})$, $\pi(\sigma^2 ~|~ \vartheta, \boldsymbol{z})$ and 
$\pi(\bfbeta ~|~ \sigma^2, \vartheta, \boldsymbol{z})$.
The only non--standard distribution is $\pi(\vartheta ~|~ \boldsymbol{z})$, which is one--dimensional
since the correlation function depends on a single range parameter.
To sample from it, we use the generalized ratio--of--uniforms algorithm, which is efficiently implemented 
in the {\tt R} package {\tt rust}. 
Details can be found in Sun (2006) and Northrop (2020).

\bigskip

\noindent
{\bf S6. Analysis of a Data Set Simulated on an Irregular Sampling Design}

\begin{figure}[b!]
\begin{center}
\psfig{figure=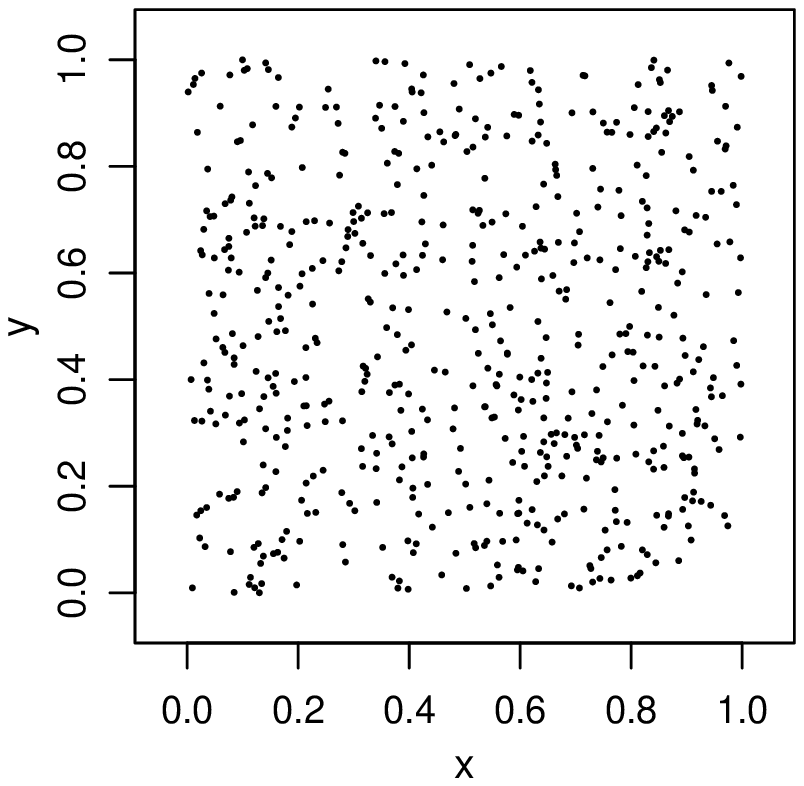, width=7cm,height=6cm} 
\psfig{figure=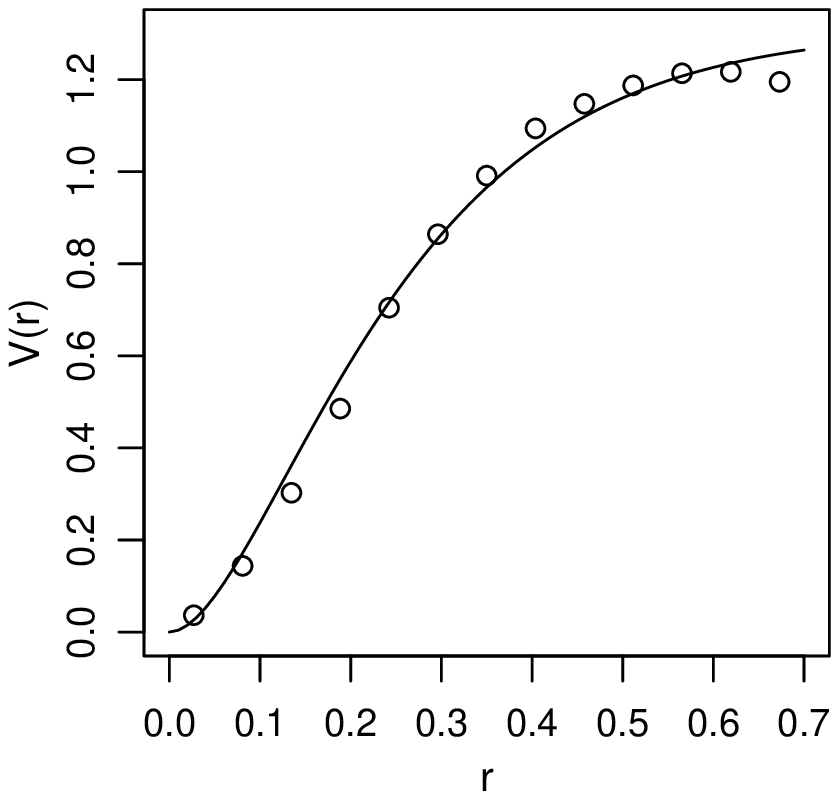, width=7cm,height=6cm} 
\end{center}
\vspace{-0.5 cm}
\caption{ Left: Sampling locations of the simulated data set. 
Right: Integrated likelihood of $\nu$ for the simulated same data.}
\label{fi:data2-fig1}
\end{figure}

Here we carry out data analyses that is parallel to 
those done in Section 6, but based on a 
(simulated) data set with different features from those of the lead concentration data.
Specifically, for the sampling design consisting of $n = 600$ locations forming a random sample from
the ${\rm unif}((0, 1)^2)$ distribution, we simulated a Gaussian random field with mean 0 and Mat\'ern
covariance function with parameters $\sigma^2 = 1$, $\vartheta = 0.3$ and $\nu = 1.5$.
For the analysis below it is assumed that $\hat{\nu} = 1.37$, obtained by maximizing the integrated likelihood of $\nu$ 
displayed in Figure \ref{fi:data2-fig2} (left).
Figure \ref{fi:data2-fig1} (left) displays the sampling design, while the empirical semivariogram of 
the data and the semivariogram function fitted by least squares ($\hat{\sigma}^2 = 1.31$ and $\hat{\vartheta} = 0.33$) 
are displayed in Figure \ref{fi:data2-fig1} (right).

To compute the approximate reference prior we set $M_1 = M_2 = 26$ and $\Delta = 0.04$
(following the guidelines described in Section \ref{sec:numerical studies}), and 
$\tilde{f}^{\Delta}_{\vartheta}(\bfomega_j)$ was obtained by setting $\max\{|l_1|, |l_2|\} \leq 5$.
Then, two Bayesian analyses were carried out based on the exact and approximate reference priors,
where samples of sizes $10^4$ from the corresponding posteriors of $(\beta_1,\sigma^2,\vartheta)$
were simulated using again the Monte Carlo algorithm described in part 5 of the Supplementary Materials.

\begin{figure}[t!]
\begin{center}
\psfig{figure=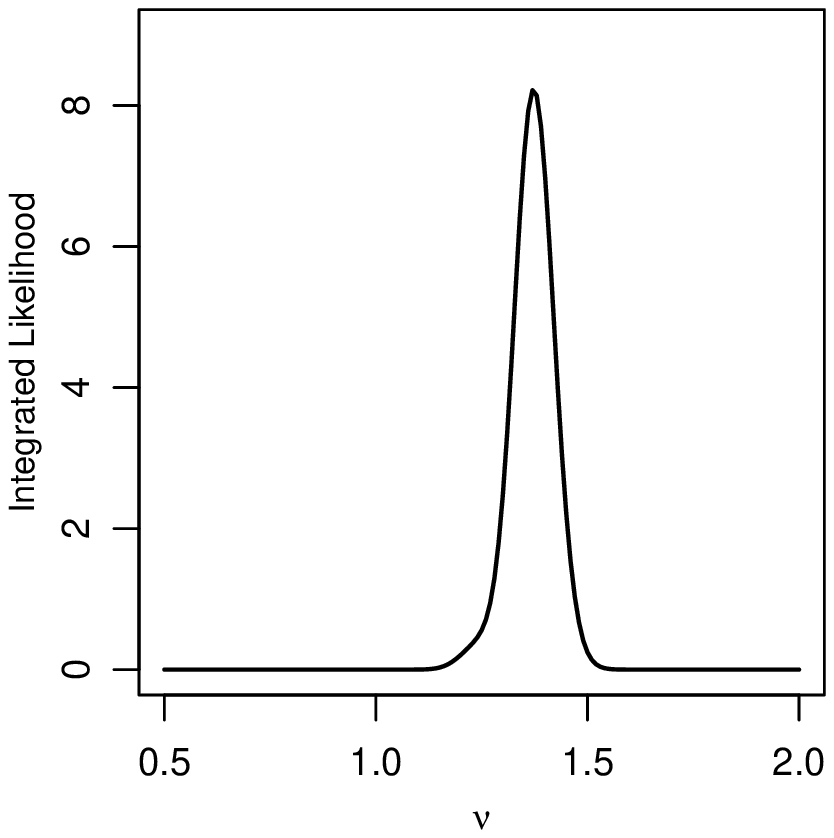, width=7cm,height=6cm} 
\psfig{figure=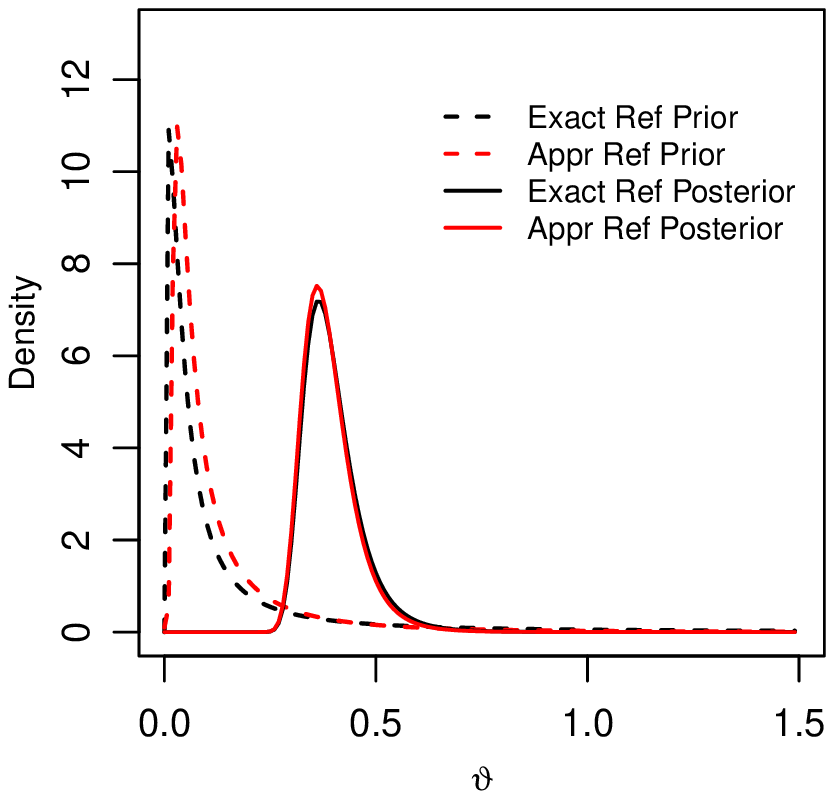, width=7cm,height=6cm}
\end{center}
\vspace{-0.5 cm}
\caption{Left: Empirical semivariogram of the simulated data set and  
its least squares fit. 
Right: Densities of exact and approximate marginal reference priors and 
posteriors of $\vartheta$ for the simulated data set.}
\label{fi:data2-fig2}
\end{figure}

Figure \ref{fi:data2-fig2} (right) displays the (normalized) exact and approximate 
reference priors of $\vartheta$, as well as their corresponding posteriors. 
The prior densities exhibit a slightly larger discrepancy than those for the lead concentration data, 
which is attributed to the sampling design being quite  irregular.
Nevertheless, the difference is not large, the shape of the two priors are similar, with both 
placing large and small probability masses in the same regions of the parameter space, and
the posterior densities are very close to each other.
Table \ref{ta:data2-tab} reports the Bayesian estimators of the model parameters and their 
corresponding $95\%$ highest posterior density (HPD) credible intervals based on both posteriors, 
showing that both inferences are, for practical purposes, equivalent. 

The computation time to draw $10^4$ posterior samples based on the 
approximate reference prior was about $2240$ seconds, while the time to do the same task 
based on exact reference prior was $5502$ seconds.
In both the exact likelihood was used so the time difference is due to prior evaluations.
Since the assumed smoothness is not equal to $m + 1/2$ for some non--negative integer $m$,
the evaluation of the covariance function and its derivative w.r.t. $\vartheta$ involved in 
the computation of the exact reference prior requires $O(n^2)$ evaluations of Bessel functions. 

\begin{table}[t]
\begin{center}
\caption{Parameter estimates of the 
 data with an irregular sampling design   
using exact and approximate reference priors. 
The estimate $\hat{\vartheta}$ is the posterior mode,  $\hat{\sigma}^2$ is the posterior median 
and $\hat{\beta}_1$ is the posterior mean. The 95\% credible intervals are the HPD.}
\label{ta:data2-tab}
\begin{tabular}{lccc}
\toprule
\multirow{2}{*}{Prior} & $\hat{\beta}_1$ & $\hat{\sigma}^2$ & $\hat{\vartheta}$ \\ 
 &  $(95\%~{\rm CI})$  & $(95\%~{\rm CI})$   & $(95\%~{\rm CI})$ \\
 \midrule
\multirow{2}{*}{Exact Reference} 
& $0.589$ & $1.299$ & $0.365$  \\ 
 & $(-0.410, 1.706)$ & $(0.547, 2.950)$ & $(0.286, 0.535)$ \\ \midrule
\multirow{2}{*}{Approximate Reference} & $0.582$   & $1.261$ &  $0.361$ \\ 
 & $(-0.394, 1.681)$ &  $(0.564, 2.863)$ & $(0.287, 0.522)$ \\ \bottomrule
\end{tabular}
\end{center}
\end{table}

\newpage

\noindent{\bf S7. Additional References}

\medskip

\noindent Chipman, J.S. (1964), 
On Least Squares with Insufficient Observations, 
{\it Journal of the American Statistical Association}, 59, 1078-1111.

\medskip
\noindent Dahlquist, G. and Bj\"{o}rck, {\AA} (2008), {\it Numerical Methods in Scientific Computing}, 
Volume I. SIAM.

\medskip
\noindent Ghosh, J.K. and Mukerjee, R. (1992), Non--Informative Priors. 
In: {\it Bayesian Statistics 4}, J.M. Bernardo, J.O. Berger, A.P. Dawid and A.F.M. Smith (eds.), 
Oxford University Press, pp 195-210.

\medskip
\noindent Mohammadi, M. (2016), 
On the Bounds for Diagonal and Off-diagonal Elements of the Hat Matrix in the Linear  
Regression Model, {\it REVSTAT - Statistical Journal}, 14, 75-87. 

\medskip
\noindent
Meeker, W.Q. and L.A. Escobar. (1995), Teaching About Approximate Confidence Regions Based
on Maximum Likelihood Estimation, {\it The American Statistician}, 49, 48-53.

\medskip
\noindent Northrop, P. (2020), rust: Ratio--of--Uniforms Simulation with Transformation, 
{\it R package version 1.3.10},  \url{https://github.com/paulnorthrop/rust}.

\medskip
\noindent Sedrakyan, N. (1997), About the Applications of One Useful Inequality. 
{\it Kvant Journal}, 97, 42-44. 

\medskip
\noindent Sun, X. (2006), Bayesian Spatial Data Analysis with Application to the Missouri Ozark Forest 
Ecosystem Project, Ph.D. Dissertation, University of Missouri.

\end{document}